\documentclass[conference]{IEEEtran}
\usepackage[ruled,vlined]{algorithm2e}
\usepackage{cite}
\usepackage{graphicx}
\usepackage{psfrag}
\usepackage{subfigure}
\usepackage{url}
\usepackage{amsmath}
\usepackage{yhmath}
\usepackage{array}
\usepackage{amssymb}
\usepackage{amsfonts}
\usepackage{graphicx}
\usepackage{multirow}

\newtheorem{lemma}{Lemma}

\newtheorem{definition}{Definition}

\newtheorem{example}{Example}
\title{Wireless Network-Coded Four-Way Relaying Using Latin Hyper-Cubes}
\begin{document}

\author{
\authorblockN{Srishti Shukla and B. Sundar Rajan}\\
\authorblockA{Email: {$\lbrace$srishti, bsrajan$\rbrace$} @ece.iisc.ernet.in\\
IISc Mathematics Initiative (IMI), Dept. of Mathematics and Dept. of Electrical Comm. Engg., IISc, Bangalore\\
}
}

\maketitle
%\pagestyle{plain}	
%%%%%%%%
\begin{abstract}
This paper deals with physical layer network-coding for the four-way wireless relaying scenario where four nodes A, B, C and D wish to communicate their messages to all the other nodes with the help of the relay node R. The scheme given in the paper is based on the denoise-and-forward scheme proposed first by Popovski et al. in \cite{PoY1}. Intending to minimize the number of channel uses, the protocol employs two phases: Multiple Access (MA) phase and Broadcast (BC) phase with each phase utilizing one channel use. This paper does the equivalent for the four-way relaying scenario as was done for the two-way relaying scenario by Koike-Akino et al. \cite{KPT}, and for three-way relaying scenario in \cite{SVR}. It is observed that adaptively changing the network coding map used at the relay according to the channel conditions greatly reduces the impact of multiple access interference which occurs at the relay during the MA phase. These network coding maps are so chosen so that they satisfy a requirement called \textit{exclusive law}. We show that when the four users transmit points from the same M-PSK constellation, every such network coding map that satisfies the exclusive law can be represented by a 4-fold Latin Hyper-Cube of side M. The network code map used by the relay for the BC phase is explicitly obtained and is aimed at reducing the effect of interference at the MA stage.
\end{abstract} 

\section{Background And Preliminaries}
Physical layer network coding for the two-way relay channel exploits the multiple access interference occurring at the relay so that the communication between the end nodes can be done using a two stage protocol. This two-stage protocol was first introduced in \cite{ZLL}, and \cite{KMT}, \cite{PoY} deal with the information theoretic studies for the scheme. In \cite{KPT}, modulation schemes to be used at the nodes for uncoded transmission for the two-way relaying were studied.

Work done for the relay channels with three or more user nodes is given in \cite{SVR,LiA,PiR,PaO,JKPL}. In \cite{LiA}, a two stage operation called joint network and superposition coding, which employs four channel uses, three for the MA phase and one for the BC phase, has been proposed. The three users transmit to the relay node one-by-one in the first phase, and the relay node makes two superimposed XOR-ed packets and transmits back to the users in the BC phase. The packet from the node with the worst channel gain is XOR-ed with the other two packets. It is claimed by the authors that this scheme can be extended to more than three users as well. Work by Pischella and Ruyet in \cite{PiR} proposes a lattice-based coding scheme combined with power control, composed of alternate MA and BC phases, consisting of four channel uses for three-way relaying. The relay receives an integer linear combination of the symbols transmitted by the user nodes. It is stated that the scheme can be extended to more number of users. These two works however deal with the information theoretic aspects of multi-way relaying.

In \cite{JKPL}, the authors Jeon et al. adopt an `opportunistic scheduling technique' for physical network coding where using a channel norm criterion and a minimum distance criterion, users in the MA as well as the BC phase are selected on the basis of instantaneous SNR. Their approach utilizes six channel uses in case of three-way relaying and it is mentioned that the approach can be extended to more number of users. A `Latin square-like condition' for the three-way relay channel network code is proposed in \cite{PaO}, and it is suggested that cell swapping techniques on these Latin Cubes can be employed to improve upon these network codes. The protocol employs five channel uses, and the network coding map discussed doesn't deal with the channel gains associated with the channels explicitly. Latin Cubes have been further explored as a tool to find the network coding map used by the relay, depending on the channel gain in \cite{SVR}. The throughput performance of the two stage protocol for three-way relaying given in \cite{SVR} is better than the throughput performance of the `opportunistic scheduling technique' given in \cite{JKPL} at high SNR, as can be observed from the plots given in \cite{SVR}. Our work in this paper further extends the approach used in \cite{SVR} to four-way relaying and employs two channel uses for the entire information exchange amongst the four users, which makes the throughput performance of the scheme better than the other existing schemes

\begin{figure}[tp]
\center
\includegraphics[height=35mm]{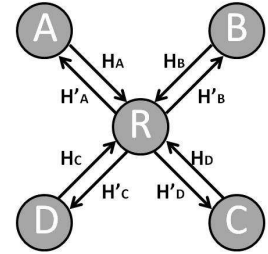}
\vspace{-.2 cm}
\caption{A four-way relay channel}
\vspace{-0.8 cm}
\end{figure}

We consider the four-way wireless relaying scenario shown in Fig. 1, where four-way data transfer takes place among the nodes A, B, C and D with the help of the relay R assuming that the four nodes operate in half-duplex mode. The relaying protocol consists of two phases, \textit{multiple access} (MA) phase, consisting of one channel use during which A, B, C and D transmit to R; and \textit{broadcast} (BC) phase, in which R transmits to A, B, C and D in a single channel use. Network Coding is employed at R in such a way that A(/B/C/D) can decode B's, C's and D's(/A's, C's and D's /A's, B's and C's) messages, given that A(/B/C/D) knows its own message. 

Our physical layer network coding strategy uses a mathematical structure called a Latin Hyper-Cube defined as follows:
\begin{definition}
An \textit{n-fold Latin Hyper-Cube L of r-th order of side M} \cite{Kis} is an $M \times M \times ... \times M ~(n~times) $ array containing $M^n$ entries, $M^{n-r}$ of each of $M^r$ kinds, such that each symbol occurs at most once for each value taken by each dimension of the hyper-cube. \footnote{The definition has been modified slightly from the referred article ``On Latin and Hyper-Graeco-Latin Cubes and Hyper Cubes'' by K. Kishen (Current Science, Vol. 11, pp. 98--99, 1942), in accordance with the context.}
\end{definition}

For our purposes, we use only 4-fold Latin Hyper-Cubes of side M on the symbols from the set $\mathbb{Z}_{t}=\left\{0,1,2,...,t-1\right\}$, $t \geq M^3$. In this hyper-cube, fixing the first dimension, that represents A's transmitted symbol, we get a three dimensional array, which we will be referring to as cube;  then B(/C/D)'s transmitted symbols are along the files(/rows/columns) of each cube. 

\section{Signal Model}
\noindent \textbf{\textit{Multiple Access (MA) Phase:}}\\
\indent Suppose A(/B/C/D) wants to send a 2-bit binary tuple to B, C and D(/A, C and D/A, B and D/A, B and C). The symmetric 4-PSK constellation $\left\{\pm 1,~\pm j\right\}$ denoted by $\mathcal{S}$ as shown in Fig. 2, is used at A, B, C and D, and $ \mu : \mathbb{F}^{2}_{2} \rightarrow \mathcal{S} $ denotes the map from bits to complex symbols used at A, B, C and D where $\mathbb{F}_{2}=\left\{0,1\right\}$. Let $ x_{A}=\mu\left(s_{A}\right), x_{B}=\mu\left(s_{B}\right), x_{C}=\mu\left(s_{C}\right), x_{D}=\mu\left(s_{D}\right) \in \mathcal{S}$ denote the complex symbols transmitted by A, B, C and D respectively, where $s_{A}, s_{B}, s_{C}, s_D \in \mathbb{F}^{2}_{2}$. We assume that the Channel State Information (CSI) is not available at the transmitting nodes and perfect CSI is available at the receiving nodes. The received signal at R in the MA phase is given by,
{\vspace{-.1 cm}
\begin{equation}
\label{yr}
%\hspace{-1.8 cm}
Y_{R}=H_{A}x_{A}+H_{B}x_{B}+H_{C}x_{C}+H_{D}x_{D}+Z_{R},
\vspace{-0.1 cm}\end{equation}}where $H_{A}$, $H_{B}$, $H_C$ and $H_{D}$ are the fading coefficients associated with the A-R, B-R, C-R and D-R link respectively. The additive noise $Z_{R} \sim \mathcal{CN}\left(0,\sigma^2 \right)$, where $\mathcal{CN}\left(0,\sigma^2 \right)$ denotes the circularly symmetric complex Gaussian random variable with variance $\sigma^2$.

\begin{figure}[tp]
\center
\includegraphics[height=35mm]{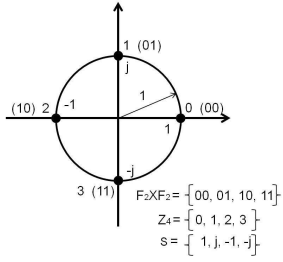}
\vspace{-.2 cm}
\caption{4-PSK constellation}
\vspace{-.5 cm}
\end{figure}
 
Let the effective constellation seen at the relay during the MA phase channel use be denoted by $ \mathcal{S}_{R} \left( H_A, H_B, H_C, H_D \right)$, i.e., 

{\footnotesize $\mathcal{S}_{R} \left( H_A, H_B, H_C, H_D \right) = \left\{H_A x_{A} + H_B x_{B} + H_C x_{C}+ H_D x_{D}| \right. $ \\
$ \left. ~~~~~~~~~~~~~~~~~~~~~~~~~~~~~~~~~~~~~~~~~~~~~~~~~~~~~~~~~~~ x_{A}, x_{B}, x_{C}, x_D \in \mathcal{S}\right\}.$}

The minimum distance between the points in the constellation $ \mathcal{S}_{R} \left( H_A, H_B, H_C, H_D \right) $ denoted by $d_{min}\left(H_A, H_B, H_C, H_D\right)$ is given in (\ref{dist}) on the next page, where $ \mathcal{S}^{n}=\mathcal{S} \times \mathcal{S} \times .. \times \mathcal{S} \ (n~times)$. From (\ref{dist}), it is clear that there exists values of $(H_A, H_B, H_C, H_D)$, for which $d_{min}\left(H_A, H_B, H_C, H_D\right)=0$. Let $\mathcal{H}=\left\{ (H_A, H_B, H_C, H_D) \in \mathbb{C}^4 | d_{min}\left(H_A, H_B, H_C, H_D\right)=0 \right\}$. The elements of $\mathcal{H}$ are called singular fade states. For singular fade states, $\left|\mathcal{S}_{R} \left( H_A, H_B, H_C, H_D \right)\right| < 4^{4}$.

\begin{definition}
A fade state $(H_A, H_B, H_C, H_D)$ is defined to be a \textit{singular fade state} for the MA phase of four-way relaying, if the cardinality of the signal set $ \mathcal{S}_{R} \left( H_A, H_B, H_C, H_D \right)$ is less than $4^{4}$. 
\end{definition}

Let the Maximum Likelihood (ML) estimate of $\left(x_{A}, x_{B}, x_{C}, x_D\right) $ be denoted by $\left(\hat{x}_{A}, \hat{x}_{B}, \hat{x}_{C}, \hat{x}_D\right) \in \mathcal{S}^{4}$ at R based on the received complex number $Y_{R}$, i.e., \vspace{-.4 cm}

{\footnotesize
\begin{equation}
\vspace{-.3 cm}
\left(\hat{x}_{A}, \hat{x}_{B},\hat{x}_{C}, \hat{x}_D\right)=\arg \min_{\left({x_{A}}, {x_{B}}, {x_{C}}, x_D\right) \in \mathcal{S}^{4}}\left\|Y_R - HX\right\|, 
\end{equation}}where $ H=\left[H_{A}\  H_{B}\  H_{C} \ H_D \right]$  and $ X=\left[ x_{A}\ x_{B}\ x_{C}\ x_D \right]^T.$

\begin{figure*}
\footnotesize
\begin{align}
&
\label{dist}
d_{min}(H_A, H_B, H_C, H_D)=\hspace{-0.5 cm}\min_{\substack {{(x_A,x_B,x_C,x_D),(x'_A,x'_B,x'_C,x'_D) \in \mathcal{S}^{4}} \\ {(x_A,x_B,x_C,x_D) \neq (x'_A,x'_B,x'_C,x'_D)}}}\hspace{-0.5 cm}\vert H_A \left(x_A-x'_A\right)+H_B \left(x_B-x'_B\right) + H_C \left(x_C-x'_C\right) + H_D \left(x_D-x'_D\right)\vert \\
\hline
\vspace{1cm}
&
\label{cl1}
d_{min}^{\mathcal{L}_{i},\mathcal{L}_{j}}\left(H_A, H_B, H_C, H_D\right)=\hspace{-0.2 cm}\min_{\substack {{(x_A,x_B,x_C,x_D) \in \mathcal{L}_{i}},\\ (x'_A,x'_B,x'_C,x'_D) \in \mathcal{L}_{j}}} \hspace{-0.2 cm}  \left| H_A \left( x_A-x'_A\right)+ H_B \left(x_B-x'_B\right) + H_C \left(x_C-x'_C\right) + H_D \left(x_D-x'_D\right) \right| \\
\hline
\vspace{1cm}
&
\label{cl2}
d_{min} \left(\mathcal{C}^{H_A, H_B, H_C, H_D}\right)=\hspace{-1.7 cm}\min_{\substack {{(x_A,x_B,x_C,x_D),(x'_A,x'_B,x'_C, x'_D) \in \mathcal{S}^{4},} \\ {\mathcal{M}^{H_A, H_B, H_C, H_D}(x_A,x_B,x_C,x_D) \neq \mathcal{M}^{H_A, H_B, H_C,H_D}(x'_A,x'_B,x'_C,x'_D)}}}\hspace{-1.7 cm} \left| H_A \left( x_A-x'_A\right)+ H_B \left(x_B-x'_B\right) + H_C \left(x_C-x'_C\right) + H_D \left(x_D-x'_D\right) \right| \\
\hline
\vspace{0.cm}
&
\label{cl3}
d_{min} \left(\mathcal{C}^{\left\{\left(H_A, H_B, H_C,H_D\right)\right\}} , h_A, h_B, h_C, h_D \right)=\hspace{-2.9 cm}\min_{\substack {\vspace{0.15cm} {(x_A,x_B,x_C,x_D),(x'_A,x'_B,x'_C,x'_D) \in \mathcal{S}^{4},} \\ {\mathcal{M}^{H_A,H_B,H_C,H_D}(x_A,x_B,x_C,x_D) \neq \mathcal{M}^{H_A, H_B, H_C,H_D}(x'_A,x'_B,x'_C,x'_D)}}}\hspace{-2.6 cm}  \left| h_A \left( x_A-x'_A\right)+ h_B \left(x_B-x'_B\right) + h_C \left(x_C-x'_C\right) + h_D \left(x_D-x'_D\right)\right|\\
\hline
\vspace{0.cm}
&
\label{mel1}
\mathcal{M}^{H_A,H_B,H_C,H_D}\left(x_{A},x_{B},x_{C},x_D\right) \neq \mathcal{M}^{H_A,H_B,H_C,H_D}\left(x_{A},x'_{B},x'_{C},x'_D\right), \forall x_{A}, x_{B}, x'_{B}, x_{C}, x'_{C}, x_{D}, x'_{D}  \in \mathcal{S}, \left(x_{B},x_{C},x_D\right) \neq \left(x'_{B},x'_{C},x'_D\right)\\
\vspace{0. cm}
&
\label{mel2}
\mathcal{M}^{H_A,H_B,H_C,H_D}\left(x_{A},x_{B},x_{C},x_D\right) \neq \mathcal{M}^{H_A,H_B,H_C,H_D}\left(x'_{A},x_{B},x'_{C},x'_D\right), \forall x_{A}, x'_{A}, x_{B}, x_{C}, x'_{C}, x_{D}, x'_{D}  \in \mathcal{S}, \left(x_{A},x_{C},x_D\right) \neq \left(x'_{A},x'_{C},x'_D\right)\\
\vspace{0. cm}
& 
\label{mel3}
\mathcal{M}^{H_A,H_B,H_C,H_D}\left(x_{A},x_{B},x_{C},x_D\right) \neq \mathcal{M}^{H_A,H_B,H_C,H_D}\left(x'_{A},x'_{B},x_{C},x'_D\right), \forall x_{A}, x'_{A}, x_{B}, x'_{B}, x_{C}, x_{D}, x'_{D}  \in \mathcal{S}, \left(x_{A},x_{B},x_D\right) \neq \left(x'_{A},x'_{B},x'_D\right)\\
\vspace{0.cm}
& 
\label{mel4}
\mathcal{M}^{H_A,H_B,H_C,H_D}\left(x_{A},x_{B},x_{C},x_D\right) \neq \mathcal{M}^{H_A,H_B,H_C,H_D}\left(x'_{A},x'_{B},x'_{C},x_D\right), \forall x_{A}, x'_{A}, x_{B}, x'_{B}, x_{C}, x_{D}, x'_{D}  \in \mathcal{S}, \left(x_{A},x_{B},x_C\right) \neq \left(x'_{A},x'_{B},x'_C\right)\\
\hline
\nonumber
\end{align}
%\vspace{-1.5cm}
\end{figure*}\hspace{-.1 cm}

\noindent \textbf{\textit{Broadcast (BC) Phase:}}\\
\indent During the BC phase, the received signals at A, B, C and D are respectively given by, \vspace{-.3 cm}

{\tiny \begin{equation}
\vspace{-.4 cm} Y_{A}=H'_{A}X_{R}+Z_{A},\ Y_{B}=H'_{B}X_{R}+Z_{B},\ Y_{C}=H'_{C}X_{R}+Z_{C}, Y_D=H'_{D}X_{R} +Z_D
\end{equation}}

\noindent where $X_{R}=\mathcal{M}^{H_A, H_B, H_C, H_D}\left(\left(\hat{x}_{A},\hat{x}_{B},\hat{x}_{C},\hat{x}_{D}\right)\right) \in \mathcal{S}^{'}$ denotes the complex number transmitted by R and $H_{A}^{'},$ $H_{B}^{'},$ $H_{C}^{'},$ and $H_{D}^{'}$ respectively are the fading coefficients corresponding to the links R-A, R-B, R-C and R-D. The additive noises $Z_{A},$ $Z_{B},$ $Z_{C}$ and $Z_{D}$ are  $\mathcal{CN}\left(0,\sigma^{2}\right)$. During \textit{BC} phase, R transmits a point from a signal set $\mathcal{S}^{'}$ given by a many to one map $\mathcal{M}^{H_A, H_B, H_C, H_D} : \mathcal{S}^4 \rightarrow \mathcal{S}^{'} $ chosen by R, depending on the values of $H_A$, $H_B$, $H_C$ and $H_D$. The cardinality of $\mathcal{S}^{'} \geq 64=2^6$, since 6 bits about the other three users needs to be conveyed to each of A, B, C and D.

A \textit{cluster} is the set of elements in $\mathcal{S}^4 $ which are mapped to the same signal point in $\mathcal{S}^{'}$ by the map $\mathcal{M}^{H_A, H_B, H_C, H_D}$. Let $\mathcal{C}^{H_A, H_B, H_C, H_D}=\left\{\mathcal{L}_{1}, \mathcal{L}_{2},.., \mathcal{L}_{l}\right\}$ denote the set of all such clusters.

\begin{definition}
The \textit{cluster distance} between clusters $\mathcal{L}_i, ~\mathcal{L}_j \in \mathcal{C}^{H_A, H_B, H_C, H_D}$, as given in (\ref{cl1}) on the next page, is the minimum among all the distances calculated between the points $\left(x_{A}, x_{B}, x_{C}, x_D\right) \in \mathcal{L}_{i}$ and $\left(\acute{x_{A}}, \acute{x_{B}}, \acute{x_{C}}, \acute{x_{D}}\right) \in \mathcal{L}_{j}$ in the effective constellation seen at the relay node R. The minimum among all the cluster distances among all pairs of the clustering $\mathcal{C}^{H_A, H_B, H_C, H_D}$ is its \textit{minimum cluster distance}, as given in (\ref{cl2}) on the next page.
\end{definition}

The performance during the MA phase depends on the minimum cluster distance, while the performance during the BC phase is dependent on the minimum distance of the signal set $\mathcal{S}^{'}$. A phenomenon referred to as \textit{distance shortening}, given in [4], is described as the significant reduction in the value of $d_{min}\left(\mathcal{C}^{H_A, H_B, H_C, H_D}\right)$ for values of $\left(H_A, H_B, H_C, H_D\right)$ in the neighborhood of the singular fade states. If the clustering used at the relay node R in the BC phase is chosen such that the minimum cluster distance at the singular fade state is non zero and is also maximized, then the effect of distance shortening can be avoided.

We say a clustering $\mathcal{C}^{H_A, H_B, H_C, H_D}$ removes a singular fade state $\left(H_A, H_B, H_C, H_D\right) \in \mathcal{H}, $ if $d_{min}\left(\mathcal{C}^{H_A, H_B, H_C, H_D}\right)>0$. Let $\mathcal{C}^{\left\{\left(H_A, H_B, H_C, H_D\right)\right\}} $ denote the clustering which removes the singular fade state $\left(H_A, H_B, H_C, H_D\right)$ (selecting one randomly if there are multiple clusterings which remove the same singular fade state $\left(H_A, H_B, H_C, H_D\right)$). Let the set of all such clusterings be denoted by $\mathcal{C_{H}}$, i.e., $\mathcal{C_{H}}=\left\{\mathcal{C}^{\left\{\left(H_A, H_B, H_C, H_D\right)\right\}} : \left(H_A, H_B, H_C, H_D\right)\in \mathcal{H}\right\} $.

\begin{definition}
The minimum cluster distance of the clustering $\mathcal{C}^{\left\{\left(H_A, H_B, H_C, H_D\right)\right\}}$ for $\left(H_A,H_B,H_C,H_D\right) \in \mathcal{H}$, when the fade state $(h_A, h_B, h_C,h_D)$ occurs in the MA phase, denoted by $d_{min}\left(\mathcal{C}^{\left\{\left(H_A, H_B, H_C, H_D\right)\right\}},h_A, h_B, h_C, h_D\right)$, is the minimum among all its cluster distances.
\end{definition}

When $\left(h_A, h_B, h_C, h_D\right) \notin \mathcal{H}, $ we choose the clustering $\mathcal{C}^{h_A, h_B, h_C, h_D} $ to be $\mathcal{C}^{\left\{\left(H_A, H_B, H_C, H_D\right)\right\}} \in \mathcal{C}_{\mathcal{H}}$, such that {\footnotesize $ d_{min}\left(\mathcal{C}^{\left\{\left(H_A, H_B, H_C, H_D\right)\right\}},h_A, h_B, h_C, h_D\right) \geq d_{min}\left(\mathcal{C}^{\left\{\left(H'_A, H'_B, H'_C, H'_D\right)\right\}},h_A, h_B, h_C, h_D\right), \forall \left(H_{A}, H_{B}, H_C, H_D\right) \neq \left(H'_{A}, H'_{B}, H'_{C}, H'_D\right) \in \mathcal{H}$}.
The clustering used by the relay is indicated to A, B, C and D using overhead bits.

In order to ensure that A(/B/C/D) is able to decode B's, C's and D's(/A's, C's and D's /A's, B's and D's /A's, B's and C's) message, the clustering $\mathcal{C}$ should satisfy the exclusive law, as given in (\ref{mel1}), (\ref{mel2}), (\ref{mel3}), (\ref{mel4}) on the next page. We explain Exclusive Law in more detail in the next section.

The contributions of this paper are as follows: 
\begin{itemize}
\item Exchange of information in the wireless four-way relaying scenario is made possible with totally two channel uses using our proposed scheme.
\item It is shown that if the four users A, B, C, D transmit points from the same M-PSK constellation, the requirement of satisfying the exclusive law is same as the clustering being represented by a 4-fold Latin Hyper-Cube of side M. (Section III)
\item The singular fade states for the four-way relaying scenario are identified. (Section IV)
\item Clusterings that remove these singular fade states are obtained and are of size between 64 to 90. (Section V)
\item Simulation results are provided to verify that the adaptive clustering as obtained in the paper indeed performs better than non-adaptive clustering at high SNR. (Section VI) 
\end{itemize}

The remaining content is organized as follows: Section III demonstrates how a 4-fold Latin Hyper-Cube of side 4 can be utilized to represent the network code for four user communication. In Section IV we define what we mean by singular fade subspaces and in Section V, focus in on the removal of such singular fade subspaces using 4-fold Latin Hyper-Cube of side 4. Simulation results are provided in Section VI. Section VII concludes the paper. 

%%%%%%%%%%%%%%%%%%%%%%%%%%%%%
\section{The Exclusive Law and Latin Hyper-Cubes}
%%%%%%%%%%%%%%%%%%%%%%%%%%%%%%%

In the previous section, it was stated that in order to ensure that A(/B/C/D) is able to decode B's, C's and D's(/A's, C's and D's /A's, B's and D's /A's, B's and C's) message, the clustering $\mathcal{C}$ that represents the map used at the relay should satisfy the exclusive law. We assume that the nodes A, B, C and D transmit symbols from the 4-PSK constellation. Consider a $4 \times 4 \times 4 \times 4 $ array, containing $4^4=256$ entries indexed by $\left(x_{A}, x_{B}, x_{C}, x_D\right)$, i.e., the four symbols sent by A, B, C and D in the MA phase. The four cubes of this $4 \times 4 \times 4 \times 4$ array, are indexed by four values taken by $x_{A}$. Each file (/row /column) of each such cube is indexed by a value of $x_B$ (/$x_C$ /$x_D$), for this fixed value of $x_A$. The repetition of a symbol in a cube results in the failure of exclusive law given by (\ref{mel1}). Considering the $4 \times 4 \times 4$ array with its files being the first(/second/third/forth) files of the $4 \times 4 \times 4 \times 4$ array. Each $4 \times 4 \times 4$ array so obtained, corresponds to a single value of $x_{B}$. A repetition of a symbol in this array will result in the failure of exclusive law given by (\ref{mel2}). Similarly, a repetition of a symbol in the $4 \times 4\times 4$ array with its rows(columns) being the first(/second/third/forth) columns of the $4 \times 4 \times 4 \times 4 $ array, that corresponds to a single value of $x_{C} (x_D)$, will result in the failure of exclusive law given by (\ref{mel3})((\ref{mel4})). Thus, for the exclusive law to be satisfied, the cells of this array should be filled such that the $4 \times 4 \times 4 \times 4$ array so obtained, is a 4-fold Latin Hyper-Cube of side 4, for $t \geq 64$ (Definition 1). The clusters are obtained by putting together all the tuples $\left(i,j,k,l\right), i,j,k,l \in 0,1,...t-1 $ such that the entry in the $\left(i,j,k,l\right)$-th slot is the same entry from $\mathbb{Z}_{t}$.

\begin{figure}[tp]
\center
\vspace{.2cm}
\includegraphics[height=30mm]{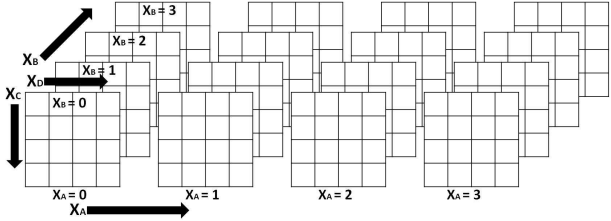}
\vspace{-.4cm}
\caption{A 4-fold Latin Hyper-Cube of side 4 represents the exclusive law constraint for the relay map when 4-PSK is used at end nodes}
\vspace{-.5cm}
\end{figure}

The same can be extended to an M-PSK constellation, i.e., if the network code used by the relay node in the BC phase satisfies exclusive law, it can be represented by a 4-fold Latin Hyper-Cube of side M and $t \geq M^3$ with cubes(files/rows/columns) being indexed by the constellation point used by A(/B/C/D), symbols from the set $ \mathbb{Z}_{M}$ (Fig. 3). Any arbitrary but unique symbol from $\left\{\mathcal{L}_1,..,\mathcal{L}_t\right\}$ denotes a unique cluster of a particular clustering.

%%%%%%%%%%%%%%%%%%
\section{Singular Fade Subspaces}

As stated in Section II, a clustering $\mathcal{C}^{H_A, H_B, H_C, H_D}$ is said to remove singular fade state $\left(H_A, H_B, H_C, H_D\right) \in \mathcal{H}, $ if $d_{min}\left(\mathcal{C}^{H_A, H_B, H_C, H_D}\right)>0$, i.e., any two message sequences $\left(x_{A},x_{B},x_{C},x_D\right) \in \mathcal{S}^4$ that coincide in the effective constellation received at the relay during the MA phase should be kept in the same cluster of $\mathcal{C}^{H_A, H_B, H_C, H_D}$. So, removing singular fade states for a four-way relay channel can alternatively be defined as:
\begin{definition}
A clustering $\mathcal{C}^{H_A, H_B, H_C, H_D}$ is said to \textit{remove the singular fade state} $\left(H_A, H_B, H_C, H_D\right) \in \mathcal{H}$, if any two possibilities of the messages sent by the users $\left(x_{A},x_{B},x_{C},x_D\right), \left(x'_{A},x'_{B},x'_{C},x'_D\right) \in \mathcal{S}^{4}$ that satisfy

{\footnotesize  $$ H_A x_A+ H_Bx_B+ H_C x_C+H_D x_D=H_A x'_A+ H_B x'_B+ H_C x'_C +H_D x'_D$$} are placed together in the same cluster by the clustering.
\end{definition}

\begin{definition}
A set $\left\{(x_A, x_B, x_C,x_D)\right\} \in \mathcal{S}^4 $ consisting of all the possibilities of $(x_A, x_B, x_C, x_D)$ that must be placed in the same cluster of the clustering used at relay node R in the BC phase in order to remove the singular fade state $\left(H_A, H_B, H_C, H_D\right)$ is referred to as a \textit{Singularity Removal Constraint} for the fade state $\left(H_A, H_B, H_C, H_D\right)$ for four-way relaying scenario.
\end{definition}

The relay receives a complex number, given by (\ref{yr}), at the end of MA phase. Using the ML estimate of this received complex number, R transmits a point from the constellation $\mathcal{S}'$ with cardinality at most $4^{4}$. Instead of R transmitting a point from the $4^{4}$ point constellation resulting from all the possibilities of $\left(x_{A}, x_{B}, x_{C}, x_D\right)$, depending on the fade states, the relay R can choose to group these possibilities into clusters represented by a smaller constellation, so that the minimum cluster distance is minimized, as well as all the users receive the messages from the other three users, i.e., the clustering satisfies the exclusive law. We provide one such clustering for the case of four-way relaying in the following. 

Let $(H_A, H_B ,H_C, H_D)$ be the fade coefficient in the MA phase. Suppose $(H_A, H_B, H_C, H_D) $ is a singular fade state, and ${\Gamma}$ is a singularity removal constraint corresponding to the singular fade state $(H_A, H_B, H_C, H_D) $. Then there exist $(x_A, x_B, x_C, x_D), (x'_A, x'_B, x'_C,x'_D) \in \Gamma $, $(x_A, x_B, x_C, x_D) \neq (x'_A, x'_B, x'_C, x'_D) $ such that, 

{\tiny
\vspace{-.2 cm}
\begin{align}
\nonumber
& H_A x_A+ H_Bx_B+ H_C x_C+ H_D x_D=H_A x'_A+ H_B x'_B+ H_C x'_C + H_D x'_D \\
\nonumber
\Rightarrow & H_A (x_A-x'_A)+ H_B (x_B-x'_B)+ H_C (x_C-x'_C)+ H_D (x_D-x'_D)=0 \\
\Rightarrow & (H_A, H_B, H_C, H_D) \in 
\left\langle  \left[ {\begin{array}{cc}
x_{A}-x'_A \\
x_{B}-x'_B \\
x_{C}-x'_C \\
x_{D}-x'_D \\
\end{array} } \right] \right\rangle ^{\bot} 
%\vspace{-.5 cm}
\label{nullspace} 
\vspace{-.2 cm}
\end{align}}where for a $4 \times 1$ non-zero vector $v$ over $\mathbb{C}$,

{
\scriptsize
\vspace{-.2 cm}
\begin{equation}
\left\langle v\right\rangle ^{\bot}=\left\{w=(w_1,w_2,w_3,w_4) \in \mathbb{C}^4 ~|~ w_1v_1+w_2v_2+w_{3} v_{3}+w_4 v_4=0 \right\}.
\vspace{-.1 cm}
\end{equation}}
It can be proven that $\left\langle v\right\rangle ^{\bot}$ is a three-dimensional vector subspace of $\mathbb{C}^4$. Since $x_A, x_B, x_C, x_D, x'_A, x'_B, x'_C, x'_D \in \mathcal{S}$, where $\mathcal{S}$ is finite, there are only finitely many possibilities for the right-hand side of (\ref{nullspace}). Thus the uncountably infinite singular fade states $(H_A, H_B, H_C, H_D)$, are points in a finite number of vector subspaces of $\mathbb{C}^4$. We shall refer to these finite number of vector subspaces as the \textit{Singular Fade Subspaces} \cite{SVR}. There are four possibilities of singular fade subspaces for four-way relaying as we explain individually in the following. 

\noindent \textbf{\textit{Case 1:}} Only one of $x_A-x'_A$, $x_B-x'_B$, $x_C-x'_C$ or $x_D-x'_D$ is non-zero, resulting in the following 4 subcases:
\begin{enumerate}
	\item $x_B=x'_B,~ x_C=x'_C,~ x_D=x'_D \text{~and~} x_A \neq x'_A$
	\item $x_A=x'_A,~ x_C=x'_C,~ x_D=x'_D \text{~and~} x_B \neq x'_B$
	\item $x_A=x'_A,~ x_B=x'_B,~ x_D=x'_D \text{~and~} x_C \neq x'_C$
	\item $x_A=x'_A,~ x_B=x'_B,~ x_C=x'_C \text{~and~} x_D \neq x'_D$
\end{enumerate}
\textbf{\textit{Case 2:}} Two of $x_A-x'_A$, $x_B-x'_B$, $x_C-x'_C$ or $x_D-x'_D$ are non-zero, resulting in the following 6 subcases:
\begin{enumerate}
	\item $x_A=x'_A,~ x_B = x'_B, ~x_C \neq x'_C\text{~and~} x_D \neq x'_D$
	\item $x_A=x'_A,~ x_C = x'_C, ~x_B \neq x'_B\text{~and~} x_D \neq x'_D$
	\item $x_A=x'_A,~ x_D = x'_D, ~x_B \neq x'_B\text{~and~} x_C \neq x'_C$
	\item $x_B=x'_B,~ x_C = x'_C, ~x_A \neq x'_A\text{~and~} x_D \neq x'_D$
	\item $x_B=x'_B,~ x_D = x'_D, ~x_A \neq x'_A\text{~and~} x_C \neq x'_C$
	\item $x_C=x'_C,~ x_D = x'_D, ~x_A \neq x'_A\text{~and~} x_B \neq x'_B$
\end{enumerate}
\textbf{\textit{Case 3:}} Three of $x_A-x'_A$, $x_B-x'_B$, $x_C-x'_C$ or $x_D-x'_D$ are non-zero, resulting in the following 4 subcases:
\begin{enumerate}
	\item $x_A=x'_A,~ x_B \neq x'_B, ~x_C \neq x'_C\text{~and~} x_D \neq x'_D$
	\item $x_A\neq x'_A,~ x_B = x'_B, ~x_C \neq x'_C\text{~and~} x_D \neq x'_D$
	\item $x_A\neq x'_A,~ x_B \neq x'_B, ~x_C =x'_C\text{~and~} x_D \neq x'_D$
	\item $x_A \neq x'_A,~ x_B \neq  x'_B, ~x_C \neq x'_C\text{~and~} x_D = x'_D$
\end{enumerate}
\textbf{\textit{Case 4:}} All of $x_A-x'_A$, $x_B-x'_B$, $x_C-x'_C$ or $x_D-x'_D$ are non-zero, i.e., $x_A \neq x'_A,~ x_B \neq x'_B,~ x_C \neq x'_C \text{~and~} x_D \neq x'_D$.

\vspace{.5cm}
\noindent \textbf{\textit{Case 1:}} Only one of $x_A-x'_A$, $x_B-x'_B$, $x_C-x'_C$ or $x_D-x'_D$ is non-zero, resulting in 4 subcases.
Without loss of generality, we discuss the first subcase of \textit{Case 1}. The singular fade subspace in this case will be given by 
{\footnotesize \vspace{-0 cm}$\mathcal{S}'=\left\langle \left[ {\begin{array}{cc}
x_{A}-x'_A \\
0 \\
0 \\
0 \\
\end{array} } \right]\right\rangle ^{\bot}.$} Let $\mathcal{D}$ denote the set of differences of the points of $\mathcal{S}$, i.e., $\mathcal{D}=\left\{x_i-x_j ~|~ x_i, x_j \in \mathcal{S} \right\} = \left\{0, \pm 1\pm j, \pm2j, \pm 2\right\}. $ We can write, $\mathcal{D}=\left\{0\right\} \cup \mathcal{D}_{1} \cup \mathcal{D}_{2},$ where $\mathcal{D}_{1}=\left\{ \pm1\pm j\right\}$ and $\mathcal{D}_{2}=\left\{ \pm2j, \pm2\right\}$. 

Now, $x_A- x'_A \in \mathcal{D}$ can take eight non-zero values. So there are eight total possibilities for the vector $ \left[x_A-x'_A, ~0, ~0, ~0\right]^t $, where $v^t$ denotes the transpose of a vector $v$. Since each one of $\pm 1\pm j$, $ \pm2j $ and $ \pm2$ can be obtained as scalar multiples of $1+j$ (over $\mathbb{C}$), we have,\\

{\footnotesize
\vspace{-.7cm}
\begin{align}
\nonumber
(-1) \left[ 1+j, ~0, ~0, ~0 \right]^T&=\left[ -1-j, ~ 0, ~ 0, ~ 0 \right]^T\\
\nonumber
(j) \left[ 1+j, ~ 0, ~ 0, ~ 0 \right]^T&=\left[ -1+j, ~ 0, ~ 0, ~ 0 \right]^T\\
\nonumber
(-j) \left[ 1+j, ~ 0, ~ 0, ~ 0 \right]^T&=\left[ 1-j, ~ 0, ~ 0, ~ 0 \right]^T\\
\nonumber
(1+j) \left[ 1+j, ~ 0, ~ 0, ~ 0 \right]^T&=\left[ 2j, ~ 0, ~ 0, ~ 0 \right]^T\\
\nonumber
(-1-j) \left[ 1+j, ~ 0, ~ 0, ~ 0 \right]^T&=\left[ -2j, ~ 0, ~ 0, ~ 0 \right]^T\\
\nonumber
(1-j) \left[ 1+j, ~ 0, ~ 0, ~ 0 \right]^T&=\left[ 2, ~ 0, ~ 0, ~ 0 \right]^T\\
\nonumber
(-1+j) \left[ 1+j, ~ 0, ~ 0, ~ 0 \right]^T&=\left[ -2, ~ 0, ~ 0, ~ 0 \right]^T
\nonumber
\end{align}}Therefore,\\

\vspace{-.2cm}
\begin{tiny}$~~~~~~~~~~~~~~~~
\left\langle \left[ {\begin{array}{cc}
1+j \\
0 \\
0 \\
0\\
\end{array} } \right]\right\rangle
\hspace{-.05 cm}
=\hspace{-.05 cm}
\left\langle \left[ {\begin{array}{cc}
\pm 1\pm j \\
0 \\
0 \\
0\\
\end{array} } \right]\right\rangle
\hspace{-.05 cm}
=\hspace{-.05 cm}
\left\langle \left[ {\begin{array}{cc}
\pm 2j \\
0 \\
0 \\
0\\
\end{array} } \right]\right\rangle 
\hspace{-.05 cm}
=\hspace{-.05 cm}
\left\langle \left[ {\begin{array}{cc}
\pm 2 \\
0 \\
0 \\
0\\
\end{array} } \right]\right\rangle.$\end{tiny}
\vspace{.2cm}

As a result, only one singular fade subspace arises from this subcase, viz., {\footnotesize $\mathcal{S}'=\left\langle \left[ 1+j, ~ 0, ~ 0, ~ 0 \right]^T\right\rangle ^{\bot}. $} 
Similarly, for the other three subcases, there is a single singular fade subspace, resulting in a total of 4 singular fade subspaces arising from this case.
%%%%%%%%%%%%%%%%%%%%%%%%%%%%%%%%%%%%5

\begin{figure*}
\centering
\tiny
$1. \left\langle \left[ {\begin{array}{cc}
1+j \\
1+j \\
0 \\
0 \\
\end{array} } \right]\right\rangle
=\left\langle \left[ {\begin{array}{cc}
-1-j \\
-1-j \\
0 \\
0 \\
\end{array} } \right]\right\rangle
=\left\langle \left[ {\begin{array}{cc}
1-j \\
1-j \\
0 \\
0 \\
\end{array} } \right]\right\rangle 
=\left\langle \left[ {\begin{array}{cc}
-1+j \\
-1+j \\
0 \\
0 \\
\end{array} } \right]\right\rangle
=\left\langle \left[ {\begin{array}{cc}
2j \\
2j \\
0 \\
0 \\
\end{array} } \right]\right\rangle
=\left\langle \left[ {\begin{array}{cc}
-2j \\
-2j \\
0 \\
0 \\
\end{array} } \right]\right\rangle
=\left\langle \left[ {\begin{array}{cc}
2 \\
2 \\
0 \\
0 \\
\end{array} } \right]\right\rangle 
=\left\langle \left[ {\begin{array}{cc}
-2 \\
-2 \\
0 \\
0 \\
\end{array} } \right]\right\rangle $~~~~~~~~~~~~\\
$2. \left\langle \left[ {\begin{array}{cc}
1+j \\
-1-j \\
0 \\
0 \\
\end{array} } \right]\right\rangle
=\left\langle \left[ {\begin{array}{cc}
-1-j \\
1+j \\
0 \\
0 \\
\end{array} } \right]\right\rangle
=\left\langle \left[ {\begin{array}{cc}
-1+j \\
1-j \\
0 \\
0 \\
\end{array} } \right]\right\rangle 
=\left\langle \left[ {\begin{array}{cc}
1-j \\
-1+j \\
0 \\
0 \\
\end{array} } \right]\right\rangle
=\left\langle \left[ {\begin{array}{cc}
2j \\
-2j \\
0 \\
0 \\
\end{array} } \right]\right\rangle
=\left\langle \left[ {\begin{array}{cc}
-2j \\
2j \\
0 \\
0 \\
\end{array} } \right]\right\rangle
=\left\langle \left[ {\begin{array}{cc}
2 \\
-2 \\
0 \\
0 \\
\end{array} } \right]\right\rangle 
=\left\langle \left[ {\begin{array}{cc}
-2 \\
2 \\
0 \\
0 \\
\end{array} } \right]\right\rangle $\\
$3. \left\langle \left[ {\begin{array}{cc}
1+j \\
1-j \\
0 \\
0 \\
\end{array} } \right]\right\rangle
=\left\langle \left[ {\begin{array}{cc}
-1-j \\
-1+j \\
0 \\
0 \\
\end{array} } \right]\right\rangle
=\left\langle \left[ {\begin{array}{cc}
-1+j \\
1+j \\
0 \\
0 \\
\end{array} } \right]\right\rangle 
=\left\langle \left[ {\begin{array}{cc}
1-j \\
-1-j \\
0 \\
0 \\
\end{array} } \right]\right\rangle
=\left\langle \left[ {\begin{array}{cc}
2j \\
2 \\
0 \\
0 \\
\end{array} } \right]\right\rangle
=\left\langle \left[ {\begin{array}{cc}
-2j \\
-2 \\
0 \\
0 \\
\end{array} } \right]\right\rangle
=\left\langle \left[ {\begin{array}{cc}
-2 \\
2j \\
0 \\
0 \\
\end{array} } \right]\right\rangle 
=\left\langle \left[ {\begin{array}{cc}
2 \\
-2j \\
0 \\
0 \\
\end{array} } \right]\right\rangle  $ ~~~~\\
$ 4.\left\langle \left[ {\begin{array}{cc}
1+j \\
-1+j \\
0 \\
0 \\
\end{array} } \right]\right\rangle
=\left\langle \left[ {\begin{array}{cc}
-1-j \\
1-j \\
0 \\
0 \\
\end{array} } \right]\right\rangle
=\left\langle \left[ {\begin{array}{cc}
-1+j \\
-1-j \\
0 \\
0 \\
\end{array} } \right]\right\rangle 
=\left\langle \left[ {\begin{array}{cc}
1-j \\
1+j \\
0 \\
0 \\
\end{array} } \right]\right\rangle 
=\left\langle \left[ {\begin{array}{cc}
2j \\
-2 \\
0 \\
0 \\
\end{array} } \right]\right\rangle
=\left\langle \left[ {\begin{array}{cc}
-2j \\
2 \\
0 \\
0 \\
\end{array} } \right]\right\rangle
=\left\langle \left[ {\begin{array}{cc}
2 \\
2j \\
0 \\
0 \\
\end{array} } \right]\right\rangle 
=\left\langle \left[ {\begin{array}{cc}
-2 \\
-2j \\
0 \\
0 \\
\end{array} } \right]\right\rangle $ ~~~~\\
$ 5. \left\langle \left[ {\begin{array}{cc}
1+j \\
2j \\
0 \\
0 \\
\end{array} } \right]\right\rangle
=\left\langle \left[ {\begin{array}{cc}
-1-j \\
-2j \\
0 \\
0 \\
\end{array} } \right]\right\rangle
=\left\langle \left[ {\begin{array}{cc}
-1+j \\
2 \\
0 \\
0 \\
\end{array} } \right]\right\rangle 
=\left\langle \left[ {\begin{array}{cc}
1-j \\
-2 \\
0 \\
0 \\
\end{array} } \right]\right\rangle $ ~~~~
$6. \left\langle \left[ {\begin{array}{cc}
1+j \\
-2j \\
0 \\
0 \\
\end{array} } \right]\right\rangle
=\left\langle \left[ {\begin{array}{cc}
-1-j \\
2j \\
0 \\
0 \\
\end{array} } \right]\right\rangle
=\left\langle \left[ {\begin{array}{cc}
-1+j \\
-2 \\
0 \\
0 \\
\end{array} } \right]\right\rangle 
=\left\langle \left[ {\begin{array}{cc}
1-j \\
2 \\
0 \\
0 \\
\end{array} } \right]\right\rangle $ \\
$ 7. \left\langle \left[ {\begin{array}{cc}
1+j \\
2 \\
0 \\
0 \\
\end{array} } \right]\right\rangle
=\left\langle \left[ {\begin{array}{cc}
-1-j \\
-2 \\
0 \\
0 \\
\end{array} } \right]\right\rangle
=\left\langle \left[ {\begin{array}{cc}
-1+j \\
2j \\
0 \\
0 \\
\end{array} } \right]\right\rangle 
=\left\langle \left[ {\begin{array}{cc}
1-j \\
-2j \\
0 \\
0 \\
\end{array} } \right]\right\rangle $ ~~~~
$8. \left\langle \left[ {\begin{array}{cc}
1+j \\
-2 \\
0 \\
0 \\
\end{array} } \right]\right\rangle
=\left\langle \left[ {\begin{array}{cc}
-1-j \\
2 \\
0 \\
0 \\
\end{array} } \right]\right\rangle
=\left\langle \left[ {\begin{array}{cc}
-1+j \\
-2j \\
0 \\
0 \\
\end{array} } \right]\right\rangle 
=\left\langle \left[ {\begin{array}{cc}
1-j \\
2j \\
0 \\
0 \\
\end{array} } \right]\right\rangle $ \\
$9. \left\langle \left[ {\begin{array}{cc}
2j \\
1+j \\
0 \\
0 \\
\end{array} } \right]\right\rangle
=\left\langle \left[ {\begin{array}{cc}
-2j \\
-1-j \\
0 \\
0 \\
\end{array} } \right]\right\rangle
=\left\langle \left[ {\begin{array}{cc}
2 \\
-1+j \\
0 \\
0 \\
\end{array} } \right]\right\rangle 
=\left\langle \left[ {\begin{array}{cc}
-2 \\
1-j \\
0 \\
0 \\
\end{array} } \right]\right\rangle $ ~~~~
$ 10. \left\langle \left[ {\begin{array}{cc}
-2j \\
1+j \\
0 \\
0 \\
\end{array} } \right]\right\rangle
=\left\langle \left[ {\begin{array}{cc}
2j \\
-1-j \\
0 \\
0 \\
\end{array} } \right]\right\rangle
=\left\langle \left[ {\begin{array}{cc}
-2 \\
-1+j \\
0 \\
0 \\
\end{array} } \right]\right\rangle 
=\left\langle \left[ {\begin{array}{cc}
2 \\
1-j \\
0 \\
0 \\
\end{array} } \right]\right\rangle $ \\
$11. \left\langle \left[ {\begin{array}{cc}
2 \\
1+j \\
0 \\
0 \\
\end{array} } \right]\right\rangle
=\left\langle \left[ {\begin{array}{cc}
-2 \\
-1-j \\
0 \\
0 \\
\end{array} } \right]\right\rangle
=\left\langle \left[ {\begin{array}{cc}
2j \\
-1+j \\
0 \\
0 \\
\end{array} } \right]\right\rangle 
=\left\langle \left[ {\begin{array}{cc}
-2j \\
1-j \\
0 \\
0 \\
\end{array} } \right]\right\rangle $ ~~~
$12. \left\langle \left[ {\begin{array}{cc}
-2 \\
1+j \\
0 \\
0 \\
\end{array} } \right]\right\rangle
=\left\langle \left[ {\begin{array}{cc}
2 \\
-1-j \\
0 \\
0 \\
\end{array} } \right]\right\rangle
=\left\langle \left[ {\begin{array}{cc}
-2j \\
-1+j \\
0 \\
0 \\
\end{array} } \right]\right\rangle 
=\left\langle \left[ {\begin{array}{cc}
2j \\
1-j \\
0 \\
0 \\
\end{array} } \right]\right\rangle. $
\caption{Null Spaces of the Singular Fades Subspaces for the case $x_A \neq x'_A, ~ x_B \neq x'_B,~ x_C = x'_C \text{~and~} x_D=x'_D$.}
\label{fig_sfscase2}
\vspace{-.7 cm}
\end{figure*}
%%%%%%%%%%%%%%%%%%%%%%%%%%%%%%%%%%%%%%%%
\vspace{.5cm}
\noindent \textbf{\textit{Case 2:}} Two of $x_A-x'_A$, $x_B-x'_B$, $x_C-x'_C$ or $x_D-x'_D$ are non-zero, resulting in 6 subcases.
Without loss of generality, we consider the case for which the singular fade subspace for this case is given by {\footnotesize \vspace{-0 cm}$\mathcal{S}''=\left\langle \left[ {\begin{array}{cc}
x_{A}-x'_A \\
x_B-x'_B \\
0 \\
0 \\
\end{array} } \right]\right\rangle ^{\bot}.$} Both $x_A- x'_A \text{~and~} x_B-x'_B \in \mathcal{D}$ can take eight non-zero values each. There are therefore, 64 total possibilities for the vector $ \left[x_A-x'_A, ~x_B-x'_B, ~0,~0\right]^t $.

%%%%%%%%%%%%%%%%%%%%%%%%%%%%%%%%%%%%
\begin{lemma} Given {\scriptsize {$v= \left[x_A-x'_A, ~x_B-x'_B, ~x_C-x'_C, ~x_D-x'_D\right]^t$}} with entries from $\mathcal{D}_1 \cup \mathcal{D}_2$, when at least one of $x_A \neq x'_A,~  x_B \neq x'_B , ~ x_C \neq x'_C \text{~and~} x_D \neq x'_D$, there are precisely 4 or 8 vectors (including $v$) with entries from $\mathcal{D}_1 \cup \mathcal{D}_2$ that generate the same vector space over $\mathbb{C}$ as $v$.
\end{lemma}
%%%%%%%%%%%%%%%%%%%%%%%%%%%%%%%%%%%%

We partition and prove this Lemma in three parts: Lemma 2, Lemma 3, Lemma 4 given in Case 2, Case 3, and Case 4 respectively of this paper. 
\begin{lemma} When $x_A \neq x'_A,~  x_B \neq x'_B , ~ x_C=x'_C \text{~and~} x_D=x'_D$, for a given vector $v= \left[x_A-x'_A, ~x_B-x'_B, ~0, ~0\right]^t$ over $\mathcal{D}_1 \cup \mathcal{D}_2$, there are precisely 4 or 8 vectors (including $v$) over $\mathcal{D}_1 \cup \mathcal{D}_2$ that generate the same vector space over $\mathbb{C}$ as $v$.
\end{lemma}
\begin{proof} The difference constellation $\mathcal{D}= \Delta \mathcal{S}=\left\{s-s': s,s' \in \mathcal{S}\right\}$ where $\mathcal{S}$ is 4-PSK constellation, is given by \cite{NMR},
\begin{align}
\nonumber
\Delta \mathcal{S}=\left\{0\right\} & \cup \left\{2sin\left( \pi n / 4\right) e^{j k \pi / 2} | n \text{~odd}\right\}\\
\nonumber
&\cup \left\{2sin\left( \pi n / 4\right) e^{j\left( k \pi / 2 + \pi / 4\right)} | n \text{~even}\right\},
\end{align}
where $ 1 \leq n \leq 2 $ and $ 0 \leq k \leq 3 $. So, $v$ can also be written as,
$$v= \left[ {\begin{array}{cc}
x_{A}-x'_A \\
x_B-x'_B \\
0 \\
0 \\
\end{array} } \right] = \left[ {\begin{array}{cc}
\vspace{0.15cm}
2 sin \frac{ \pi k_{1}}{4} e^{j \phi_{1}} \\
\vspace{0.15cm}
2 sin \frac{ \pi k_{2}}{4} e^{j \phi_{2}} \\
0 \\
0 \\
\end{array} } \right] $$
where $\phi_{i}=k_{i} \pi /2$ if $k_{i}$ is odd and $\phi_{i}=k_{i} \pi /2 + \pi/4$ if $k_{i}$ is even. \\

A vector $w$ over $\mathcal{D}_1 \cup \mathcal{D}_2$ shall generate the same vector space over $\mathbb{C}$ iff $w$ is a scalar multiple of $v$, i.e., for some complex number $r e^{j \theta} \in \mathbb{C}$, 
$$v = r e^{j\theta} w \Rightarrow \left[ {\begin{array}{cc}
\vspace{0.15cm}
2 sin \frac{ \pi k_{1}}{4} e^{j \phi_{1}} \\
\vspace{0.15cm}
2 sin \frac{ \pi k_{2}}{4} e^{j \phi_{2}} \\
0 \\
0 \\
\end{array} } \right] = r e^{j\theta} \left[ {\begin{array}{cc}
\vspace{0.15cm}
2 sin \frac{ \pi k_{3}}{4} e^{j \phi_{3}} \\
\vspace{0.15cm}
2 sin \frac{ \pi k_{4}}{4} e^{j \phi_{4}} \\
0 \\
0 \\
\end{array} } \right] $$ 
where for $i=3,4$ $\phi_{i}=k_{i} \pi /2$ if $k_{i}$ is odd and $\phi_{i}=k_{i} \pi /2 + \pi/4$ if $k_{i}$ is even. \\

Then,
\begin{equation}
\label{firstcomp}
2 sin \frac{ \pi k_{1}}{4} e^{j \phi_{1}} = r e^{j\theta} \times 2 sin \frac{ \pi k_{3}}{4} e^{j \phi_{3}}
\end{equation}
and
\begin{equation}
\label{secondcomp}
2 sin \frac{ \pi k_{2}}{4} e^{j \phi_{2}} = r e^{j\theta} \times 2 sin \frac{ \pi k_{4}}{4} e^{j \phi_{4}}.
\end{equation}

Dividing (\ref{firstcomp}) by (\ref{secondcomp}) and taking modulus of both sides, we get 
\begin{equation}
\label{rel}
\frac{sin \frac{ \pi k_{1}}{4}}{sin \frac{ \pi k_{2}}{4}}  =  \frac{sin \frac{ \pi k_{3}}{4} }{sin \frac{ \pi k_{4}}{4} }.
\end{equation}
As shown in \cite{NMR}, this is possible only if $k_1=k_3$ and $k_2=k_4$. Also, from (\ref{firstcomp}) and (\ref{secondcomp}) we have 
\begin{equation}
\label{div}
\frac{ sin \frac{ \pi k_{1}}{4}}{sin \frac{ \pi k_{2}}{4}} e^{j (\phi_{1} - \phi_{2})} = \frac{ sin \frac{ \pi k_{3}}{4}}{sin \frac{ \pi k_{4}}{4}} e^{j (\phi_{3} - \phi_{4})}.
\end{equation}
From (\ref{rel}) and (\ref{div}), we have 
\begin{equation}
e^{j (\phi_{1} - \phi_{2})} =  e^{j (\phi_{3} - \phi_{4})}.
\end{equation}
Note that here, the LHS is fixed. It therefore suffices to compute the number of values that RHS takes for the fixed value of the LHS. It can be verified, that for a fixed value of $\phi_{1} - \phi_{2}$, there are precisely four pair of values of $\phi_{3} $ and $\phi_{4}$ that result in the same value of $\phi_{3} - \phi_{4}$. We now look at the following possibilities:\\
\textit{Case 1:} $k_1=k_2$. 
Then, $k_3=k_4$, i.e., there are exactly two possibilities for $k_1$ and $k_2$, viz., $k_1=k_2=1$ and $k_1=k_2=2$. With two pairs of values for $k_3$ and $k_4$ and four pairs of values for $\phi_{3} $ and $\phi_{4}$, we have a total of eight set of values that $w$ can take. Hence, in this case, the vector space generated by $v$ can be generated by exactly eight other vectors over $\mathcal{D}_1 \cup \mathcal{D}_2$. 
\textit{Case 2:} $k_1 \neq k_2$. 
Then, $k_3 \neq k_4$, i.e., there is precisely one possibility for $k_1$ and $k_2$, viz., $k_1=k_3$ and $k_1=k_4$. With only one possible set of values for $k_3$ and $k_4$ and four pairs of values for $\phi_{3} $ and $\phi_{4}$, we have a total of four set of values that $w$ can take. Hence, in this case, the vector space generated by $v$ can be generated by exactly four other vectors over $\mathcal{D}_1 \cup \mathcal{D}_2$. 

\end{proof}
%%%%%%%%%%%%%%%%%%%%%%%%%%%%%%%%%%%%

In this case we end up with 12 singular fade subspaces given by the null spaces of the spaces given above in Fig. \ref{fig_sfscase2}. Similarly, for each of the other five subcases, there are 12 singular fade subspaces, resulting in a total of 72 singular fade subspaces for the case.

\vspace{.5cm}
\noindent \textbf{\textit{Case 3:}} Three of $x_A-x'_A$, $x_B-x'_B$, $x_C-x'_C$ or $x_D-x'_D$ are non-zero, resulting in 4 subcases.
Consider the subcase, where the singular fade subspace in this case is given by, {\footnotesize \vspace{-0 cm}$\mathcal{S}'''=\left\langle \left[ {\begin{array}{cc}
x_{A}-x'_A \\
x_B-x'_B \\
x_C-x'_C \\
0 \\
\end{array} } \right]\right\rangle ^{\bot}.$} Here each of $x_A- x'_A,~ x_B-x'_B \text{~and~} x_C-x'_C \in \mathcal{D}$ can take eight non-zero values. There are therefore, 512 total possibilities for the vector $ \left[x_A-x'_A, ~x_B-x'_B, ~x_C-x'_C, ~0\right]^t $. 
%%%%%%%%%%%%%%%%%%%%%%%%%%%%%%%%%%
\begin{lemma} For the case when $x_A- x'_A\neq 0,~  x_B-x'_B \neq 0, ~x_C-x'_C \neq 0 \text{~and~} x_D-x'_D = 0$, for a given vector $v= \left[x_A-x'_A, ~x_B-x'_B, ~x_C-x'_C, ~0\right]^t$ over $\mathcal{D}_1 \cup \mathcal{D}_2$, there are precisely 4 or 8 vectors (including $v$) over $\mathcal{D}_1 \cup \mathcal{D}_2$ that generate the same vector space over $\mathbb{C}$ as $v$.
\end{lemma}
\begin{proof} The proof of this Lemma is similar to the proof of Lemma 1 and is therefore omitted.
\end{proof}
%%%%%%%%%%%%%%%%%%%%%%%%%%%%%

Since $x_A- x'_A,~ x_B-x'_B \text{~and~} x_C-x'_C \in \mathcal{D}$ are non-zero, we can say that $$x_A- x'_A,~ x_B-x'_B \text{~and~} x_C-x'_C \in \mathcal{D}_{1} \cup \mathcal{D}_2.$$ As a result, we have the following three subcases:
\begin{enumerate}
 \item One of $x_A- x'_A,~ x_B-x'_B \text{~and~} x_C-x'_C \in \mathcal{D}_1$
 \item Two of $x_A- x'_A,~ x_B-x'_B \text{~and~} x_C-x'_C \in \mathcal{D}_1$
 \item All of $x_A- x'_A,~ x_B-x'_B \text{~and~} x_C-x'_C \in \mathcal{D}_1$
\end{enumerate}

We deal with each one of the subcases one-by-one.

\noindent \textit{Subcase 1:} One of $x_A- x'_A,~ x_B-x'_B \text{~and~} x_C-x'_C \in \mathcal{D}_1$. Without loss of generality, we assume that $x_A-x'_A \in \mathcal{D}_1$ and $x_B-x'_B, ~ x_C-x'_C \in \mathcal{D}_2$. There are 64 possibilities for the vector $v'=\left[x_{A}-x'_A, ~ x_B-x'_B, ~ x_C-x'_C, ~0\right]^{t}$. For each of the 64 possibilities for the vector $v'=\left[x_{A}-x'_A, ~ x_B-x'_B, ~ x_C-x'_C, ~0\right]^{t}$, precisely 4 vectors of length 4 over $\mathcal{D}_{1} \cup \mathcal{D}_2$ generate the same vector space over $\mathbb{C}$. As a result, this case leads to 16 singular fade subspaces. The same holds for the case when $x_B-x'_B \in \mathcal{D}_1$ and $x_A-x'_A, ~ x_C-x'_C \in \mathcal{D}_2$, or when $x_C-x'_C \in \mathcal{D}_1$ and $x_A-x'_A, ~ x_B-x'_B \in \mathcal{D}_2$. This subcase therefore results in 48 singular fade subspaces. 

\noindent \textit{Subcase 2:} Two of $x_A- x'_A,~ x_B-x'_B \text{~and~} x_C-x'_C \in \mathcal{D}_1$. This subcase can be dealt with similar to the Subcase 1 and hence also results in 48 singular fade subspaces. 

\noindent \textit{Subcase 3:} All of $x_A- x'_A,~ x_B-x'_B \text{~and~} x_C-x'_C \in \mathcal{D}_1 ~(or~ \mathcal{D}_2)$. There are 128 possibilities for the vector $\left[x_{A}-x'_A, ~ x_B-x'_B, ~ x_C-x'_C, ~0\right]^{t}$ over $\mathcal{D}_1 ~(or~ \mathcal{D}_2)$. For a given such vector, 8 other possibilities of 4 length vectors over $ \mathcal{D}_1 \cup \mathcal{D}_2$ that generate the same vector space over $\mathbb{C}$. The subcase therefore leads to a total of 128/8=16 singular fade subspaces.

The three subcases result in a total of 48+48+16=112 singular fade subspaces. So, there are a total of $4 \times 112= 448$ singular fade subspaces corresponding to Case 3.  

\vspace{.5cm}
\noindent \textbf{\textit{Case 4:}} All of $x_A-x'_A$, $x_B-x'_B$, $x_C-x'_C$ or $x_D-x'_D$ are non-zero. The singular fade subspace in this case is given by, {\footnotesize \vspace{-0 cm}$\mathcal{S}''''=\left\langle \left[ {\begin{array}{cc}
x_{A}-x'_A \\
x_B-x'_B \\
x_C-x'_C \\
x_D-x'_D \\
\end{array} } \right]\right\rangle ^{\bot}.$} Here each of $x_A- x'_A,~ x_B-x'_B, ~x_C-x'_C \text{~and~} x_D-x'_D \in \mathcal{D}$ can take eight non-zero values. There are therefore, 4096 total possibilities for the vector $ \left[x_A-x'_A, ~x_B-x'_B, ~x_C-x'_C, ~x_D-x'_D\right]^t $. 

%%%%%%%%%%%%%%%%%%%%%%%%%%%%%%
\begin{lemma} For the case when $x_A \neq x'_A,~  x_B \neq x'_B, ~x_C \neq x'_C  \text{~and~} x_D \neq x'_D $, for a given vector $v= \left[x_A-x'_A, ~x_B-x'_B, ~x_C-x'_C, ~x_D-x'_D\right]^t$ over $\mathcal{D}_1 \cup \mathcal{D}_2$, there are precisely 4 or 8 vectors (including $v$) over $\mathcal{D}_1 \cup \mathcal{D}_2$ that generate the same vector space over $\mathbb{C}$ as $v$.
\end{lemma}
\begin{proof}
As mentioned in the proof of Lemma 1, from \cite{NMR},
\begin{align}
\nonumber
\mathcal{D}=\Delta \mathcal{S}=\left\{0\right\} & \cup \left\{2sin\left( \pi n / 4\right) e^{j k \pi / 2} | n \text{~odd}\right\}\\
\nonumber
&\cup \left\{2sin\left( \pi n / 4\right) e^{j\left( k \pi / 2 + \pi / 4\right)} | n \text{~even}\right\},
\end{align}
where $ 1 \leq n \leq 2 $ and $ 0 \leq k \leq 3 $. Therefore, we can write,
$$v= \left[ {\begin{array}{cc}
x_{A}-x'_A \\
x_B-x'_B \\
x_C-x'_C \\
x_D-x'_D \\
\end{array} } \right] = \left[ {\begin{array}{cc}
\vspace{0.15cm}
2 sin \frac{ \pi k_{1}}{4} e^{j \phi_{1}} \\
\vspace{0.15cm}
2 sin \frac{ \pi k_{2}}{4} e^{j \phi_{2}} \\
\vspace{0.15cm}
2 sin \frac{ \pi k_{3}}{4} e^{j \phi_{3}} \\
2 sin \frac{ \pi k_{4}}{4} e^{j \phi_{4}} \\
\end{array} } \right] $$
where $\phi_{i}=k_{i} \pi /2$ if $k_{i}$ is odd and $\phi_{i}=k_{i} \pi /2 + \pi/4$ if $k_{i}$ is even. \\

A vector $w$ over $\mathcal{D}_1 \cup \mathcal{D}_2$ shall generate the same vector space over $\mathbb{C}$ iff $w$ is a scalar multiple of $v$, i.e. for some complex number $r e^{j \theta} \in \mathbb{C}$, 
$$v = r e^{j\theta} w \Rightarrow \left[ {\begin{array}{cc}
\vspace{0.15cm}
2 sin \frac{ \pi k_{1}}{4} e^{j \phi_{1}} \\
\vspace{0.15cm}
2 sin \frac{ \pi k_{2}}{4} e^{j \phi_{2}} \\
\vspace{0.15cm}
2 sin \frac{ \pi k_{3}}{4} e^{j \phi_{3}} \\
2 sin \frac{ \pi k_{4}}{4} e^{j \phi_{4}} \\
\end{array} } \right] = r e^{j\theta} \left[ {\begin{array}{cc}
\vspace{0.15cm}
2 sin \frac{ \pi k_{5}}{4} e^{j \phi_{5}} \\
\vspace{0.15cm}
2 sin \frac{ \pi k_{6}}{4} e^{j \phi_{6}} \\
\vspace{0.15cm}
2 sin \frac{ \pi k_{7}}{4} e^{j \phi_{7}} \\
\vspace{0.15cm}
2 sin \frac{ \pi k_{8}}{4} e^{j \phi_{8}} \\
\end{array} } \right] $$ 
where for $i=4,5,6$ $\phi_{i}=k_{i} \pi /2$ if $k_{i}$ is odd and $\phi_{i}=k_{i} \pi /2 + \pi/4$ if $k_{i}$ is even. \\

Then,
\begin{equation}
\label{firstcomp2}
2 sin \frac{ \pi k_{1}}{4} e^{j \phi_{1}} = r e^{j\theta} \times 2 sin \frac{ \pi k_{5}}{4} e^{j \phi_{5}},
\end{equation}
\begin{equation}
\label{secondcomp2}
2 sin \frac{ \pi k_{2}}{4} e^{j \phi_{2}} = r e^{j\theta} \times 2 sin \frac{ \pi k_{6}}{4} e^{j \phi_{6}}
\end{equation}
\begin{equation}
\label{thirdcomp2}
2 sin \frac{ \pi k_{3}}{4} e^{j \phi_{3}} = r e^{j\theta} \times 2 sin \frac{ \pi k_{7}}{4} e^{j \phi_{7}}
\end{equation}
and,
\begin{equation}
\label{forthcomp2}
2 sin \frac{ \pi k_{4}}{4} e^{j \phi_{4}} = r e^{j\theta} \times 2 sin \frac{ \pi k_{8}}{4} e^{j \phi_{8}}.
\end{equation}

Dividing (\ref{firstcomp2}) by (\ref{secondcomp2}) and taking modulus of both sides, we get 
\begin{equation}
\label{rel2}
\frac{sin \frac{ \pi k_{1}}{4}}{sin \frac{ \pi k_{2}}{4}}  =  \frac{sin \frac{ \pi k_{5}}{4} }{sin \frac{ \pi k_{6}}{4} }.
\end{equation}
As shown in \cite{NMR}, this is possible only if $k_1=k_5$ and $k_2=k_6$. Similarly, dividing (\ref{firstcomp2}) by (\ref{thirdcomp2}) and taking modulus of both sides, we get 
\begin{equation}
\label{rel3}
\frac{sin \frac{ \pi k_{1}}{4}}{sin \frac{ \pi k_{3}}{4}}  =  \frac{sin \frac{ \pi k_{5}}{4} }{sin \frac{ \pi k_{7}}{4} }.
\end{equation}
This is possible only if $k_1=k_5$ and $k_3=k_7$. 
Similarly, 
\begin{equation}
\label{rel4}
k_1=k_6 \text{~and~} k_4=k_8.
\end{equation}
Also, from (\ref{firstcomp2}) and (\ref{secondcomp2}) we have 
\begin{equation}
\label{div2}
\frac{ sin \frac{ \pi k_{1}}{4}}{sin \frac{ \pi k_{2}}{4}} e^{j (\phi_{1} - \phi_{2})} = \frac{ sin \frac{ \pi k_{5}}{4}}{sin \frac{ \pi k_{6}}{4}} e^{j (\phi_{5} - \phi_{6})}.
\end{equation}
From (\ref{rel2}) and (\ref{div2}), we have 
\begin{equation}
\label{div3}
e^{j (\phi_{1} - \phi_{2})} =  e^{j (\phi_{5} - \phi_{6})}.
\end{equation}
Similarly, 
\begin{equation}
\label{div4}
e^{j (\phi_{1} - \phi_{3})} =  e^{j (\phi_{5} - \phi_{7})}
\end{equation}
and 
\begin{equation}
\label{div5}
e^{j (\phi_{1} - \phi_{4})} =  e^{j (\phi_{5} - \phi_{8})}.
\end{equation}
In (\ref{div3}), (\ref{div4}) and (\ref{div5}), the LHS is fixed. It therefore suffices to compute the number of values that the RHS takes for fixed LHS in the three equations. It can be verified, that for a fixed value of $\phi_{1},~ \phi_{2},~ \phi_{3} \text{~and~} \phi_4$, there are precisely four set of values of $\phi_{5},~ \phi_{6},~ \phi_{7} \text{~and~} \phi_8$ that result in the same value of $\phi_{1} - \phi_{2}$, $\phi_{1} - \phi_{3}$ and $\phi_{1} - \phi_{4}$. We now look at the following possibilities:\\
\textit{Case 1:} $k_1=k_2=k_3=k_4$. 
Then, $k_5=k_6=k_7=k_8$, i.e., there are exactly two possibilities for $k_5$, $k_6$, $k_7$ and $k_8$, viz., $k_5=k_6=k_7=k_8=1$ or $k_5=k_6=k_7=k_8=2$. With two sets of values for $k_5$, $k_6$, $k_7$ and $k_8$ and four sets of values for $\phi_{5} $, $\phi_{6}$, $\phi_7$ and $\phi_8$, we have a total of eight set of values that $w$ can take. Hence, in this case, the vector space generated by $v$ can be generated by exactly eight other vectors over $\mathcal{D}_1 \cup \mathcal{D}_2$. 

\noindent \textit{Case 2:} At least one of $k_i \neq k_j$, for $1 \leq i \leq 4$, $1 \leq j \leq 4$. 
Then, $k_{i+4} \neq k_{j+4}$, so that, using (\ref{rel2}), (\ref{rel3}) and (\ref{rel4}), there is precisely one possibility for $k_5,~ k_6,~k_7,~k_8$, viz., $k_5=k_1$, $k_6=k_2$, $k_7=k_3$ and $k_8=k_4$. With only one possible set of values for $k_5,~k_6,~k_7$ and $k_8$ and four sets of values for $\phi_{5}, ~ \phi_6 $, $\phi_{7}$ and $\phi_{8}$, we have a total of four set of values that $w$ can take. Hence, in this case, the vector space generated by $v$ can be generated by exactly four other vectors over $\mathcal{D}_1 \cup \mathcal{D}_2$.
\end{proof}
%%%%%%%%%%%%%%%%%%%%%%%%%%%%%

In Case 4, since all of $x_A- x'_A,~ x_B-x'_B,~x_C-x'_C  \text{~and~} x_D-x'_D \in \mathcal{D}$ are non-zero, we can say that $$x_A- x'_A,~ x_B-x'_B, ~x_C-x'_C \text{~and~} x_D-x'_D \in \mathcal{D}_{1} \cup \mathcal{D}_2.$$ As a result, we have the following four subcases:
\begin{enumerate}
 \item Only one of $x_A- x'_A,~ x_B-x'_B, ~x_C-x'_C \text{~and~} x_D-x'_D \in \mathcal{D}_1$
 \item Two of $x_A- x'_A,~ x_B-x'_B, ~x_C-x'_C \text{~and~} x_D-x'_D \in \mathcal{D}_1$
 \item Three of $x_A- x'_A,~ x_B-x'_B, ~x_C-x'_C \text{~and~} x_D-x'_D \in \mathcal{D}_1$
 \item All of $x_A- x'_A,~ x_B-x'_B, ~x_C-x'_C \text{~and~} x_D-x'_D \in \mathcal{D}_1$
\end{enumerate}
We deal with each one of the subcases one-by-one.

\textit{Subcase 1:} Only one of $x_A- x'_A,~ x_B-x'_B,~ x_C-x'_C \text{~and~} x_D-x'_D \in \mathcal{D}_1$. Without loss of generality, we assume that $x_A-x'_A \in \mathcal{D}_1$ and $x_B-x'_B, ~ x_C-x'_C, ~ x_D-x'_D \in \mathcal{D}_2$. There are 256 possibilities for the vector $v'=\left[x_{A}-x'_A, ~ x_B-x'_B, ~ x_C-x'_C, ~ x_D-x'_D\right]^{t}$. But precisely 4 vectors of length 4 over $\mathcal{D}_{1} \cup \mathcal{D}_2$, i.e. the $\left\{\pm1, \pm j \right\}$ scalar multiples of the vector generate the same vector space over $\mathbb{C}$. As a result, the case $x_A-x'_A \in \mathcal{D}_1$ and $x_B-x'_B, ~ x_C-x'_C \in \mathcal{D}_2$ leads to 64 singular fade subspaces. The same holds for the case when $x_B-x'_B \in \mathcal{D}_1$ and $x_A-x'_A, ~ x_C-x'_C, ~x_D-x'_D \in \mathcal{D}_2$, when $x_C-x'_C \in \mathcal{D}_1$, and $x_A-x'_A, ~ x_B-x'_B,~x_D-x'_D \in \mathcal{D}_2$, or when $x_D-x'_D \in \mathcal{D}_1$ $x_A-x'_A, ~ x_B-x'_B,~x_D-x'_D \in \mathcal{D}_2$. This subcase therefore results in 256 singular fade subspaces.

%%%%%%%%%%%%%%%%%%%%%%%%%%%%%%%%%%%
\textit{Subcase 2:} Two of $x_A- x'_A,~ x_B-x'_B,~ x_C-x'_C \text{~and~} x_D-x'_D \in \mathcal{D}_1$. 
This subcase can be dealt with similar to the Subcase 1 and results in 384 singular fade subspaces.

%%%%%%%%%%%%%%%%%%%%%%%%%%%%
\textit{Subcase 3:} Three of $x_A- x'_A,~ x_B-x'_B,~ x_C-x'_C \text{~and~} x_D-x'_D \in \mathcal{D}_1$. 
This subcase can also be dealt with similar to the Subcase 1 and results in 256 singular fade subspaces.

%%%%%%%%%%%%%%%%%%%%%%%%%%%%
\textit{Subcase 4:} All of, or none of $x_A- x'_A,~ x_B-x'_B,~ x_C-x'_C \text{~and~} x_D-x'_D \in \mathcal{D}_1$. There are 256 possibilities for the vector $\left[x_{A}-x'_A, ~ x_B-x'_B, ~ x_C-x'_C\right]^{t}$ over $\mathcal{D}_1$. For a given such vector, possibilities of other 4 length vectors over $ \mathcal{D}_1 \cup \mathcal{D}_2$ that generate the same vector space over $\mathbb{C}$ are the $\left\{\pm1, \pm j, \pm1\pm j \right\}$ scalar multiples of the vector. The case $x_A-x'_A, ~x_B-x'_B, ~ x_C-x'_C \in \mathcal{D}_1$ leads to a total of 64 singular fade subspaces. These 64 singular fade subspaces have been listed in the Appendix.

The four cases result in a total of 4+72+448+960=1484 singular fade subspaces. As done for 4-PSK case, this computation can be similarly done for the case when M-PSK ($M > 4$) constellation is used at the nodes. We now discuss how these singular fade subspaces can be removed using 4-fold Latin Hyper-Cube of side 4.

%%%%%%%%%%%%%%%%%%%%%%%%%%%%%%%%%%%%%%%%%%%%%%%%%%%%%%
\section{Removing singular fade subspaces}

We now cluster the possibilities of $\left(x_{A},x_{B},x_{C},x_D\right)$ into a clustering that can be represented by a 4-fold Latin Hyper-Cubes of side 4. This clustering is represented by a constellation given by $\mathcal{S}'$, which is utilized by the relay node R in the BC phase. The objective is to obtain a clustering that removes the singular fade subspaces, and then attempts to minimize the size of this constellation used by R. This clustering to be used at R, first constrains the possibilities of $\left(x_{A},x_{B},x_{C},x_D\right)$ received at the MA phase, with the objective of removing the singular fade subspaces, and fills the entries of a $4 \times 4\times 4 \times 4$ array representing the map to be used at the relay using these constraints, and then completes the partially filled array obtained to form a 4-fold Latin hyper-cube of side 4. The mapping to be used at R can be obtained from the complete Latin hyper-cube keeping in mind the equivalence between the relay map with the 4-fold Latin Hyper-Cube of side 4 as shown in Section III.

During the MA phase for the four-way relaying scenario, nodes A, B, C and D transmit to the relay R. As shown in the previous section, there is a total of 1484 singular fade subspaces. Let the fade state $(h_A, h_B, h_C, h_D)$ denote a point in one of the 1484 singular fade subspaces. The constraints on the $4 \times 4 \times 4 \times 4$ array representing the map at the relay node R during BC phase for a singular fade state, can be obtained using the vectors of differences, viz., $ \left[x_A-x'_A ~ x_B-x'_B,~ x_C-x'_C,~ x_D-x'_D\right]$ contributing to this particular singular fade state. So, if $(h_A, h_B, h_C, h_D) \in \left[x_A-x'_A ~ x_B-x'_B,~ x_C-x'_C,~ x_D-x'_D\right]$, then, for $\left(x_{A}, x_{B}, x_C, x_D\right), \left(x'_{A},x'_{B},x'_C, x'_D\right) \in \mathcal{S}^4$, $ h_A x_{A} + h_B x_{B} + h_C x_C +h_D x_D= h_A x'_{A} + h_B x'_{B} + h_C x'_C +h_D x'_D$. For a clustering to remove the singular fade state $(h_A, h_B, h_C, h_D)$, i.e., for the minimum distance of the clustering to be greater than 0 (Section II), the pair $\left(x_{A}, x_{B}, x_C, x_D\right), \left(x'_{A},x'_{B},x'_C,x'_D\right)$ must be kept in the same cluster. Alternatively, we can say that the entry corresponding to $\left(x_{A}, x_{B}, x_C, x_D\right)$ in the $4 \times 4 \times 4 \times 4$ array must be the same as the entry corresponding to $\left(x'_{A},x'_{B},x'_C, x'_D\right)$. Similarly, every other such pair in $\mathcal{S}^4$ contributing to this same singular fade subspace must be kept in the same cluster. Apart from all such pairs in $\mathcal{S}^4$ being kept in the same cluster of the clustering, in order to remove this particular fade state, there are no other constraints. This constrained $4 \times 4 \times 4 \times 4$ array can then be completed by simply filling the top-most and the left-most empty cell in the earliest file of the earliest cube, with $\mathcal{L}_i, i \geq 1$ in the increasing order of $i$ such that the completed array is a 4-fold Latin Hyper-Cube of side 4. This algorithm is given in detail in Algorithm 1. The above clustering scheme, however, cannot be utilized to remove the singular fade subspaces of \textit{Case 1}, \textit{Case 2} and \textit{Case 3} of the previous section, as shown in the following lemma.

%%%%%%%%%%%%%%%%%%
\begin{lemma}
The clustering map used at the relay node R cannot remove a singular fade state corresponding to the \textit{Case 1}, \textit{Case 2} and \textit{Case 3} of the previous section and simultaneously satisfy the mutually exclusive law.
\end{lemma}
%%%%%%%%%%%%%%%%%%%

We refer to the singular fade subspaces given in \textit{Case 1}, \textit{Case 2} and \textit{Case 3} of the previous section, whose harmful effects cannot be removed by a proper choice of the clustering, as the \textit{non-removable singular fade subspaces}. We now illustrate the removal of a \textit{Case 4} singular fade state with the help of the following example.
%%%%%%%%%%%%%%%%%%%%%%%%%%%%%%%%%%%%%%%%%%%%%%
%%%%%%%%%%%%%%%%%%%%%%%%%%%%%%%%%%%%%%%%%%%%%
\begin{example}
Consider the case for which the singular fade subspace is given by%\vspace{-.1 cm}

% Experiment here.

\begin{tiny}
\noindent\vspace{-.20 cm}
\begin{displaymath}
\mathcal{S}''= \left\langle \left[ {\begin{array}{c}
-1-j\\
1+j\\
1+j\\
1-j\\
\end{array} } \right]\right\rangle ^{\bot}\hspace{-0.3cm}
=\hspace{-0.1cm}\left\langle \left[ {\begin{array}{c}
1-j \\
-1+j \\
-1+j \\
1+j \\
\end{array} } \right]\right\rangle ^{\bot}\hspace{-0.3cm}
=\hspace{-0.1cm}\left\langle \left[ {\begin{array}{c}
1+j \\
-1-j \\
-1-j \\
-1+j \\
\end{array} } \right]\right\rangle ^{\bot}\hspace{-0.3cm}
=\hspace{-0.1cm}\left\langle \left[ {\begin{array}{c}
-1+j \\
1-j \\
1-j \\
-1-j \\
\end{array} } \right]\right\rangle ^{\bot}
\end{displaymath}
\vspace{-0.65cm}

\begin{displaymath}
~~~~~~=\left\langle \left[ {\begin{array}{c}
-2j \\
2j \\
2j \\
2 \\
\end{array} } \right]\right\rangle ^{\bot}
=\left\langle \left[ {\begin{array}{c}
2 \\
-2 \\
-2 \\
2j \\
\end{array} } \right]\right\rangle ^{\bot}
=\left\langle \left[ {\begin{array}{c}
2j \\
-2j \\
-2j \\
-2 \\
\end{array} } \right]\right\rangle ^{\bot}
=\left\langle \left[ {\begin{array}{c}
-2 \\
2 \\
2 \\
-2j \\
\end{array} } \right]\right\rangle ^{\bot}.
\end{displaymath}
\vspace{-.4 cm}
\end{tiny}

The first vector is $\left[ -1-j, ~ 1+j, ~ 1+j, ~ 1-j \right]$. Now, $-1-j$ can be obtained either as a difference of $x_A=-1$ and $x'_A=j$ or as a difference of $x_A=-j$ and $x'_A=1$. Similarly, $1+j$ can be obtained as a difference of $1$ and $-j$ or as a difference of $j$ and $-1$; $1-j$ can be obtained as a difference of $x_D=1$ and $x'_D=j$ or as a difference of $x_D=-j$ and $x'_D=1$. Thus, the entries corresponding to {\footnotesize $\left\{(-1,1,1,1),(j,-j,-j,j)\right\}$, $\left\{(-1,1,1,-j),(j,-j,-j,-1)\right\}$, $\left\{(-1,1,j,1),(j,-j,-1,j)\right\}$, $\left\{(-1,1,j,-j),(j,-j,-1,-1)\right\}$, $\left\{(-1,j,1,1),(j,-1,-j,j)\right\}$, $\left\{(-1,j,1,-j),(j,-1,,-j,-1)\right\}$, $\left\{(-1,j,j,1),(j,-1,-1,j)\right\}$, $\left\{(-1,j,j,-j),(j,-1,-1,-1)\right\}$, $\left\{(-j,1,1,1),(1,-j,-j,j)\right\}$, $\left\{(-j,1,1,-j),(1,-j,-j,-1)\right\}$, $\left\{(-j,1,j,1),(1,-j,-1,j)\right\}$, $\left\{(-j,1,j,-j),(1,-j,-1,-1)\right\}$, $\left\{(-j,j,1,1),(1,-1,-j,j)\right\}$, $\left\{(-j,j,1,-j),(1,-1,-j,-1)\right\}$, $\left\{(-j,j,j,1),(1,-1,-1,j)\right\}$, $\left\{(-j,j,j,-j),(1,-1,-1,-1)\right\}$ } must be the same in the $4 \times 4 \times 4 \times 4$ array representing the clustering. Similarly the other constraints can be obtained. Replacing $1,j,-1,-j$ with $0,1,2,3$, these constraints can be represented in $4 \times 4 \times 4 \times 4$ array representing the clustering. The constrained Hyper-Cube of side 4 is completed using Algorithm 1 to form a 4-fold Latin Hyper-Cube of side 4. The completed Latin Hyper Cube is as shown in Table \ref{fig:example}. 

Similarly, for each of the 960 possibilities of singular fade subspaces of \textit{Case 4} of the previous section, a clustering of size between 64 to 90 can be achieved by first constraining the array representing the relay map in order to remove the singular fade state and then completing the constrained array using the provided algorithm.

%%%%%%%%%%%%%%%%%%%%%%%%%%%%%%%%%%%
\begin{algorithm}
\SetLine
\linesnumbered
\KwIn{The constrained $4 \times 4\times 4 \times 4$ array}
\KwOut{A 4-fold Latin Hyper-Cube of side 4 representing the clustering map at the relay}

Start with the constrained $4 \times 4\times 4 \times 4$ array $\mathcal{X}$

Initialize all empty cells of $\mathcal{X}$ to 0

%Let $\mathcal{Y}$ denote the $4 \times 16$ matrix obtained by concatenating $\mathcal{X}$ row-wise, and let $\mathcal{Z}$ denote the $16 \times 4$ matrix obtained by concatenating $\mathcal{X}$ column-wise
%
The $\left(i,j,k,l\right)^{th}$ element of $\mathcal{X}$ is the $i^{th}$ cube, the $j^{th}$ file, the $k^{th}$ row and the $l^th$ column cell.

\For{$1\leq i\leq 4 $}{

\For{$1\leq j\leq 4 $}{

\For{$1\leq k\leq 4 $}{

\For{$1\leq l\leq 4 $}{

\If{cell $\left(i,j,k,l\right)$ of $\mathcal{X}$ is NULL}{

Initialize c=1

\eIf{$\mathcal{L}_{c}$ does not occur in the $i^{th}$ cube of $\mathcal{X}$, the $j^{th}$ file of $\mathcal{X}$, the $k^{th}$ row of $\mathcal{X}$ and the $l^{th}$ column of $\mathcal{X}$}{
replace 0 at cell $\left(i,j,k,l\right)$ of $\mathcal{X}$ with $\mathcal{L}_{c}$\;
%replace $\mathcal{Y}$ with the $4 \times 16$ matrix obtained by concatenating $\mathcal{X}$ row-wise, and $\mathcal{Z}$ by the $16 \times 4$ matrix obtained by concatenating $\mathcal{X}$ column-wise\;
}{
c=c+1\; 
}
}
}
}
}
}
\caption{Obtaining the 4-fold Latin Hyper-Cube of side 4 from the constrained $4 \times 4 \times 4 \times 4 $ array}

\end{algorithm}

%%%%%%%%%%%%%%%%%%%%%%%%%%%%%%%%%%%

\begin{table*}[tp]
%\vspace{-.5 cm}
\tiny
\centering
\renewcommand{\arraystretch}{1.3}
\begin{tabular}{!{\vrule width 1.2pt}c!{\vrule width 0.9pt}p{0.3cm}|p{0.3cm}|p{0.3cm}|p{0.3cm}!{\vrule width 1.2pt}c!{\vrule width 0.9pt}p{0.3cm}|p{0.3cm}|p{0.3cm}|p{0.3cm}!{\vrule width 1.2pt}c!{\vrule width 0.9pt}p{0.3cm}|p{0.3cm}|p{0.3cm}|p{0.3cm}!{\vrule width 1.2pt}c!{\vrule width 0.9pt}p{0.3cm}|p{0.3cm}|p{0.3cm}|p{0.3cm}!{\vrule width 1.2pt}}\noalign{\hrule height 1.2pt}
 $x_A=0$   &\multirow{2}{*}{0} & \multirow{2}{*}{1} & \multirow{2}{*}{2} & \multirow{2}{*}{3}                  & $x_A=1$ & \multirow{2}{*}{0} & \multirow{2}{*}{1} & \multirow{2}{*}{2} & \multirow{2}{*}{3}                                    & $x_A=2$ & \multirow{2}{*}{0} & \multirow{2}{*}{1} & \multirow{2}{*}{2} & \multirow{2}{*}{3}            &  $x_A=3$ & \multirow{2}{*}{0} & \multirow{2}{*}{1} & \multirow{2}{*}{2} & \multirow{2}{*}{3} \\
 $x_B=0$   &   &   &   &                                                                                       & $x_B=0$ &   &   &   &                                                                                                          & $x_B=0$ &   &   &   &                                                                                  &  $x_B=0$ &   &   &   &   \\\noalign{\hrule height 0.9pt}
   0       &    $\mathcal{L}_{16}$  &    $\mathcal{L}_{17}$    & $ \mathcal{L}_{18}$   &  $\mathcal{L}_{19}$   & 0   &    $\mathcal{L}_{20}$  &    $\mathcal{L}_{21}$    & $ \pmb{\mathcal{L}}_{\pmb{10}}$   &  $\mathcal{L}_{9}$               & 0   &    $\pmb{\mathcal{L}}_{\pmb{1}}$  &    $\mathcal{L}_{22}$    & $ \pmb{\mathcal{L}}_{\pmb{15}}$   &  $\pmb{\mathcal{L}}_{\pmb{2}}$  &  0       &    $\pmb{\mathcal{L}}_{\pmb{8}}$  &    $\mathcal{L}_{23}$    & $ \mathcal{L}_{24}$   &  $\pmb{\mathcal{L}}_{\pmb{7}}$ \\\hline
   1       &    $\mathcal{L}_{30}$  &    $\mathcal{L}_{29}$    & $ \mathcal{L}_{32}$   &  $\mathcal{L}_{31}$   & 1   &    $\mathcal{L}_{26}$  &    $\mathcal{L}_{25}$    & $ \mathcal{L}_{28}$               &  $\mathcal{L}_{27}$              & 1   &    $\pmb{\mathcal{L}}_{\pmb{3}}$  &    $\mathcal{L}_{36}$    & $ \mathcal{L}_{35}$               &  $\pmb{\mathcal{L}}_{\pmb{4}}$  &  1       &    $\pmb{\mathcal{L}}_{\pmb{6}}$  &    $\mathcal{L}_{34}$    & $ \mathcal{L}_{33}$   &  $\pmb{\mathcal{L}}_{\pmb{5}}$ \\\hline
   2       &    $\mathcal{L}_{47}$  &    $\mathcal{L}_{48}$    & $ \mathcal{L}_{45}$   &  $\mathcal{L}_{46}$   & 2   &    $\mathcal{L}_{51}$  &    $\mathcal{L}_{52}$    & $ \mathcal{L}_{49}$               &  $\mathcal{L}_{50}$              & 2   &    $\mathcal{L}_{39}$             &    $\mathcal{L}_{40}$    & $ \mathcal{L}_{37}$               &  $\mathcal{L}_{38}$             &  2       &    $\mathcal{L}_{43}$             &    $\mathcal{L}_{44}$    & $ \mathcal{L}_{41}$   &  $\mathcal{L}_{42}$ \\\hline
   3       &    $\mathcal{L}_{64}$  &    $\mathcal{L}_{63}$    & $ \mathcal{L}_{62}$   &  $\mathcal{L}_{61}$   & 3   &    $\mathcal{L}_{60}$  &    $\mathcal{L}_{59}$    & $ \pmb{\mathcal{L}}_{\pmb{12}}$   &  $\pmb{\mathcal{L}}_{\pmb{11}}$  & 3   &    $\mathcal{L}_{58}$             &    $\mathcal{L}_{57}$    & $ \pmb{\mathcal{L}}_{\pmb{13}}$   &  $\pmb{\mathcal{L}}_{\pmb{14}}$ &  3       &    $\mathcal{L}_{56}$             &    $\mathcal{L}_{55}$    & $ \mathcal{L}_{54}$   &  $\mathcal{L}_{53}$ \\\noalign{\hrule height 1.2pt}
% T0 R1
 $x_A=0$   &\multirow{2}{*}{0} & \multirow{2}{*}{1} & \multirow{2}{*}{2} & \multirow{2}{*}{3}                                   & $x_A=1$ & \multirow{2}{*}{0} & \multirow{2}{*}{1} & \multirow{2}{*}{2} & \multirow{2}{*}{3}                & $x_A=2$ & \multirow{2}{*}{0} & \multirow{2}{*}{1} & \multirow{2}{*}{2} & \multirow{2}{*}{3}                   &  $x_A=3$ & \multirow{2}{*}{0} & \multirow{2}{*}{1} & \multirow{2}{*}{2} & \multirow{2}{*}{3} \\
 $x_B=1$   &   &   &   &                                                                                                        & $x_B=1$ &   &   &   &                                                                                      & $x_B=1$ &   &   &   &                                                                                         &  $x_B=1$ &   &   &   &   \\\noalign{\hrule height 0.9pt}
   0   &    $\mathcal{L}_{25}$             &    $\mathcal{L}_{26}$                & $ \mathcal{L}_{27}$   &  $\mathcal{L}_{28}$ & 0   &    $\mathcal{L}_{29}$  &    $\mathcal{L}_{30}$    & $ \mathcal{L}_{31}$   &  $\mathcal{L}_{32}$ &0   &    $\pmb{\mathcal{L}}_{\pmb{5}}$ &    $\mathcal{L}_{33}$    & $ \mathcal{L}_{34}$   &  $\pmb{\mathcal{L}}_{\pmb{6}}$ & 0   &    $\pmb{\mathcal{L}}_{\pmb{4}}$  &    $\mathcal{L}_{35}$                & $ \mathcal{L}_{36}$   &  $\pmb{\mathcal{L}}_{\pmb{3}}$ \\\hline
   1   &    $\pmb{\mathcal{L}}_{\pmb{9}}$  &    $\pmb{\mathcal{L}}_{\pmb{10}}$    & $ \mathcal{L}_{20}$   &  $\mathcal{L}_{21}$ & 1   &    $\mathcal{L}_{17}$  &    $\mathcal{L}_{16}$    & $ \mathcal{L}_{19}$   &  $\mathcal{L}_{18}$ &1   &    $\pmb{\mathcal{L}}_{\pmb{7}}$ &    $\mathcal{L}_{24}$    & $ \mathcal{L}_{23}$   &  $\pmb{\mathcal{L}}_{\pmb{8}}$ & 1   &    $\pmb{\mathcal{L}}_{\pmb{2}}$  &    $\pmb{\mathcal{L}}_{\pmb{15}}$    & $ \mathcal{L}_{22}$   &  $\pmb{\mathcal{L}}_{\pmb{1}}$ \\\hline
   2   &    $\pmb{\mathcal{L}}_{\pmb{11}}$  &    $\pmb{\mathcal{L}}_{\pmb{12}}$   & $ \mathcal{L}_{60}$   &  $\mathcal{L}_{59}$ & 2   &    $\mathcal{L}_{63}$  &    $\mathcal{L}_{64}$    & $ \mathcal{L}_{61}$   &  $\mathcal{L}_{62}$ &2   &    $\mathcal{L}_{55}$            &    $\mathcal{L}_{56}$    & $ \mathcal{L}_{53}$   &  $\mathcal{L}_{54}$           & 2   &    $\pmb{\mathcal{L}}_{\pmb{14}}$  &    $\pmb{\mathcal{L}}_{\pmb{13}}$    & $ \mathcal{L}_{58}$   &  $\mathcal{L}_{57}$ \\\hline
   3   &    $\mathcal{L}_{52}$              &    $\mathcal{L}_{51}$               & $ \mathcal{L}_{50}$   &  $\mathcal{L}_{49}$ & 3   &    $\mathcal{L}_{48}$  &    $\mathcal{L}_{47}$    & $ \mathcal{L}_{46}$   &  $\mathcal{L}_{45}$ &3   &    $\mathcal{L}_{44}$            &    $\mathcal{L}_{43}$    & $ \mathcal{L}_{42}$   &  $\mathcal{L}_{41}$           & 3   &    $\mathcal{L}_{40}$              &    $\mathcal{L}_{39}$                & $ \mathcal{L}_{38}$   &  $\mathcal{L}_{37}$ \\\noalign{\hrule height 1.2pt}
% T0 R2
 $x_A=0$   &\multirow{2}{*}{0} & \multirow{2}{*}{1} & \multirow{2}{*}{2} & \multirow{2}{*}{3}                                            & $x_A=1$ & \multirow{2}{*}{0} & \multirow{2}{*}{1} & \multirow{2}{*}{2} & \multirow{2}{*}{3}            & $x_A=2$ & \multirow{2}{*}{0} & \multirow{2}{*}{1} & \multirow{2}{*}{2} & \multirow{2}{*}{3}                                    &  $x_A=3$ & \multirow{2}{*}{0} & \multirow{2}{*}{1} & \multirow{2}{*}{2} & \multirow{2}{*}{3} \\
 $x_B=2$   &   &   &   &                                                                                                                 & $x_B=2$ &   &   &   &                                                                                  & $x_B=2$ &   &   &   &                                                                                                          &  $x_B=2$ &   &   &   &   \\\noalign{\hrule height 0.9pt}
 0   &    $\mathcal{L}_{37}$              &    $\mathcal{L}_{38}$               & $ \mathcal{L}_{39}$              &  $\mathcal{L}_{40}$ &0   &    $\mathcal{L}_{41}$  &    $\mathcal{L}_{42}$               & $ \mathcal{L}_{43}$              &  $\mathcal{L}_{44}$  & 0   &    $\mathcal{L}_{45}$  &    $\mathcal{L}_{46}$    & $ \mathcal{L}_{47}$   &  $\mathcal{L}_{48}$     & 0   &    $\mathcal{L}_{49}$              &    $\mathcal{L}_{50}$               & $ \mathcal{L}_{51}$   &  $\mathcal{L}_{52}$ \\\hline
 1   &    $\pmb{\mathcal{L}}_{\pmb{13}}$  &    $\pmb{\mathcal{L}}_{\pmb{14}}$   & $ \mathcal{L}_{57}$              &  $\mathcal{L}_{58}$ &1   &    $\mathcal{L}_{54}$  &    $\mathcal{L}_{53}$               & $ \mathcal{L}_{56}$              &  $\mathcal{L}_{55}$  & 1   &    $\mathcal{L}_{62}$  &    $\mathcal{L}_{61}$    & $ \mathcal{L}_{64}$   &  $\mathcal{L}_{63}$     & 1   &    $\pmb{\mathcal{L}}_{\pmb{12}}$  &    $\pmb{\mathcal{L}}_{\pmb{11}}$    & $ \mathcal{L}_{59}$   &  $\mathcal{L}_{60}$ \\\hline
 2   &    $\pmb{\mathcal{L}}_{\pmb{15}}$  &    $\pmb{\mathcal{L}}_{\pmb{2}}$    & $ \pmb{\mathcal{L}}_{\pmb{1}}$   &  $\mathcal{L}_{22}$ &2   &    $\mathcal{L}_{23}$  &    $\pmb{\mathcal{L}}_{\pmb{7}}$    & $ \pmb{\mathcal{L}}_{\pmb{8}}$   &  $\mathcal{L}_{24}$  & 2   &    $\mathcal{L}_{18}$  &    $\mathcal{L}_{19}$    & $ \mathcal{L}_{16}$   &  $\mathcal{L}_{17}$     & 2   &    $\pmb{\mathcal{L}}_{\pmb{10}}$  &    $\pmb{\mathcal{L}}_{\pmb{9}}$    & $ \mathcal{L}_{21}$   &  $\mathcal{L}_{20}$ \\\hline
 3   &    $\mathcal{L}_{36}$              &    $\pmb{\mathcal{L}}_{\pmb{4}}$    & $ \pmb{\mathcal{L}}_{\pmb{3}}$   &  $\mathcal{L}_{35}$ &3   &    $\mathcal{L}_{34}$  &    $\pmb{\mathcal{L}}_{\pmb{5}}$    & $ \pmb{\mathcal{L}}_{\pmb{6}}$   &  $\mathcal{L}_{33}$  & 3   &    $\mathcal{L}_{32}$  &    $\mathcal{L}_{31}$    & $ \mathcal{L}_{30}$   &  $\mathcal{L}_{29}$     & 3   &    $\mathcal{L}_{28}$              &    $\mathcal{L}_{27}$                & $ \mathcal{L}_{26}$   &  $\mathcal{L}_{25}$ \\\noalign{\hrule height 1.2pt}

 $x_A=0$   &\multirow{2}{*}{0} & \multirow{2}{*}{1} & \multirow{2}{*}{2} & \multirow{2}{*}{3}          & $x_A=1$ & \multirow{2}{*}{0} & \multirow{2}{*}{1} & \multirow{2}{*}{2} & \multirow{2}{*}{3}                                                                 & $x_A=2$ & \multirow{2}{*}{0} & \multirow{2}{*}{1} & \multirow{2}{*}{2} & \multirow{2}{*}{3}                                               &  $x_A=3$ & \multirow{2}{*}{0} & \multirow{2}{*}{1} & \multirow{2}{*}{2} & \multirow{2}{*}{3} \\
 $x_B=3$   &   &   &   &                                                                               & $x_B=3$ &   &   &   &                                                                                                                                       & $x_B=3$ &   &   &   &                                                                                                                     &  $x_B=3$ &   &   &   &   \\\noalign{\hrule height 0.9pt}
 0   &    $\mathcal{L}_{53}$  &    $\mathcal{L}_{54}$               & $ \mathcal{L}_{55}$              &  $\mathcal{L}_{56}$ & 0   &    $\mathcal{L}_{57}$  &    $\mathcal{L}_{58}$               & $ \pmb{\mathcal{L}}_{\pmb{14}}$    &  $\pmb{\mathcal{L}}_{\pmb{13}}$     & 0   &    $\mathcal{L}_{59}$  &    $\mathcal{L}_{60}$    & $\pmb{\mathcal{L}}_{\pmb{11}}$   &  $\pmb{\mathcal{L}}_{\pmb{12}}$ &0   &    $\mathcal{L}_{61}$  &    $\mathcal{L}_{62}$    & $ \mathcal{L}_{63}$   &  $\mathcal{L}_{64}$ \\\hline 
 1   &    $\mathcal{L}_{42}$  &    $\mathcal{L}_{41}$               & $ \mathcal{L}_{44}$              &  $\mathcal{L}_{43}$ & 1   &    $\mathcal{L}_{38}$  &    $\mathcal{L}_{37}$               & $ \mathcal{L}_{40}$                &  $\mathcal{L}_{39}$                 & 1   &    $\mathcal{L}_{50}$  &    $\mathcal{L}_{49}$    & $ \mathcal{L}_{52}$              &  $\mathcal{L}_{51}$                  &1   &    $\mathcal{L}_{46}$  &    $\mathcal{L}_{45}$    & $ \mathcal{L}_{48}$   &  $\mathcal{L}_{47}$ \\\hline 
 2   &    $\mathcal{L}_{33}$  &    $\pmb{\mathcal{L}}_{\pmb{6}}$    & $ \pmb{\mathcal{L}}_{\pmb{5}}$   &  $\mathcal{L}_{24}$ & 2   &    $\mathcal{L}_{35}$  &    $\pmb{\mathcal{L}}_{\pmb{3}}$    & $ \pmb{\mathcal{L}}_{\pmb{4}}$     &  $\mathcal{L}_{36}$                 & 2    &    $\mathcal{L}_{27}$  &    $\mathcal{L}_{28}$    & $ \mathcal{L}_{25}$             &  $\mathcal{L}_{26}$                  &2   &    $\mathcal{L}_{31}$  &    $\mathcal{L}_{32}$    & $ \mathcal{L}_{29}$   &  $\mathcal{L}_{30}$ \\\hline 
 3   &    $\mathcal{L}_{24}$  &    $\pmb{\mathcal{L}}_{\pmb{8}}$    & $ \pmb{\mathcal{L}}_{\pmb{7}}$   &  $\mathcal{L}_{33}$ & 3   &    $\mathcal{L}_{22}$  &    $\pmb{\mathcal{L}}_{\pmb{1}}$    & $ \pmb{\mathcal{L}}_{\pmb{2}}$     &  $\pmb{\mathcal{L}}_{\pmb{15}}$     & 3   &    $\mathcal{L}_{21}$  &    $\mathcal{L}_{20}$    & $ \pmb{\mathcal{L}}_{\pmb{9}}$   &  $\pmb{\mathcal{L}}_{\pmb{10}}$      &3   &    $\mathcal{L}_{19}$  &    $\mathcal{L}_{18}$    & $ \mathcal{L}_{17}$   &  $\mathcal{L}_{16}$ \\\noalign{\hrule height 1.2pt}
\end{tabular}
\vspace{0.05 cm}
\caption{Example 1: Latin-Hyper Cube representing the Relay Map where entries $x_A$'s and $x_B$'s are as mentioned, $x_C$'s entries are along the rows and $x_D$'s entries are along the columns of each $4\times 4$ matrix for fixed values of $x_A$ and $x_B$.}
\label{fig:example}
\end{table*}
\end{example}

\begin{table*}[tp]
%\vspace{-.5 cm}
\tiny
\centering
\renewcommand{\arraystretch}{1.3}
\begin{tabular}{!{\vrule width 1.2pt}c!{\vrule width 0.9pt}p{0.3cm}|p{0.3cm}|p{0.3cm}|p{0.3cm}!{\vrule width 1.2pt}c!{\vrule width 0.9pt}p{0.3cm}|p{0.3cm}|p{0.3cm}|p{0.3cm}!{\vrule width 1.2pt}c!{\vrule width 0.9pt}p{0.3cm}|p{0.3cm}|p{0.3cm}|p{0.3cm}!{\vrule width 1.2pt}c!{\vrule width 0.9pt}p{0.3cm}|p{0.3cm}|p{0.3cm}|p{0.3cm}!{\vrule width 1.2pt}}\noalign{\hrule height 1.2pt}
 $x_A=0$   &\multirow{2}{*}{0} & \multirow{2}{*}{1} & \multirow{2}{*}{2} & \multirow{2}{*}{3}                  & $x_A=1$ & \multirow{2}{*}{0} & \multirow{2}{*}{1} & \multirow{2}{*}{2} & \multirow{2}{*}{3}            & $x_A=2$ & \multirow{2}{*}{0} & \multirow{2}{*}{1} & \multirow{2}{*}{2} & \multirow{2}{*}{3}            &  $x_A=3$ & \multirow{2}{*}{0} & \multirow{2}{*}{1} & \multirow{2}{*}{2} & \multirow{2}{*}{3} \\
 $x_B=0$   &   &   &   &                                                                                       & $x_B=0$ &   &   &   &                                                                                  & $x_B=0$ &   &   &   &                                                                                  &  $x_B=0$ &   &   &   &   \\\noalign{\hrule height 0.9pt}
   0       &    $\mathcal{L}_{1}$  &    $\mathcal{L}_{2}$    & $ \mathcal{L}_{3}$   &  $\mathcal{L}_{4}$       & 0   &    $\mathcal{L}_{17}$  &    $\mathcal{L}_{18}$    & $ \mathcal{L}_{19}$   &  $\mathcal{L}_{20}$  & 0   &    $\mathcal{L}_{33}$  &    $\mathcal{L}_{34}$    & $ \mathcal{L}_{35}$   &  $\mathcal{L}_{36}$  &  0       &    $\mathcal{L}_{49}$  &    $\mathcal{L}_{50}$    & $ \mathcal{L}_{51}$   &  $\mathcal{L}_{52}$ \\\hline
   1       &    $\mathcal{L}_{5}$  &    $\mathcal{L}_{6}$    & $ \mathcal{L}_{7}$   &  $\mathcal{L}_{8}$       & 1   &    $\mathcal{L}_{21}$  &    $\mathcal{L}_{22}$    & $ \mathcal{L}_{23}$   &  $\mathcal{L}_{24}$  & 1   &    $\mathcal{L}_{37}$  &    $\mathcal{L}_{38}$    & $ \mathcal{L}_{39}$   &  $\mathcal{L}_{40}$  &  1       &    $\mathcal{L}_{53}$  &    $\mathcal{L}_{54}$    & $ \mathcal{L}_{55}$   &  $\mathcal{L}_{56}$ \\\hline
   2       &    $\mathcal{L}_{9}$  &    $\mathcal{L}_{10}$    & $ \mathcal{L}_{11}$   &  $\mathcal{L}_{12}$    & 2   &    $\mathcal{L}_{25}$  &    $\mathcal{L}_{26}$    & $ \mathcal{L}_{27}$   &  $\mathcal{L}_{28}$  & 2   &    $\mathcal{L}_{41}$  &    $\mathcal{L}_{42}$    & $ \mathcal{L}_{43}$   &  $\mathcal{L}_{44}$  &  2       &    $\mathcal{L}_{57}$  &    $\mathcal{L}_{58}$    & $ \mathcal{L}_{59}$   &  $\mathcal{L}_{60}$ \\\hline
   3       &    $\mathcal{L}_{13}$  &    $\mathcal{L}_{14}$    & $ \mathcal{L}_{15}$   &  $\mathcal{L}_{16}$   & 3   &    $\mathcal{L}_{29}$  &    $\mathcal{L}_{30}$    & $ \mathcal{L}_{31}$   &  $\mathcal{L}_{32}$  & 3   &    $\mathcal{L}_{45}$  &    $\mathcal{L}_{46}$    & $ \mathcal{L}_{47}$   &  $\mathcal{L}_{48}$  &  3       &    $\mathcal{L}_{61}$  &    $\mathcal{L}_{62}$    & $ \mathcal{L}_{63}$   &  $\mathcal{L}_{64}$ \\\noalign{\hrule height 1.2pt}
% T0 R1
 $x_A=0$   &\multirow{2}{*}{0} & \multirow{2}{*}{1} & \multirow{2}{*}{2} & \multirow{2}{*}{3}                  & $x_A=1$ & \multirow{2}{*}{0} & \multirow{2}{*}{1} & \multirow{2}{*}{2} & \multirow{2}{*}{3}            & $x_A=2$ & \multirow{2}{*}{0} & \multirow{2}{*}{1} & \multirow{2}{*}{2} & \multirow{2}{*}{3}            &  $x_A=3$ & \multirow{2}{*}{0} & \multirow{2}{*}{1} & \multirow{2}{*}{2} & \multirow{2}{*}{3} \\
 $x_B=1$   &   &   &   &                                                                                       & $x_B=1$ &   &   &   &                                                                                  & $x_B=1$ &   &   &   &                                                                                  &  $x_B=1$ &   &   &   &   \\\noalign{\hrule height 0.9pt}
   0   &    $\mathcal{L}_{22}$  &    $\mathcal{L}_{21}$    & $ \mathcal{L}_{24}$   &  $\mathcal{L}_{23}$ & 0   &    $\mathcal{L}_{38}$  &    $\mathcal{L}_{37}$    & $ \mathcal{L}_{40}$   &  $\mathcal{L}_{39}$ &0   &    $\mathcal{L}_{54}$  &    $\mathcal{L}_{53}$    & $ \mathcal{L}_{56}$   &  $\mathcal{L}_{55}$  & 0   &    $\mathcal{L}_{6}$  &    $\mathcal{L}_{5}$    & $ \mathcal{L}_{8}$   &  $\mathcal{L}_{7}$ \\\hline
   1   &    $\mathcal{L}_{18}$  &    $\mathcal{L}_{17}$    & $ \mathcal{L}_{20}$   &  $\mathcal{L}_{19}$ & 1   &    $\mathcal{L}_{34}$  &    $\mathcal{L}_{33}$    & $ \mathcal{L}_{36}$   &  $\mathcal{L}_{35}$ &1   &    $\mathcal{L}_{50}$  &    $\mathcal{L}_{49}$    & $ \mathcal{L}_{52}$   &  $\mathcal{L}_{51}$  & 1   &    $\mathcal{L}_{2}$  &    $\mathcal{L}_{1}$    & $ \mathcal{L}_{4}$   &  $\mathcal{L}_{3}$ \\\hline
   2   &    $\mathcal{L}_{30}$  &    $\mathcal{L}_{29}$    & $ \mathcal{L}_{32}$   &  $\mathcal{L}_{31}$ & 2   &    $\mathcal{L}_{46}$  &    $\mathcal{L}_{45}$    & $ \mathcal{L}_{48}$   &  $\mathcal{L}_{47}$ &2   &    $\mathcal{L}_{62}$  &    $\mathcal{L}_{61}$    & $ \mathcal{L}_{64}$   &  $\mathcal{L}_{63}$  & 2   &    $\mathcal{L}_{14}$  &    $\mathcal{L}_{13}$    & $ \mathcal{L}_{16}$   &  $\mathcal{L}_{15}$ \\\hline
   3   &    $\mathcal{L}_{26}$  &    $\mathcal{L}_{25}$    & $ \mathcal{L}_{28}$   &  $\mathcal{L}_{27}$ & 3   &    $\mathcal{L}_{42}$  &    $\mathcal{L}_{41}$    & $ \mathcal{L}_{44}$   &  $\mathcal{L}_{43}$ &3   &    $\mathcal{L}_{58}$  &    $\mathcal{L}_{57}$    & $ \mathcal{L}_{60}$   &  $\mathcal{L}_{59}$  & 3   &    $\mathcal{L}_{10}$  &    $\mathcal{L}_{9}$    & $ \mathcal{L}_{12}$   &  $\mathcal{L}_{11}$ \\\noalign{\hrule height 1.2pt}
% T0 R2
 $x_A=0$   &\multirow{2}{*}{0} & \multirow{2}{*}{1} & \multirow{2}{*}{2} & \multirow{2}{*}{3}                  & $x_A=1$ & \multirow{2}{*}{0} & \multirow{2}{*}{1} & \multirow{2}{*}{2} & \multirow{2}{*}{3}            & $x_A=2$ & \multirow{2}{*}{0} & \multirow{2}{*}{1} & \multirow{2}{*}{2} & \multirow{2}{*}{3}            &  $x_A=3$ & \multirow{2}{*}{0} & \multirow{2}{*}{1} & \multirow{2}{*}{2} & \multirow{2}{*}{3} \\
 $x_B=2$   &   &   &   &                                                                                       & $x_B=2$ &   &   &   &                                                                                  & $x_B=2$ &   &   &   &                                                                                  &  $x_B=2$ &   &   &   &   \\\noalign{\hrule height 0.9pt}
 0   &    $\mathcal{L}_{43}$  &    $\mathcal{L}_{44}$    & $ \mathcal{L}_{41}$   &  $\mathcal{L}_{42}$ &0   &    $\mathcal{L}_{59}$  &    $\mathcal{L}_{60}$    & $ \mathcal{L}_{57}$   &  $\mathcal{L}_{58}$  & 0   &    $\mathcal{L}_{11}$  &    $\mathcal{L}_{12}$    & $ \mathcal{L}_{9}$   &  $\mathcal{L}_{10}$  & 0   &    $\mathcal{L}_{27}$  &    $\mathcal{L}_{28}$    & $ \mathcal{L}_{25}$   &  $\mathcal{L}_{26}$ \\\hline
 1   &    $\mathcal{L}_{47}$  &    $\mathcal{L}_{48}$    & $ \mathcal{L}_{45}$   &  $\mathcal{L}_{46}$ &1   &    $\mathcal{L}_{63}$  &    $\mathcal{L}_{64}$    & $ \mathcal{L}_{61}$   &  $\mathcal{L}_{62}$  & 1   &    $\mathcal{L}_{15}$  &    $\mathcal{L}_{16}$    & $ \mathcal{L}_{13}$   &  $\mathcal{L}_{14}$ & 1   &    $\mathcal{L}_{31}$  &    $\mathcal{L}_{32}$    & $ \mathcal{L}_{29}$   &  $\mathcal{L}_{30}$ \\\hline
 2   &    $\mathcal{L}_{35}$  &    $\mathcal{L}_{36}$    & $ \mathcal{L}_{33}$   &  $\mathcal{L}_{34}$ &2   &    $\mathcal{L}_{51}$  &    $\mathcal{L}_{52}$    & $ \mathcal{L}_{49}$   &  $\mathcal{L}_{50}$  & 2   &    $\mathcal{L}_{3}$  &    $\mathcal{L}_{4}$    & $ \mathcal{L}_{1}$   &  $\mathcal{L}_{2}$     & 2   &    $\mathcal{L}_{19}$  &    $\mathcal{L}_{20}$    & $ \mathcal{L}_{17}$   &  $\mathcal{L}_{18}$ \\\hline
 3   &    $\mathcal{L}_{39}$  &    $\mathcal{L}_{40}$    & $ \mathcal{L}_{37}$   &  $\mathcal{L}_{38}$ &3   &    $\mathcal{L}_{55}$  &    $\mathcal{L}_{56}$    & $ \mathcal{L}_{53}$   &  $\mathcal{L}_{54}$  & 3   &    $\mathcal{L}_{7}$  &    $\mathcal{L}_{8}$    & $ \mathcal{L}_{5}$   &  $\mathcal{L}_{6}$     & 3   &    $\mathcal{L}_{23}$  &    $\mathcal{L}_{24}$    & $ \mathcal{L}_{21}$   &  $\mathcal{L}_{22}$ \\\noalign{\hrule height 1.2pt}

 $x_A=0$   &\multirow{2}{*}{0} & \multirow{2}{*}{1} & \multirow{2}{*}{2} & \multirow{2}{*}{3}                  & $x_A=1$ & \multirow{2}{*}{0} & \multirow{2}{*}{1} & \multirow{2}{*}{2} & \multirow{2}{*}{3}            & $x_A=2$ & \multirow{2}{*}{0} & \multirow{2}{*}{1} & \multirow{2}{*}{2} & \multirow{2}{*}{3}            &  $x_A=3$ & \multirow{2}{*}{0} & \multirow{2}{*}{1} & \multirow{2}{*}{2} & \multirow{2}{*}{3} \\
 $x_B=3$   &   &   &   &                                                                                       & $x_B=3$ &   &   &   &                                                                                  & $x_B=3$ &   &   &   &                                                                                  &  $x_B=3$ &   &   &   &   \\\noalign{\hrule height 0.9pt}
 0   &    $\mathcal{L}_{64}$  &    $\mathcal{L}_{63}$    & $ \mathcal{L}_{62}$   &  $\mathcal{L}_{61}$ & 0   &    $\mathcal{L}_{48}$  &    $\mathcal{L}_{47}$    & $ \mathcal{L}_{46}$   &  $\mathcal{L}_{45}$     & 0   &    $\mathcal{L}_{32}$  &    $\mathcal{L}_{31}$    & $ \mathcal{L}_{30}$   &  $\mathcal{L}_{29}$ &0   &    $\mathcal{L}_{16}$  &    $\mathcal{L}_{15}$    & $ \mathcal{L}_{14}$   &  $\mathcal{L}_{13}$ \\\hline 
 1   &    $\mathcal{L}_{60}$  &    $\mathcal{L}_{59}$    & $ \mathcal{L}_{58}$   &  $\mathcal{L}_{57}$ & 1   &    $\mathcal{L}_{44}$  &    $\mathcal{L}_{43}$    & $ \mathcal{L}_{42}$   &  $\mathcal{L}_{41}$     & 1   &    $\mathcal{L}_{28}$  &    $\mathcal{L}_{27}$    & $ \mathcal{L}_{26}$   &  $\mathcal{L}_{25}$ &1   &    $\mathcal{L}_{12}$  &    $\mathcal{L}_{11}$    & $ \mathcal{L}_{10}$   &  $\mathcal{L}_{9}$ \\\hline 
 2   &    $\mathcal{L}_{56}$  &    $\mathcal{L}_{55}$    & $ \mathcal{L}_{54}$   &  $\mathcal{L}_{53}$ & 2   &    $\mathcal{L}_{40}$  &    $\mathcal{L}_{39}$    & $ \mathcal{L}_{38}$   &  $\mathcal{L}_{37}$ & 2   &    $\mathcal{L}_{24}$  &    $\mathcal{L}_{23}$    & $ \mathcal{L}_{22}$   &  $\mathcal{L}_{21}$ &2   &    $\mathcal{L}_{8}$  &    $\mathcal{L}_{7}$    & $ \mathcal{L}_{6}$   &  $\mathcal{L}_{5}$ \\\hline 
 3   &    $\mathcal{L}_{52}$  &    $\mathcal{L}_{51}$    & $ \mathcal{L}_{50}$   &  $\mathcal{L}_{49}$ & 3   &    $\mathcal{L}_{36}$  &    $\mathcal{L}_{35}$    & $ \mathcal{L}_{34}$   &  $\mathcal{L}_{33}$  & 3   &    $\mathcal{L}_{20}$  &    $\mathcal{L}_{19}$    & $ \mathcal{L}_{18}$   &  $\mathcal{L}_{17}$ &3   &    $\mathcal{L}_{4}$  &    $\mathcal{L}_{3}$    & $ \mathcal{L}_{2}$   &  $\mathcal{L}_{1}$ \\\noalign{\hrule height 1.2pt}
\end{tabular}
\vspace{0.05 cm}
\caption{The 4-fold Latin Hyper-Cube of side 4 used at the relay node as an encoder in the BC for non-adaptive network coding using two channel uses for all channel conditions}
\label{fig:xor}
\end{table*}

%%%%%%%%%%%%%%%%%%%%%%%%%%%%%%%%%%%%%%%%%%5
\section{SIMULATION RESULTS}

Our scheme enables information exchange in the four way relaying scenario amongst the four users in total two channel uses. It is based on the removal of singular fade states for the case. Removing a singular fade state essentially means ensuring a minimum cluster distance greater than zero for the fade state. This scheme removes a subset of singular fade states, called removable singular fade states. Simulation results presented in this section, done for the case when the end nodes use 4-PSK signal set, identify some cases where the proposed scheme outperforms the naive approach that uses the same map for all fade states and vice verse. 

\begin{figure}[ht]
\centering
\includegraphics[totalheight=2.2in,width=3.6in]{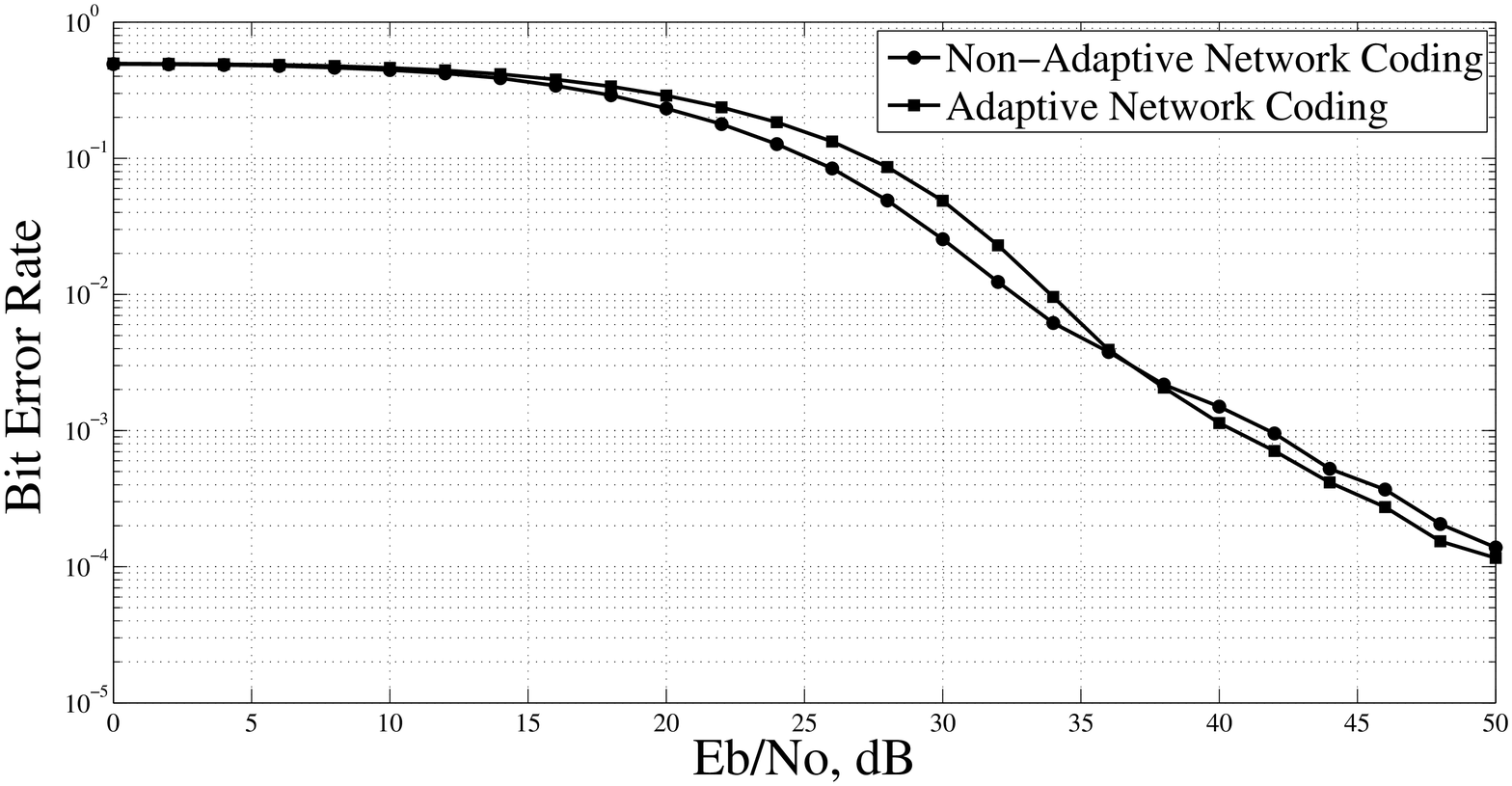}
\vspace{-1 cm}
\caption{SNR Vs BER curves for different schemes}	
\label{fig:plot_bc_rician}	
\end{figure}

Here, $H_A, H_B, H_C, H_D, H'_A, H'_B, H'_C$ and $H'_D$ are distributed according to Rician distribution and channel variances equal to 0 dB. The plot for the case with a Rician Factor of 20 dB for a frame length of 256 bits is as shown Fig. \ref{fig:plot_bc_rician}. The figure shows the SNR vs bit-error-rate curves for the following schemes:  (a) the adaptive network coding scheme presented in this paper; (b) the non-adaptive network coding using two channel uses, in which the same $4 \times 4\times 4 \times 4$ array is used by relay as an encoder for all channel conditions, given by Table \ref{fig:xor}. It can be seen from Fig. \ref{fig:plot_bc_rician} that the scheme based on the adaptive clustering relaying perform better than the schemes based on non-adaptive clustering at high SNR, since adaptive clustering removes 960 singular fade states.

\begin{figure}[ht]
\centering
\includegraphics[totalheight=2.2in,width=3.6in]{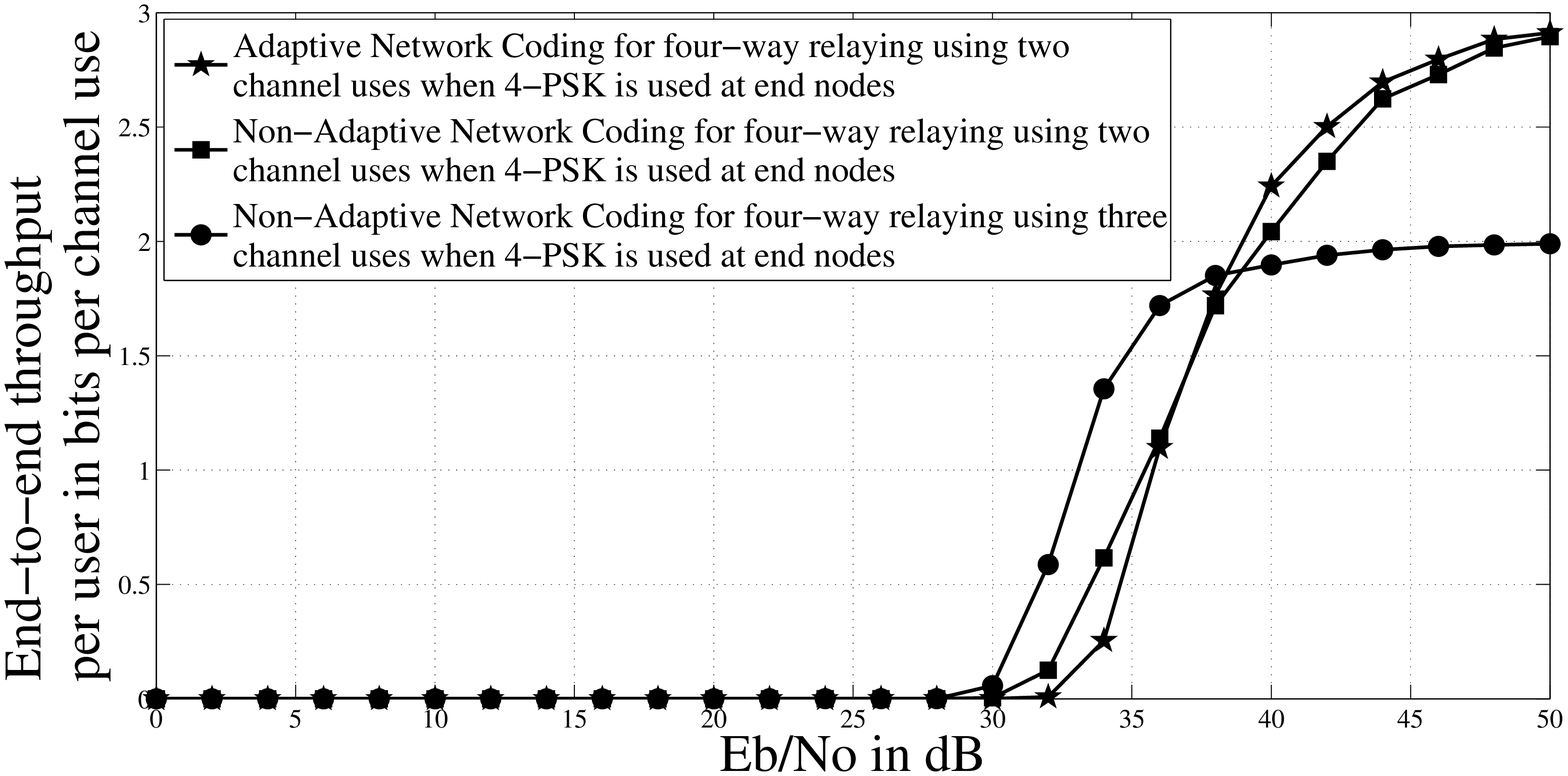}
\vspace{-1 cm}
\caption{SNR Vs Throughput curves for different schemes}
\label{tput}
\end{figure}

Fig. \ref{tput} shows the throughput comparison between (a), (b), and (c): a non-adaptive network coding scheme that uses three channel uses, viz., A, B transmit in the first channel use, C, D transmit in the next channel use, and the relay node transmits the function of these MA phase transmissions in the third channel use. At high SNR regime, the scheme presented in this paper leads to higher throughput, i.e., (a) outperforms both (b) and (c).

%%%%%%%%%%%%%%%%%%%%%%%%%%%%%%

\section{Conclusion}
We considered the four-way wireless relaying scenario, where four nodes operate in half-duplex mode and transmit points from the same M-PSK constellation. Information exchange is made possible using just two channels uses, unlike the existing work done for the case, to the best of our knowledge. The Relay node clusters the $4^{4}$ possible transmitted tuples $\left(x_{A},x_{B},x_{C},x_D\right)$ into various clusters such that \textit{the exclusive law} is satisfied and singular fade subspaces are removed. This necessary requirement of satisfying the exclusive law is shown to be the same as the clustering being represented by a 4-fold Latin Hyper-Cube of side M, if M-PSK is used at the end nodes. Using the proposed scheme that uses two channel uses, the size of the resulting constellation used by the relay node R in the BC phase is reduced from $4^{4}$ to lie between 64 to 90 for M=4. The size of the clustering utilizing modified clustering may not be the best that can be achieved, and it might be possible to fill the array with lesser than 90 symbols.

\begin{figure*}
\begin{center}
\textbf{\textsc{APPENDIX}}\\
\end{center}
\vspace{.5cm}
%\begin{center}
Null Spaces of the Singular Fades Spaces for the case $x_A - x'_A, ~ x_B-x'_B,~ x_C-x'_C \text{~and~} x_D-x'_D \in \mathcal{D}_1$.
%\end{center}
\vspace{.5cm}

%\centering
\tiny
$ 1. \left\langle \left[ {
\begin{array}{cc}
1+j \\
1+j \\
1+j \\
1+j \\
\end{array} } \right]\right\rangle
=\left\langle \left[ {\begin{array}{cc}
-1-j \\
-1-j \\
-1-j \\
-1-j \\
\end{array} } \right]\right\rangle
=\left\langle \left[ {\begin{array}{cc}
-1+j \\
-1+j \\
-1+j \\
-1+j \\
\end{array} } \right]\right\rangle 
=\left\langle \left[ {\begin{array}{cc}
1-j \\
1-j \\
1-j \\
1-j \\
\end{array} } \right]\right\rangle
=\left\langle \left[ {\begin{array}{cc}
2j \\
2j \\
2j \\
2j \\
\end{array} } \right]\right\rangle
=\left\langle \left[ {\begin{array}{cc}
-2j \\
-2j \\
-2j \\
-2j \\
\end{array} } \right]\right\rangle
=\left\langle \left[ {\begin{array}{cc}
2 \\
2 \\
2 \\
2 \\
\end{array} } \right]\right\rangle 
=\left\langle \left[ {\begin{array}{cc}
-2 \\
-2 \\
-2 \\
-2 \\
\end{array} } \right]\right\rangle $\\ 
\vspace{0.15cm} %%%%%%%%%%%%%%%%%%%%%%%%%%%%%%%%%%
$ 2. \left\langle \left[ {\begin{array}{cc}
1+j \\
1+j \\
1+j \\
-1-j \\
\end{array} } \right]\right\rangle
=\left\langle \left[ {\begin{array}{cc}
-1-j \\
-1-j \\
-1-j \\
1+j \\
\end{array} } \right]\right\rangle
=\left\langle \left[ {\begin{array}{cc}
-1+j \\
-1+j \\
-1+j \\
1-j \\
\end{array} } \right]\right\rangle 
=\left\langle \left[ {\begin{array}{cc}
1-j \\
1-j \\
1-j \\
-1+j \\
\end{array} } \right]\right\rangle 
=\left\langle \left[ {\begin{array}{cc}
2j \\
2j \\
2j \\
-2j \\
\end{array} } \right]\right\rangle
=\left\langle \left[ {\begin{array}{cc}
-2j \\
-2j \\
-2j \\
2j \\
\end{array} } \right]\right\rangle
=\left\langle \left[ {\begin{array}{cc}
2 \\
2 \\
2 \\
-2 \\
\end{array} } \right]\right\rangle 
=\left\langle \left[ {\begin{array}{cc}
-2 \\
-2 \\
-2 \\
2 \\
\end{array} } \right]\right\rangle $\\ 
\vspace{0.15cm} %%%%%%%%%%%%%%%%%%%%%%%%%%%%%%%%%%
$ 3. \left\langle \left[ {\begin{array}{cc}
1+j \\
1+j \\
1+j \\
1-j \\
\end{array} } \right]\right\rangle
=\left\langle \left[ {\begin{array}{cc}
-1-j \\
-1-j \\
-1-j \\
-1+j \\
\end{array} } \right]\right\rangle
=\left\langle \left[ {\begin{array}{cc}
-1+j \\
-1+j \\
-1+j \\
1+j \\
\end{array} } \right]\right\rangle 
=\left\langle \left[ {\begin{array}{cc}
1-j \\
1-j \\
1-j \\
-1-j \\
\end{array} } \right]\right\rangle 
=\left\langle \left[ {\begin{array}{cc}
2j \\
2j \\
2j \\
2 \\
\end{array} } \right]\right\rangle
=\left\langle \left[ {\begin{array}{cc}
-2j \\
-2j \\
-2j \\
-2 \\
\end{array} } \right]\right\rangle
=\left\langle \left[ {\begin{array}{cc}
2 \\
2 \\
2 \\
-2j \\
\end{array} } \right]\right\rangle 
=\left\langle \left[ {\begin{array}{cc}
-2 \\
-2 \\
-2 \\
2j \\
\end{array} } \right]\right\rangle $\\ 
\vspace{0.15cm} %%%%%%%%%%%%%%%%%%%%%%%%%%%%%%%%%%
$ 4. \left\langle \left[ {\begin{array}{cc}
1+j \\
1+j \\
1+j \\
-1+j \\
\end{array} } \right]\right\rangle
=\left\langle \left[ {\begin{array}{cc}
-1-j \\
-1-j \\
-1-j \\
1-j \\
\end{array} } \right]\right\rangle
=\left\langle \left[ {\begin{array}{cc}
-1+j \\
-1+j \\
-1+j \\
-1-j \\
\end{array} } \right]\right\rangle 
=\left\langle \left[ {\begin{array}{cc}
1-j \\
1-j \\
1-j \\
1+j \\
\end{array} } \right]\right\rangle 
=\left\langle \left[ {\begin{array}{cc}
2j \\
2j \\
2j \\
-2 \\
\end{array} } \right]\right\rangle
=\left\langle \left[ {\begin{array}{cc}
-2j \\
-2j \\
-2j \\
2 \\
\end{array} } \right]\right\rangle
=\left\langle \left[ {\begin{array}{cc}
2 \\
2 \\
2 \\
2j \\
\end{array} } \right]\right\rangle 
=\left\langle \left[ {\begin{array}{cc}
-2 \\
-2 \\
-2 \\
-2j \\
\end{array} } \right]\right\rangle $\\ 
\vspace{0.15cm} %%%%%%%%%%%%%%%%%%%%%%%%%%%%%%%%%%
$ 5. \left\langle \left[ {\begin{array}{cc}
1+j \\
1+j \\
-1-j \\
1+j \\
\end{array} } \right]\right\rangle
=\left\langle \left[ {\begin{array}{cc}
-1-j \\
-1-j \\
1+j \\
-1+j \\
\end{array} } \right]\right\rangle
=\left\langle \left[ {\begin{array}{cc}
-1+j \\
-1+j \\
1-j \\
-1+j \\
\end{array} } \right]\right\rangle 
=\left\langle \left[ {\begin{array}{cc}
1-j \\
1-j \\
-1+j \\
1-j \\
\end{array} } \right]\right\rangle 
=\left\langle \left[ {\begin{array}{cc}
2j \\
2j \\
-2j \\
2j \\
\end{array} } \right]\right\rangle
=\left\langle \left[ {\begin{array}{cc}
-2j \\
-2j \\
2j \\
-2j \\
\end{array} } \right]\right\rangle
=\left\langle \left[ {\begin{array}{cc}
2 \\
2 \\
-2 \\
2 \\
\end{array} } \right]\right\rangle 
=\left\langle \left[ {\begin{array}{cc}
-2 \\
-2 \\
2 \\
-2 \\
\end{array} } \right]\right\rangle $\\ 
\vspace{0.15cm} %%%%%%%%%%%%%%%%%%%%%%%%%%%%%%%%%%
$ 6. \left\langle \left[ {\begin{array}{cc}
1+j \\
1+j \\
1-j \\
1+j \\
\end{array} } \right]\right\rangle
=\left\langle \left[ {\begin{array}{cc}
-1-j \\
-1-j \\
-1+j \\
-1-j \\
\end{array} } \right]\right\rangle
=\left\langle \left[ {\begin{array}{cc}
-1+j \\
-1+j \\
1+j \\
-1+j \\
\end{array} } \right]\right\rangle 
=\left\langle \left[ {\begin{array}{cc}
1-j \\
1-j \\
-1-j \\
1-j \\
\end{array} } \right]\right\rangle 
=\left\langle \left[ {\begin{array}{cc}
2j \\
2j \\
2 \\
2j \\
\end{array} } \right]\right\rangle
=\left\langle \left[ {\begin{array}{cc}
-2j \\
-2j \\
-2 \\
-2j \\
\end{array} } \right]\right\rangle
=\left\langle \left[ {\begin{array}{cc}
2 \\
2 \\
-2j \\
2 \\
\end{array} } \right]\right\rangle 
=\left\langle \left[ {\begin{array}{cc}
-2 \\
-2 \\
2j \\
-2 \\
\end{array} } \right]\right\rangle $\\ 
\vspace{0.15cm} %%%%%%%%%%%%%%%%%%%%%%%%%%%%%%%%%%
$ 7. \left\langle \left[ {\begin{array}{cc}
1+j \\
1+j \\
-1+j \\
1+j \\
\end{array} } \right]\right\rangle
=\left\langle \left[ {\begin{array}{cc}
-1-j \\
-1-j \\
1-j \\
-1-j \\
\end{array} } \right]\right\rangle
=\left\langle \left[ {\begin{array}{cc}
-1+j \\
-1+j \\
-1-j \\
-1+j \\
\end{array} } \right]\right\rangle 
=\left\langle \left[ {\begin{array}{cc}
1-j \\
1-j \\
1+j \\
1-j \\
\end{array} } \right]\right\rangle 
=\left\langle \left[ {\begin{array}{cc}
2j \\
2j \\
-2 \\
2j \\
\end{array} } \right]\right\rangle
=\left\langle \left[ {\begin{array}{cc}
-2j \\
-2j \\
2 \\
-2j \\
\end{array} } \right]\right\rangle
=\left\langle \left[ {\begin{array}{cc}
2 \\
2 \\
2j \\
2 \\
\end{array} } \right]\right\rangle 
=\left\langle \left[ {\begin{array}{cc}
-2 \\
-2 \\
-2j \\
-2 \\
\end{array} } \right]\right\rangle $\\ 
\vspace{0.15cm} %%%%%%%%%%%%%%%%%%%%%%%%%%%%%%%%%%
$ 8. \left\langle \left[ {\begin{array}{cc}
1+j \\
-1-j \\
1+j \\
1+j \\
\end{array} } \right]\right\rangle
=\left\langle \left[ {\begin{array}{cc}
-1-j \\
1+j \\
-1-j \\
-1+j \\
\end{array} } \right]\right\rangle
=\left\langle \left[ {\begin{array}{cc}
-1+j \\
1-j \\
-1+j \\
-1+j \\
\end{array} } \right]\right\rangle 
=\left\langle \left[ {\begin{array}{cc}
1-j \\
-1+j \\
1-j \\
1-j \\
\end{array} } \right]\right\rangle 
=\left\langle \left[ {\begin{array}{cc}
2j \\
-2j \\
2j \\
2j \\
\end{array} } \right]\right\rangle
=\left\langle \left[ {\begin{array}{cc}
-2j \\
2j \\
-2j \\
-2j \\
\end{array} } \right]\right\rangle
=\left\langle \left[ {\begin{array}{cc}
2 \\
-2 \\
2 \\
2 \\
\end{array} } \right]\right\rangle 
=\left\langle \left[ {\begin{array}{cc}
-2 \\
2 \\
-2 \\
-2 \\
\end{array} } \right]\right\rangle $\\ 
\vspace{0.15cm} %%%%%%%%%%%%%%%%%%%%%%%%%%%%%%%%%%
$ 9. \left\langle \left[ {\begin{array}{cc}
1+j \\
1-j \\
1+j \\
1+j \\
\end{array} } \right]\right\rangle
=\left\langle \left[ {\begin{array}{cc}
-1-j \\
-1+j \\
-1-j \\
-1-j \\
\end{array} } \right]\right\rangle
=\left\langle \left[ {\begin{array}{cc}
-1+j \\
1+j \\
-1+j \\
-1+j \\
\end{array} } \right]\right\rangle 
=\left\langle \left[ {\begin{array}{cc}
1-j \\
-1-j \\
1-j \\
1-j \\
\end{array} } \right]\right\rangle 
=\left\langle \left[ {\begin{array}{cc}
2j \\
2 \\
2j \\
2j \\
\end{array} } \right]\right\rangle
=\left\langle \left[ {\begin{array}{cc}
-2j \\
-2 \\
-2j \\
-2j \\
\end{array} } \right]\right\rangle
=\left\langle \left[ {\begin{array}{cc}
2 \\
-2j \\
2 \\
2 \\
\end{array} } \right]\right\rangle 
=\left\langle \left[ {\begin{array}{cc}
-2 \\
2j \\
-2 \\
-2 \\
\end{array} } \right]\right\rangle $\\ 
\vspace{0.15cm} %%%%%%%%%%%%%%%%%%%%%%%%%%%%%%%%%%
$ 10. \left\langle \left[ {\begin{array}{cc}
1+j \\
-1+j \\
1+j \\
1+j \\
\end{array} } \right]\right\rangle
=\left\langle \left[ {\begin{array}{cc}
-1-j \\
1-j \\
-1-j \\
-1-j \\
\end{array} } \right]\right\rangle
=\left\langle \left[ {\begin{array}{cc}
-1+j \\
-1-j \\
-1+j \\
-1+j \\
\end{array} } \right]\right\rangle 
=\left\langle \left[ {\begin{array}{cc}
1-j \\
1+j \\
1-j \\
1-j \\
\end{array} } \right]\right\rangle 
=\left\langle \left[ {\begin{array}{cc}
2j \\
-2 \\
2j \\
2j \\
\end{array} } \right]\right\rangle
=\left\langle \left[ {\begin{array}{cc}
-2j \\
2 \\
-2j \\
-2j \\
\end{array} } \right]\right\rangle
=\left\langle \left[ {\begin{array}{cc}
2 \\
2j \\
2 \\
2 \\
\end{array} } \right]\right\rangle 
=\left\langle \left[ {\begin{array}{cc}
-2 \\
-2j \\
-2 \\
-2 \\
\end{array} } \right]\right\rangle $\\ 
\vspace{0.15cm} %%%%%%%%%%%%%%%%%%%%%%%%%%%%%%%%%%
$11. \left\langle \left[ {\begin{array}{cc}
-1-j \\
1+j \\
1+j \\
1+j \\
\end{array} } \right]\right\rangle
=\left\langle \left[ {\begin{array}{cc}
1+j \\
-1-j \\
-1-j \\
-1+j \\
\end{array} } \right]\right\rangle
=\left\langle \left[ {\begin{array}{cc}
1-j \\
-1+j \\
-1+j \\
-1+j \\
\end{array} } \right]\right\rangle 
=\left\langle \left[ {\begin{array}{cc}
-1+j \\
1-j \\
1-j \\
1-j \\
\end{array} } \right]\right\rangle 
=\left\langle \left[ {\begin{array}{cc}
-2j \\
2j \\
2j \\
2j \\
\end{array} } \right]\right\rangle
=\left\langle \left[ {\begin{array}{cc}
2j \\
-2j \\
-2j \\
-2j \\
\end{array} } \right]\right\rangle
=\left\langle \left[ {\begin{array}{cc}
-2 \\
2 \\
2 \\
2 \\
\end{array} } \right]\right\rangle 
=\left\langle \left[ {\begin{array}{cc}
2 \\
-2 \\
-2 \\
-2 \\
\end{array} } \right]\right\rangle $\\ 
\vspace{0.15cm} %%%%%%%%%%%%%%%%%%%%%%%%%%%%%%%%%%
$ 12. \left\langle \left[ {\begin{array}{cc}
1-j \\
1+j \\
1+j \\
1+j \\
\end{array} } \right]\right\rangle
=\left\langle \left[ {\begin{array}{cc}
-1+j \\
-1-j \\
-1-j \\
-1-j \\
\end{array} } \right]\right\rangle
=\left\langle \left[ {\begin{array}{cc}
1+j \\
-1+j \\
-1+j \\
-1+j \\
\end{array} } \right]\right\rangle 
=\left\langle \left[ {\begin{array}{cc}
-1-j \\
1-j \\
1-j \\
1-j \\
\end{array} } \right]\right\rangle 
=\left\langle \left[ {\begin{array}{cc}
2 \\
2j \\
2j \\
2j \\
\end{array} } \right]\right\rangle
=\left\langle \left[ {\begin{array}{cc}
-2 \\
-2j \\
-2j \\
-2j \\
\end{array} } \right]\right\rangle
=\left\langle \left[ {\begin{array}{cc}
-2j \\
2 \\
2 \\
2 \\
\end{array} } \right]\right\rangle 
=\left\langle \left[ {\begin{array}{cc}
2j \\
-2 \\
-2 \\
-2 \\
\end{array} } \right]\right\rangle $\\ 
\vspace{0.15cm} %%%%%%%%%%%%%%%%%%%%%%%%%%%%%%%%%%
$ 13. \left\langle \left[ {\begin{array}{cc}
-1+j \\
1+j \\
1+j \\
1+j \\
\end{array} } \right]\right\rangle
=\left\langle \left[ {\begin{array}{cc}
1-j \\
-1-j \\
-1-j \\
-1-j \\
\end{array} } \right]\right\rangle
=\left\langle \left[ {\begin{array}{cc}
-1-j \\
-1+j \\
-1+j \\
-1+j \\
\end{array} } \right]\right\rangle 
=\left\langle \left[ {\begin{array}{cc}
1+j \\
1-j \\
1-j \\
1-j \\
\end{array} } \right]\right\rangle 
=\left\langle \left[ {\begin{array}{cc}
-2 \\
2j \\
2j \\
2j \\
\end{array} } \right]\right\rangle
=\left\langle \left[ {\begin{array}{cc}
2 \\
-2j \\
-2j \\
-2j \\
\end{array} } \right]\right\rangle
=\left\langle \left[ {\begin{array}{cc}
2j \\
2 \\
2 \\
2 \\
\end{array} } \right]\right\rangle 
=\left\langle \left[ {\begin{array}{cc}
-2j \\
-2 \\
-2 \\
-2 \\
\end{array} } \right]\right\rangle $\\ 
\vspace{0.15cm} %%%%%%%%%%%%%%%%%%%%%%%%%%%%%%%%%%
$ 14. \left\langle \left[ {\begin{array}{cc}
1+j \\
1+j \\
-1-j \\
-1-j \\
\end{array} } \right]\right\rangle
=\left\langle \left[ {\begin{array}{cc}
-1-j \\
-1-j \\
1+j \\
1+j \\
\end{array} } \right]\right\rangle
=\left\langle \left[ {\begin{array}{cc}
-1+j \\
-1+j \\
1-j \\
1-j \\
\end{array} } \right]\right\rangle 
=\left\langle \left[ {\begin{array}{cc}
1-j \\
1-j \\
-1+j \\
-1+j \\
\end{array} } \right]\right\rangle 
=\left\langle \left[ {\begin{array}{cc}
2j \\
2j \\
-2j \\
-2j \\
\end{array} } \right]\right\rangle
=\left\langle \left[ {\begin{array}{cc}
-2j \\
-2j \\
2j \\
2j \\
\end{array} } \right]\right\rangle
=\left\langle \left[ {\begin{array}{cc}
2 \\
2 \\
-2 \\
-2 \\
\end{array} } \right]\right\rangle 
=\left\langle \left[ {\begin{array}{cc}
-2 \\
-2 \\
2 \\
2 \\
\end{array} } \right]\right\rangle $\\ 
\vspace{0.15cm} %%%%%%%%%%%%%%%%%%%%%%%%%%%%%%%%%%
$ 15. \left\langle \left[ {\begin{array}{cc}
1+j \\
1+j \\
-1-j \\
1-j \\
\end{array} } \right]\right\rangle
=\left\langle \left[ {\begin{array}{cc}
-1-j \\
-1-j \\
1+j \\
-1+j \\
\end{array} } \right]\right\rangle
=\left\langle \left[ {\begin{array}{cc}
-1+j \\
-1+j \\
1-j \\
1+j \\
\end{array} } \right]\right\rangle 
=\left\langle \left[ {\begin{array}{cc}
1-j \\
1-j \\
-1+j \\
-1-j \\
\end{array} } \right]\right\rangle 
=\left\langle \left[ {\begin{array}{cc}
2j \\
2j \\
-2j \\
2 \\
\end{array} } \right]\right\rangle
=\left\langle \left[ {\begin{array}{cc}
-2j \\
-2j \\
2j \\
-2 \\
\end{array} } \right]\right\rangle
=\left\langle \left[ {\begin{array}{cc}
2 \\
2 \\
-2 \\
-2j \\
\end{array} } \right]\right\rangle 
=\left\langle \left[ {\begin{array}{cc}
-2 \\
-2 \\
2 \\
2j \\
\end{array} } \right]\right\rangle $\\ 
\vspace{0.15cm} %%%%%%%%%%%%%%%%%%%%%%%%%%%%%%%%%%
$ 16. \left\langle \left[ {\begin{array}{cc}
1+j \\
1+j \\
-1-j \\
-1+j \\
\end{array} } \right]\right\rangle
=\left\langle \left[ {\begin{array}{cc}
-1-j \\
-1-j \\
1+j \\
1-j \\
\end{array} } \right]\right\rangle
=\left\langle \left[ {\begin{array}{cc}
-1+j \\
-1+j \\
1-j \\
-1-j \\
\end{array} } \right]\right\rangle 
=\left\langle \left[ {\begin{array}{cc}
1-j \\
1-j \\
-1+j \\
1+j \\
\end{array} } \right]\right\rangle 
=\left\langle \left[ {\begin{array}{cc}
2j \\
2j \\
-2j \\
-2 \\
\end{array} } \right]\right\rangle
=\left\langle \left[ {\begin{array}{cc}
-2j \\
-2j \\
2j \\
2 \\
\end{array} } \right]\right\rangle
=\left\langle \left[ {\begin{array}{cc}
2 \\
2 \\
-2 \\
2j \\
\end{array} } \right]\right\rangle 
=\left\langle \left[ {\begin{array}{cc}
-2 \\
-2 \\
2 \\
-2j \\
\end{array} } \right]\right\rangle $\\ 
\vspace{0.15cm} %%%%%%%%%%%%%%%%%%%%%%%%%%%%%%%%%%
$ 17. \left\langle \left[ {\begin{array}{cc}
1+j \\
1+j \\
1-j \\
-1-j \\
\end{array} } \right]\right\rangle
=\left\langle \left[ {\begin{array}{cc}
-1-j \\
-1-j \\
-1+j \\
1+j \\
\end{array} } \right]\right\rangle
=\left\langle \left[ {\begin{array}{cc}
-1+j \\
-1+j \\
1+j \\
1-j \\
\end{array} } \right]\right\rangle 
=\left\langle \left[ {\begin{array}{cc}
1-j \\
1-j \\
-1-j \\
-1+j \\
\end{array} } \right]\right\rangle 
=\left\langle \left[ {\begin{array}{cc}
2j \\
2j \\
2 \\
-2j \\
\end{array} } \right]\right\rangle
=\left\langle \left[ {\begin{array}{cc}
-2j \\
-2j \\
-2 \\
2j \\
\end{array} } \right]\right\rangle
=\left\langle \left[ {\begin{array}{cc}
2 \\
2 \\
-2j \\
-2 \\
\end{array} } \right]\right\rangle 
=\left\langle \left[ {\begin{array}{cc}
-2 \\
-2 \\
2j \\
2 \\
\end{array} } \right]\right\rangle $\\ 
\vspace{0.15cm} %%%%%%%%%%%%%%%%%%%%%%%%%%%%%%%%%%
$ 18. \left\langle \left[ {\begin{array}{cc}
1+j \\
1+j \\
1-j \\
1-j \\
\end{array} } \right]\right\rangle
=\left\langle \left[ {\begin{array}{cc}
-1-j \\
-1-j \\
-1+j \\
-1+j \\
\end{array} } \right]\right\rangle
=\left\langle \left[ {\begin{array}{cc}
-1+j \\
-1+j \\
1+j \\
1+j \\
\end{array} } \right]\right\rangle 
=\left\langle \left[ {\begin{array}{cc}
1-j \\
1-j \\
-1-j \\
-1-j \\
\end{array} } \right]\right\rangle 
=\left\langle \left[ {\begin{array}{cc}
2j \\
2j \\
2 \\
2 \\
\end{array} } \right]\right\rangle
=\left\langle \left[ {\begin{array}{cc}
-2j \\
-2j \\
-2 \\
-2 \\
\end{array} } \right]\right\rangle
=\left\langle \left[ {\begin{array}{cc}
2 \\
2 \\
-2j \\
-2j \\
\end{array} } \right]\right\rangle 
=\left\langle \left[ {\begin{array}{cc}
-2 \\
-2 \\
2j \\
2j \\
\end{array} } \right]\right\rangle $\\ 
\vspace{0.15cm} %%%%%%%%%%%%%%%%%%%%%%%%%%%%%%%%%%
$ 19. \left\langle \left[ {\begin{array}{cc}
1+j \\
1+j \\
1-j \\
-1+j \\
\end{array} } \right]\right\rangle
=\left\langle \left[ {\begin{array}{cc}
-1-j \\
-1-j \\
-1+j \\
1-j \\
\end{array} } \right]\right\rangle
=\left\langle \left[ {\begin{array}{cc}
-1+j \\
-1+j \\
1+j \\
-1-j \\
\end{array} } \right]\right\rangle 
=\left\langle \left[ {\begin{array}{cc}
1-j \\
1-j \\
-1-j \\
1+j \\
\end{array} } \right]\right\rangle 
=\left\langle \left[ {\begin{array}{cc}
2j \\
2j \\
2 \\
-2 \\
\end{array} } \right]\right\rangle
=\left\langle \left[ {\begin{array}{cc}
-2j \\
-2j \\
-2 \\
2 \\
\end{array} } \right]\right\rangle
=\left\langle \left[ {\begin{array}{cc}
2 \\
2 \\
-2j \\
2j \\
\end{array} } \right]\right\rangle 
=\left\langle \left[ {\begin{array}{cc}
-2 \\
-2 \\
2j \\
-2j \\
\end{array} } \right]\right\rangle $\\ 
\vspace{0.15cm} %%%%%%%%%%%%%%%%%%%%%%%%%%%%%%%%%%
$ 20. \left\langle \left[ {\begin{array}{cc}
1+j \\
1+j \\
-1+j \\
-1-j \\
\end{array} } \right]\right\rangle
=\left\langle \left[ {\begin{array}{cc}
-1-j \\
-1-j \\
1-j \\
1+j \\
\end{array} } \right]\right\rangle
=\left\langle \left[ {\begin{array}{cc}
-1+j \\
-1+j \\
-1-j \\
1-j \\
\end{array} } \right]\right\rangle 
=\left\langle \left[ {\begin{array}{cc}
1-j \\
1-j \\
1+j \\
-1+j \\
\end{array} } \right]\right\rangle 
=\left\langle \left[ {\begin{array}{cc}
2j \\
2j \\
-2 \\
-2j \\
\end{array} } \right]\right\rangle
=\left\langle \left[ {\begin{array}{cc}
-2j \\
-2j \\
2 \\
2j \\
\end{array} } \right]\right\rangle
=\left\langle \left[ {\begin{array}{cc}
2 \\
2 \\
2j \\
-2 \\
\end{array} } \right]\right\rangle 
=\left\langle \left[ {\begin{array}{cc}
-2 \\
-2 \\
-2j \\
2 \\
\end{array} } \right]\right\rangle $\\ 
\vspace{0.15cm} %%%%%%%%%%%%%%%%%%%%%%%%%%%%%%%%%%
$ 21. \left\langle \left[ {\begin{array}{cc}
1+j \\
1+j \\
-1+j \\
1-j \\
\end{array} } \right]\right\rangle
=\left\langle \left[ {\begin{array}{cc}
-1-j \\
-1-j \\
1-j \\
-1+j \\
\end{array} } \right]\right\rangle
=\left\langle \left[ {\begin{array}{cc}
-1+j \\
-1+j \\
-1-j \\
1+j \\
\end{array} } \right]\right\rangle 
=\left\langle \left[ {\begin{array}{cc}
1-j \\
1-j \\
1+j \\
-1-j \\
\end{array} } \right]\right\rangle 
=\left\langle \left[ {\begin{array}{cc}
2j \\
2j \\
-2 \\
2 \\
\end{array} } \right]\right\rangle
=\left\langle \left[ {\begin{array}{cc}
-2j \\
-2j \\
2 \\
-2 \\
\end{array} } \right]\right\rangle
=\left\langle \left[ {\begin{array}{cc}
2 \\
2 \\
2j \\
-2j \\
\end{array} } \right]\right\rangle 
=\left\langle \left[ {\begin{array}{cc}
-2 \\
-2 \\
-2j \\
2j \\
\end{array} } \right]\right\rangle $\\ 
\vspace{0.15cm} %%%%%%%%%%%%%%%%%%%%%%%%%%%%%%%%%%
$ 22. \left\langle \left[ {\begin{array}{cc}
1+j \\
1+j \\
-1+j \\
-1+j \\
\end{array} } \right]\right\rangle
=\left\langle \left[ {\begin{array}{cc}
-1-j \\
-1-j \\
1-j \\
1-j \\
\end{array} } \right]\right\rangle
=\left\langle \left[ {\begin{array}{cc}
-1+j \\
-1+j \\
-1-j \\
-1-j \\
\end{array} } \right]\right\rangle 
=\left\langle \left[ {\begin{array}{cc}
1-j \\
1-j \\
1+j \\
1+j \\
\end{array} } \right]\right\rangle 
=\left\langle \left[ {\begin{array}{cc}
2j \\
2j \\
-2 \\
-2 \\
\end{array} } \right]\right\rangle
=\left\langle \left[ {\begin{array}{cc}
-2j \\
-2j \\
2 \\
2 \\
\end{array} } \right]\right\rangle
=\left\langle \left[ {\begin{array}{cc}
2 \\
2 \\
2j \\
2j \\
\end{array} } \right]\right\rangle 
=\left\langle \left[ {\begin{array}{cc}
-2 \\
-2 \\
-2j \\
-2j \\
\end{array} } \right]\right\rangle $\\ 
\vspace{0.15cm} %%%%%%%%%%%%%%%%%%%%%%%%%%%%%%%%%%
%\label{fig:sfssubcase3}
\end{figure*}

\begin{figure*}
%\centering
\tiny
\vspace{0.15cm} %%%%%%%%%%%%%%%%%%%%%%%%%%%%%%%%%%
$ 23. \left\langle \left[ {\begin{array}{cc}
1+j \\
-1-j \\
1+j \\
-1-j \\
\end{array} } \right]\right\rangle
=\left\langle \left[ {\begin{array}{cc}
-1-j \\
1+j \\
-1-j \\
1+j \\
\end{array} } \right]\right\rangle
=\left\langle \left[ {\begin{array}{cc}
-1+j \\
1-j \\
-1+j \\
1-j \\
\end{array} } \right]\right\rangle 
=\left\langle \left[ {\begin{array}{cc}
1-j \\
-1+j \\
1-j \\
-1+j \\
\end{array} } \right]\right\rangle 
=\left\langle \left[ {\begin{array}{cc}
2j \\
-2j \\
2j \\
-2j \\
\end{array} } \right]\right\rangle
=\left\langle \left[ {\begin{array}{cc}
-2j \\
2j \\
-2j \\
2j \\
\end{array} } \right]\right\rangle
=\left\langle \left[ {\begin{array}{cc}
2 \\
-2 \\
2 \\
-2 \\
\end{array} } \right]\right\rangle 
=\left\langle \left[ {\begin{array}{cc}
-2 \\
2 \\
-2 \\
2 \\
\end{array} } \right]\right\rangle $\\ 
\vspace{0.15cm} %%%%%%%%%%%%%%%%%%%%%%%%%%%%%%%%%%
$ 24. \left\langle \left[ {\begin{array}{cc}
1+j \\
-1-j \\
1+j \\
1-j \\
\end{array} } \right]\right\rangle
=\left\langle \left[ {\begin{array}{cc}
-1-j \\
1+j \\
-1-j \\
-1+j \\
\end{array} } \right]\right\rangle
=\left\langle \left[ {\begin{array}{cc}
-1+j \\
1-j \\
-1+j \\
1+j \\
\end{array} } \right]\right\rangle 
=\left\langle \left[ {\begin{array}{cc}
1-j \\
-1+j \\
1-j \\
-1-j \\
\end{array} } \right]\right\rangle 
=\left\langle \left[ {\begin{array}{cc}
2j \\
-2j \\
2j \\
2 \\
\end{array} } \right]\right\rangle
=\left\langle \left[ {\begin{array}{cc}
-2j \\
2j \\
-2j \\
-2 \\
\end{array} } \right]\right\rangle
=\left\langle \left[ {\begin{array}{cc}
2 \\
-2 \\
2 \\
-2j \\
\end{array} } \right]\right\rangle 
=\left\langle \left[ {\begin{array}{cc}
-2 \\
2 \\
-2 \\
2j \\
\end{array} } \right]\right\rangle $\\ 
$ 25. \left\langle \left[ {\begin{array}{cc}
1+j \\
-1-j \\
1+j \\
-1+j \\
\end{array} } \right]\right\rangle
=\left\langle \left[ {\begin{array}{cc}
-1-j \\
1+j \\
-1-j \\
1-j \\
\end{array} } \right]\right\rangle
=\left\langle \left[ {\begin{array}{cc}
-1+j \\
1-j \\
-1+j \\
-1-j \\
\end{array} } \right]\right\rangle 
=\left\langle \left[ {\begin{array}{cc}
1-j \\
-1+j \\
1-j \\
1+j \\
\end{array} } \right]\right\rangle 
=\left\langle \left[ {\begin{array}{cc}
2j \\
-2j \\
2j \\
-2 \\
\end{array} } \right]\right\rangle
=\left\langle \left[ {\begin{array}{cc}
-2j \\
2j \\
-2j \\
2 \\
\end{array} } \right]\right\rangle
=\left\langle \left[ {\begin{array}{cc}
2 \\
-2 \\
2 \\
2j \\
\end{array} } \right]\right\rangle 
=\left\langle \left[ {\begin{array}{cc}
-2 \\
2 \\
-2 \\
-2j \\
\end{array} } \right]\right\rangle $\\ 
\vspace{0.15cm} %%%%%%%%%%%%%%%%%%%%%%%%%%%%%%%%%%
$ 26. \left\langle \left[ {\begin{array}{cc}
1+j \\
1-j \\
1+j \\
-1-j \\
\end{array} } \right]\right\rangle
=\left\langle \left[ {\begin{array}{cc}
-1-j \\
-1+j \\
-1-j \\
1+j \\
\end{array} } \right]\right\rangle
=\left\langle \left[ {\begin{array}{cc}
-1+j \\
1+j \\
-1+j \\
1-j \\
\end{array} } \right]\right\rangle 
=\left\langle \left[ {\begin{array}{cc}
1-j \\
-1-j \\
1-j \\
-1+j \\
\end{array} } \right]\right\rangle 
=\left\langle \left[ {\begin{array}{cc}
2j \\
2 \\
2j \\
-2j \\
\end{array} } \right]\right\rangle
=\left\langle \left[ {\begin{array}{cc}
-2j \\
-2 \\
-2j \\
2j \\
\end{array} } \right]\right\rangle
=\left\langle \left[ {\begin{array}{cc}
2 \\
-2j \\
2 \\
-2 \\
\end{array} } \right]\right\rangle 
=\left\langle \left[ {\begin{array}{cc}
-2 \\
2j \\
-2 \\
2 \\
\end{array} } \right]\right\rangle $\\ 
\vspace{0.15cm} %%%%%%%%%%%%%%%%%%%%%%%%%%%%%%%%%%
$ 27. \left\langle \left[ {\begin{array}{cc}
1+j \\
1-j \\
1+j \\
1-j \\
\end{array} } \right]\right\rangle
=\left\langle \left[ {\begin{array}{cc}
-1-j \\
-1+j \\
-1-j \\
-1+j \\
\end{array} } \right]\right\rangle
=\left\langle \left[ {\begin{array}{cc}
-1+j \\
1+j \\
-1+j \\
1+j \\
\end{array} } \right]\right\rangle 
=\left\langle \left[ {\begin{array}{cc}
1-j \\
-1-j \\
1-j \\
-1-j \\
\end{array} } \right]\right\rangle 
=\left\langle \left[ {\begin{array}{cc}
2j \\
2 \\
2j \\
2 \\
\end{array} } \right]\right\rangle
=\left\langle \left[ {\begin{array}{cc}
-2j \\
-2 \\
-2j \\
-2 \\
\end{array} } \right]\right\rangle
=\left\langle \left[ {\begin{array}{cc}
2 \\
-2j \\
2 \\
-2j \\
\end{array} } \right]\right\rangle 
=\left\langle \left[ {\begin{array}{cc}
-2 \\
2j \\
-2 \\
2j \\
\end{array} } \right]\right\rangle $\\ 
\vspace{0.15cm} %%%%%%%%%%%%%%%%%%%%%%%%%%%%%%%%%%
$ 28. \left\langle \left[ {\begin{array}{cc}
1+j \\
1-j \\
1+j \\
-1+j \\
\end{array} } \right]\right\rangle
=\left\langle \left[ {\begin{array}{cc}
-1-j \\
-1+j \\
-1-j \\
1-j \\
\end{array} } \right]\right\rangle
=\left\langle \left[ {\begin{array}{cc}
-1+j \\
1+j \\
-1+j \\
-1-j \\
\end{array} } \right]\right\rangle 
=\left\langle \left[ {\begin{array}{cc}
1-j \\
-1-j \\
1-j \\
1+j \\
\end{array} } \right]\right\rangle 
=\left\langle \left[ {\begin{array}{cc}
2j \\
2 \\
2j \\
-2 \\
\end{array} } \right]\right\rangle
=\left\langle \left[ {\begin{array}{cc}
-2j \\
-2 \\
-2j \\
2 \\
\end{array} } \right]\right\rangle
=\left\langle \left[ {\begin{array}{cc}
2 \\
-2j \\
2 \\
2j \\
\end{array} } \right]\right\rangle 
=\left\langle \left[ {\begin{array}{cc}
-2 \\
2j \\
-2 \\
-2j \\
\end{array} } \right]\right\rangle $\\ 
\vspace{0.15cm} %%%%%%%%%%%%%%%%%%%%%%%%%%%%%%%%%%
$ 29. \left\langle \left[ {\begin{array}{cc}
1+j \\
-1+j \\
1+j \\
-1-j \\
\end{array} } \right]\right\rangle
=\left\langle \left[ {\begin{array}{cc}
-1-j \\
1-j \\
-1-j \\
1+j \\
\end{array} } \right]\right\rangle
=\left\langle \left[ {\begin{array}{cc}
-1+j \\
-1-j \\
-1+j \\
1-j \\
\end{array} } \right]\right\rangle 
=\left\langle \left[ {\begin{array}{cc}
1-j \\
1+j \\
1-j \\
-1+j \\
\end{array} } \right]\right\rangle 
=\left\langle \left[ {\begin{array}{cc}
2j \\
-2 \\
2j \\
-2j \\
\end{array} } \right]\right\rangle
=\left\langle \left[ {\begin{array}{cc}
-2j \\
2 \\
-2j \\
2j \\
\end{array} } \right]\right\rangle
=\left\langle \left[ {\begin{array}{cc}
2 \\
2j \\
2 \\
-2 \\
\end{array} } \right]\right\rangle 
=\left\langle \left[ {\begin{array}{cc}
-2 \\
-2j \\
-2 \\
2 \\
\end{array} } \right]\right\rangle $\\ 
\vspace{0.15cm} %%%%%%%%%%%%%%%%%%%%%%%%%%%%%%%%%%
$ 30. \left\langle \left[ {\begin{array}{cc}
1+j \\
-1+j \\
1+j \\
1-j \\
\end{array} } \right]\right\rangle
=\left\langle \left[ {\begin{array}{cc}
-1-j \\
1-j \\
-1-j \\
-1+j \\
\end{array} } \right]\right\rangle
=\left\langle \left[ {\begin{array}{cc}
-1+j \\
-1-j \\
-1+j \\
1+j \\
\end{array} } \right]\right\rangle 
=\left\langle \left[ {\begin{array}{cc}
1-j \\
1+j \\
1-j \\
-1-j \\
\end{array} } \right]\right\rangle 
=\left\langle \left[ {\begin{array}{cc}
2j \\
-2 \\
2j \\
2 \\
\end{array} } \right]\right\rangle
=\left\langle \left[ {\begin{array}{cc}
-2j \\
2 \\
-2j \\
-2 \\
\end{array} } \right]\right\rangle
=\left\langle \left[ {\begin{array}{cc}
2 \\
2j \\
2 \\
-2j \\
\end{array} } \right]\right\rangle 
=\left\langle \left[ {\begin{array}{cc}
-2 \\
-2j \\
-2 \\
2j \\
\end{array} } \right]\right\rangle $\\ 
\vspace{0.15cm} %%%%%%%%%%%%%%%%%%%%%%%%%%%%%%%%%%
$ 31. \left\langle \left[ {\begin{array}{cc}
1+j \\
-1+j \\
1+j \\
-1+j \\
\end{array} } \right]\right\rangle
=\left\langle \left[ {\begin{array}{cc}
-1-j \\
1-j \\
-1-j \\
1-j \\
\end{array} } \right]\right\rangle
=\left\langle \left[ {\begin{array}{cc}
-1+j \\
-1-j \\
-1+j \\
-1-j \\
\end{array} } \right]\right\rangle 
=\left\langle \left[ {\begin{array}{cc}
1-j \\
1+j \\
1-j \\
1+j \\
\end{array} } \right]\right\rangle 
=\left\langle \left[ {\begin{array}{cc}
2j \\
-2 \\
2j \\
-2 \\
\end{array} } \right]\right\rangle
=\left\langle \left[ {\begin{array}{cc}
-2j \\
2 \\
-2j \\
2 \\
\end{array} } \right]\right\rangle
=\left\langle \left[ {\begin{array}{cc}
2 \\
2j \\
2 \\
2j \\
\end{array} } \right]\right\rangle 
=\left\langle \left[ {\begin{array}{cc}
-2 \\
-2j \\
-2 \\
-2j \\
\end{array} } \right]\right\rangle $\\ 
\vspace{0.15cm} %%%%%%%%%%%%%%%%%%%%%%%%%%%%%%%%%%
$ 32. \left\langle \left[ {\begin{array}{cc}
1+j \\
-1-j \\
-1-j \\
1+j \\
\end{array} } \right]\right\rangle
=\left\langle \left[ {\begin{array}{cc}
-1-j \\
1+j \\
1+j \\
-1-j \\
\end{array} } \right]\right\rangle
=\left\langle \left[ {\begin{array}{cc}
-1+j \\
1-j \\
1-j \\
-1+j \\
\end{array} } \right]\right\rangle 
=\left\langle \left[ {\begin{array}{cc}
1-j \\
-1+j \\
-1+j \\
1-j \\
\end{array} } \right]\right\rangle 
=\left\langle \left[ {\begin{array}{cc}
2j \\
-2j \\
-2j \\
2j \\
\end{array} } \right]\right\rangle
=\left\langle \left[ {\begin{array}{cc}
-2j \\
2j \\
2j \\
-2j \\
\end{array} } \right]\right\rangle
=\left\langle \left[ {\begin{array}{cc}
2 \\
-2 \\
-2 \\
2 \\
\end{array} } \right]\right\rangle 
=\left\langle \left[ {\begin{array}{cc}
-2 \\
2 \\
2 \\
-2 \\
\end{array} } \right]\right\rangle $\\ 
\vspace{0.15cm} %%%%%%%%%%%%%%%%%%%%%%%%%%%%%%%%%%
$ 33. \left\langle \left[ {\begin{array}{cc}
1+j \\
-1-j \\
1-j \\
1+j \\
\end{array} } \right]\right\rangle
=\left\langle \left[ {\begin{array}{cc}
-1-j \\
1+j \\
-1+j \\
-1-j \\
\end{array} } \right]\right\rangle
=\left\langle \left[ {\begin{array}{cc}
-1+j \\
1-j \\
1+j \\
-1+j \\
\end{array} } \right]\right\rangle 
=\left\langle \left[ {\begin{array}{cc}
1-j \\
-1+j \\
-1-j \\
1-j \\
\end{array} } \right]\right\rangle 
=\left\langle \left[ {\begin{array}{cc}
2j \\
-2j \\
2 \\
2j \\
\end{array} } \right]\right\rangle
=\left\langle \left[ {\begin{array}{cc}
-2j \\
2j \\
-2 \\
-2j \\
\end{array} } \right]\right\rangle
=\left\langle \left[ {\begin{array}{cc}
2 \\
-2 \\
-2j \\
2 \\
\end{array} } \right]\right\rangle 
=\left\langle \left[ {\begin{array}{cc}
-2 \\
2 \\
2j \\
-2 \\
\end{array} } \right]\right\rangle $\\ 
\vspace{0.15cm} %%%%%%%%%%%%%%%%%%%%%%%%%%%%%%%%%%
$ 34. \left\langle \left[ {\begin{array}{cc}
1+j \\
-1-j \\
-1+j \\
1+j \\
\end{array} } \right]\right\rangle
=\left\langle \left[ {\begin{array}{cc}
-1-j \\
1+j \\
1-j \\
-1-j \\
\end{array} } \right]\right\rangle
=\left\langle \left[ {\begin{array}{cc}
-1+j \\
1-j \\
-1-j \\
-1+j \\
\end{array} } \right]\right\rangle 
=\left\langle \left[ {\begin{array}{cc}
1-j \\
-1+j \\
1+j \\
1-j \\
\end{array} } \right]\right\rangle 
=\left\langle \left[ {\begin{array}{cc}
2j \\
-2j \\
-2 \\
2j \\
\end{array} } \right]\right\rangle
=\left\langle \left[ {\begin{array}{cc}
-2j \\
2j \\
2 \\
-2j \\
\end{array} } \right]\right\rangle
=\left\langle \left[ {\begin{array}{cc}
2 \\
-2 \\
2j \\
2 \\
\end{array} } \right]\right\rangle 
=\left\langle \left[ {\begin{array}{cc}
-2 \\
2 \\
-2j \\
-2 \\
\end{array} } \right]\right\rangle $\\ 
\vspace{0.15cm} %%%%%%%%%%%%%%%%%%%%%%%%%%%%%%%%%%
$ 35. \left\langle \left[ {\begin{array}{cc}
1+j \\
1-j \\
-1-j \\
1+j \\
\end{array} } \right]\right\rangle
=\left\langle \left[ {\begin{array}{cc}
-1-j \\
-1+j \\
1+j \\
-1-j \\
\end{array} } \right]\right\rangle
=\left\langle \left[ {\begin{array}{cc}
-1+j \\
1+j \\
1-j \\
-1+j \\
\end{array} } \right]\right\rangle 
=\left\langle \left[ {\begin{array}{cc}
1-j \\
-1-j \\
-1+j \\
1-j \\
\end{array} } \right]\right\rangle 
=\left\langle \left[ {\begin{array}{cc}
2j \\
2 \\
-2j \\
2j \\
\end{array} } \right]\right\rangle
=\left\langle \left[ {\begin{array}{cc}
-2j \\
-2 \\
2j \\
-2j \\
\end{array} } \right]\right\rangle
=\left\langle \left[ {\begin{array}{cc}
2 \\
-2j \\
-2 \\
2 \\
\end{array} } \right]\right\rangle 
=\left\langle \left[ {\begin{array}{cc}
-2 \\
2j \\
2 \\
-2 \\
\end{array} } \right]\right\rangle $\\ 
\vspace{0.15cm} %%%%%%%%%%%%%%%%%%%%%%%%%%%%%%%%%%
$ 36. \left\langle \left[ {\begin{array}{cc}
1+j \\
1-j \\
1-j \\
1+j \\
\end{array} } \right]\right\rangle
=\left\langle \left[ {\begin{array}{cc}
-1-j \\
-1+j \\
-1+j \\
-1-j \\
\end{array} } \right]\right\rangle
=\left\langle \left[ {\begin{array}{cc}
-1+j \\
1+j \\
1+j \\
-1+j \\
\end{array} } \right]\right\rangle 
=\left\langle \left[ {\begin{array}{cc}
1-j \\
-1-j \\
-1-j \\
1-j \\
\end{array} } \right]\right\rangle 
=\left\langle \left[ {\begin{array}{cc}
2j \\
2 \\
2 \\
2j \\
\end{array} } \right]\right\rangle
=\left\langle \left[ {\begin{array}{cc}
-2j \\
-2 \\
-2 \\
-2j \\
\end{array} } \right]\right\rangle
=\left\langle \left[ {\begin{array}{cc}
2 \\
-2j \\
-2j \\
2 \\
\end{array} } \right]\right\rangle 
=\left\langle \left[ {\begin{array}{cc}
-2 \\
2j \\
2j \\
-2 \\
\end{array} } \right]\right\rangle $\\ 
\vspace{0.15cm} %%%%%%%%%%%%%%%%%%%%%%%%%%%%%%%%%%
$ 37. \left\langle \left[ {\begin{array}{cc}
1+j \\
1-j \\
-1+j \\
1+j \\
\end{array} } \right]\right\rangle
=\left\langle \left[ {\begin{array}{cc}
-1-j \\
-1+j \\
1-j \\
-1-j \\
\end{array} } \right]\right\rangle
=\left\langle \left[ {\begin{array}{cc}
-1+j \\
1+j \\
-1-j \\
-1+j \\
\end{array} } \right]\right\rangle 
=\left\langle \left[ {\begin{array}{cc}
1-j \\
-1-j \\
1+j \\
1-j \\
\end{array} } \right]\right\rangle 
=\left\langle \left[ {\begin{array}{cc}
2j \\
2 \\
-2 \\
2j \\
\end{array} } \right]\right\rangle
=\left\langle \left[ {\begin{array}{cc}
-2j \\
-2 \\
2 \\
-2j \\
\end{array} } \right]\right\rangle
=\left\langle \left[ {\begin{array}{cc}
2 \\
-2j \\
2j \\
2 \\
\end{array} } \right]\right\rangle 
=\left\langle \left[ {\begin{array}{cc}
-2 \\
2j \\
-2j \\
-2 \\
\end{array} } \right]\right\rangle $\\ 
\vspace{0.15cm} %%%%%%%%%%%%%%%%%%%%%%%%%%%%%%%%%%
$ 38. \left\langle \left[ {\begin{array}{cc}
1+j \\
-1+j \\
-1-j \\
1+j \\
\end{array} } \right]\right\rangle
=\left\langle \left[ {\begin{array}{cc}
-1-j \\
1-j \\
1+j \\
-1-j \\
\end{array} } \right]\right\rangle
=\left\langle \left[ {\begin{array}{cc}
-1+j \\
-1-j \\
1-j \\
-1+j \\
\end{array} } \right]\right\rangle 
=\left\langle \left[ {\begin{array}{cc}
1-j \\
1+j \\
-1+j \\
1-j \\
\end{array} } \right]\right\rangle 
=\left\langle \left[ {\begin{array}{cc}
2j \\
-2 \\
-2j \\
2j \\
\end{array} } \right]\right\rangle
=\left\langle \left[ {\begin{array}{cc}
-2j \\
2 \\
2j \\
-2j \\
\end{array} } \right]\right\rangle
=\left\langle \left[ {\begin{array}{cc}
2 \\
2j \\
-2 \\
2 \\
\end{array} } \right]\right\rangle 
=\left\langle \left[ {\begin{array}{cc}
-2 \\
-2j \\
2 \\
-2 \\
\end{array} } \right]\right\rangle $\\ 
\vspace{0.15cm} %%%%%%%%%%%%%%%%%%%%%%%%%%%%%%%%%%
$ 39. \left\langle \left[ {\begin{array}{cc}
1+j \\
-1+j \\
1-j \\
1+j \\
\end{array} } \right]\right\rangle
=\left\langle \left[ {\begin{array}{cc}
-1-j \\
1-j \\
-1+j \\
-1-j \\
\end{array} } \right]\right\rangle
=\left\langle \left[ {\begin{array}{cc}
-1+j \\
-1-j \\
1+j \\
-1+j \\
\end{array} } \right]\right\rangle 
=\left\langle \left[ {\begin{array}{cc}
1-j \\
1+j \\
-1-j \\
1-j \\
\end{array} } \right]\right\rangle 
=\left\langle \left[ {\begin{array}{cc}
2j \\
-2 \\
2 \\
2j \\
\end{array} } \right]\right\rangle
=\left\langle \left[ {\begin{array}{cc}
-2j \\
2 \\
-2 \\
-2j \\
\end{array} } \right]\right\rangle
=\left\langle \left[ {\begin{array}{cc}
2 \\
2j \\
-2j \\
2 \\
\end{array} } \right]\right\rangle 
=\left\langle \left[ {\begin{array}{cc}
-2 \\
-2j \\
2j \\
-2 \\
\end{array} } \right]\right\rangle $\\ 
\vspace{0.15cm} %%%%%%%%%%%%%%%%%%%%%%%%%%%%%%%%%%
$ 40. \left\langle \left[ {\begin{array}{cc}
1+j \\
-1+j \\
-1+j \\
1+j \\
\end{array} } \right]\right\rangle
=\left\langle \left[ {\begin{array}{cc}
-1-j \\
1-j \\
1-j \\
-1-j \\
\end{array} } \right]\right\rangle
=\left\langle \left[ {\begin{array}{cc}
-1+j \\
-1-j \\
-1-j \\
-1+j \\
\end{array} } \right]\right\rangle 
=\left\langle \left[ {\begin{array}{cc}
1-j \\
1+j \\
1+j \\
1-j \\
\end{array} } \right]\right\rangle 
=\left\langle \left[ {\begin{array}{cc}
2j \\
-2 \\
-2 \\
2j \\
\end{array} } \right]\right\rangle
=\left\langle \left[ {\begin{array}{cc}
-2j \\
2 \\
2 \\
-2j \\
\end{array} } \right]\right\rangle
=\left\langle \left[ {\begin{array}{cc}
2 \\
2j \\
2j \\
2 \\
\end{array} } \right]\right\rangle 
=\left\langle \left[ {\begin{array}{cc}
-2 \\
-2j \\
-2j \\
-2 \\
\end{array} } \right]\right\rangle $\\ 
\vspace{0.15cm} %%%%%%%%%%%%%%%%%%%%%%%%%%%%%%%%%%
$ 41. \left\langle \left[ {\begin{array}{cc}
-1-j \\
1+j \\
1-j \\
1+j \\
\end{array} } \right]\right\rangle
=\left\langle \left[ {\begin{array}{cc}
1+j \\
-1-j \\
-1+j \\
-1-j \\
\end{array} } \right]\right\rangle
=\left\langle \left[ {\begin{array}{cc}
1-j \\
-1+j \\
1+j \\
-1+j \\
\end{array} } \right]\right\rangle 
=\left\langle \left[ {\begin{array}{cc}
-1+j \\
1-j \\
-1-j \\
1-j \\
\end{array} } \right]\right\rangle 
=\left\langle \left[ {\begin{array}{cc}
-2j \\
2j \\
2 \\
2j \\
\end{array} } \right]\right\rangle
=\left\langle \left[ {\begin{array}{cc}
2j \\
-2j \\
-2 \\
-2j \\
\end{array} } \right]\right\rangle
=\left\langle \left[ {\begin{array}{cc}
-2 \\
2 \\
-2j \\
2 \\
\end{array} } \right]\right\rangle 
=\left\langle \left[ {\begin{array}{cc}
2 \\
-2 \\
2j \\
-2 \\
\end{array} } \right]\right\rangle $\\ 
\vspace{0.15cm} %%%%%%%%%%%%%%%%%%%%%%%%%%%%%%%%%%
$ 42. \left\langle \left[ {\begin{array}{cc}
-1-j \\
1+j \\
-1+j \\
1+j \\
\end{array} } \right]\right\rangle
=\left\langle \left[ {\begin{array}{cc}
1+j \\
-1-j \\
1-j \\
-1-j \\
\end{array} } \right]\right\rangle
=\left\langle \left[ {\begin{array}{cc}
1-j \\
-1+j \\
-1-j \\
-1+j \\
\end{array} } \right]\right\rangle 
=\left\langle \left[ {\begin{array}{cc}
-1+j \\
1-j \\
1+j \\
1-j \\
\end{array} } \right]\right\rangle 
=\left\langle \left[ {\begin{array}{cc}
-2j \\
2j \\
-2 \\
2j \\
\end{array} } \right]\right\rangle
=\left\langle \left[ {\begin{array}{cc}
2j \\
-2j \\
2 \\
-2j \\
\end{array} } \right]\right\rangle
=\left\langle \left[ {\begin{array}{cc}
-2 \\
2 \\
2j \\
2 \\
\end{array} } \right]\right\rangle 
=\left\langle \left[ {\begin{array}{cc}
2 \\
-2 \\
-2j \\
-2 \\
\end{array} } \right]\right\rangle $\\ 
\vspace{0.15cm} %%%%%%%%%%%%%%%%%%%%%%%%%%%%%%%%%%
$ 43. \left\langle \left[ {\begin{array}{cc}
1-j \\
1+j \\
-1-j \\
1+j \\
\end{array} } \right]\right\rangle
=\left\langle \left[ {\begin{array}{cc}
-1+j \\
-1-j \\
1+j \\
-1-j \\
\end{array} } \right]\right\rangle
=\left\langle \left[ {\begin{array}{cc}
1+j \\
-1+j \\
1-j \\
-1+j \\
\end{array} } \right]\right\rangle 
=\left\langle \left[ {\begin{array}{cc}
-1-j \\
1-j \\
-1+j \\
1-j \\
\end{array} } \right]\right\rangle 
=\left\langle \left[ {\begin{array}{cc}
2 \\
2j \\
-2j \\
2j \\
\end{array} } \right]\right\rangle
=\left\langle \left[ {\begin{array}{cc}
-2 \\
-2j \\
2j \\
-2j \\
\end{array} } \right]\right\rangle
=\left\langle \left[ {\begin{array}{cc}
-2j \\
2 \\
-2 \\
2 \\
\end{array} } \right]\right\rangle 
=\left\langle \left[ {\begin{array}{cc}
2j \\
-2 \\
2 \\
-2 \\
\end{array} } \right]\right\rangle $\\ 
\vspace{0.15cm} %%%%%%%%%%%%%%%%%%%%%%%%%%%%%%%%%%
$ 44. \left\langle \left[ {\begin{array}{cc}
1-j \\
1+j \\
-1+j \\
1+j \\
\end{array} } \right]\right\rangle
=\left\langle \left[ {\begin{array}{cc}
-1+j \\
-1-j \\
1-j \\
-1-j \\
\end{array} } \right]\right\rangle
=\left\langle \left[ {\begin{array}{cc}
1+j \\
-1+j \\
-1-j \\
-1+j \\
\end{array} } \right]\right\rangle 
=\left\langle \left[ {\begin{array}{cc}
-1-j \\
1-j \\
1+j \\
1-j \\
\end{array} } \right]\right\rangle 
=\left\langle \left[ {\begin{array}{cc}
2 \\
2j \\
-2 \\
2j \\
\end{array} } \right]\right\rangle
=\left\langle \left[ {\begin{array}{cc}
-2 \\
-2j \\
2 \\
-2j \\
\end{array} } \right]\right\rangle
=\left\langle \left[ {\begin{array}{cc}
-2j \\
2 \\
2j \\
2 \\
\end{array} } \right]\right\rangle 
=\left\langle \left[ {\begin{array}{cc}
2j \\
-2 \\
-2j \\
-2 \\
\end{array} } \right]\right\rangle $\\ 
$ 45. \left\langle \left[ {\begin{array}{cc}
-1+j \\
1+j \\
-1-j \\
1+j \\
\end{array} } \right]\right\rangle
=\left\langle \left[ {\begin{array}{cc}
1-j \\
-1-j \\
1+j \\
-1-j \\
\end{array} } \right]\right\rangle
=\left\langle \left[ {\begin{array}{cc}
-1-j \\
-1+j \\
1-j \\
-1+j \\
\end{array} } \right]\right\rangle 
=\left\langle \left[ {\begin{array}{cc}
1+j \\
1-j \\
-1+j \\
1-j \\
\end{array} } \right]\right\rangle 
=\left\langle \left[ {\begin{array}{cc}
-2 \\
2j \\
-2j \\
2j \\
\end{array} } \right]\right\rangle
=\left\langle \left[ {\begin{array}{cc}
2 \\
-2j \\
2j \\
-2j \\
\end{array} } \right]\right\rangle
=\left\langle \left[ {\begin{array}{cc}
2j \\
2 \\
-2 \\
2 \\
\end{array} } \right]\right\rangle 
=\left\langle \left[ {\begin{array}{cc}
-2j \\
-2 \\
2 \\
-2 \\
\end{array} } \right]\right\rangle $\\ 
\vspace{0.15cm} %%%%%%%%%%%%%%%%%%%%%%%%%%%%%%%%%%
%\label{fig:sfssubcase3}
\end{figure*}

\begin{figure*}
%\centering
\tiny
\vspace{-3cm} %%%%%%%%%%%%%%%%%%%%%%%%%%%%%%%%%%
$ 46. \left\langle \left[ {\begin{array}{cc}
-1+j \\
1+j \\
1-j \\
1+j \\
\end{array} } \right]\right\rangle
=\left\langle \left[ {\begin{array}{cc}
1-j \\
-1-j \\
-1+j \\
-1-j \\
\end{array} } \right]\right\rangle
=\left\langle \left[ {\begin{array}{cc}
-1-j \\
-1+j \\
1+j \\
-1+j \\
\end{array} } \right]\right\rangle 
=\left\langle \left[ {\begin{array}{cc}
1+j \\
1-j \\
-1-j \\
1-j \\
\end{array} } \right]\right\rangle 
=\left\langle \left[ {\begin{array}{cc}
-2 \\
2j \\
2 \\
2j \\
\end{array} } \right]\right\rangle
=\left\langle \left[ {\begin{array}{cc}
2 \\
-2j \\
-2 \\
-2j \\
\end{array} } \right]\right\rangle
=\left\langle \left[ {\begin{array}{cc}
2j \\
2 \\
-2j \\
2 \\
\end{array} } \right]\right\rangle 
=\left\langle \left[ {\begin{array}{cc}
-2j \\
-2 \\
2j \\
-2 \\
\end{array} } \right]\right\rangle $\\ 
\vspace{0.15cm} %%%%%%%%%%%%%%%%%%%%%%%%%%%%%%%%%%
$ 47. \left\langle \left[ {\begin{array}{cc}
-1-j \\
1+j \\
1+j \\
1-j \\
\end{array} } \right]\right\rangle
=\left\langle \left[ {\begin{array}{cc}
1+j \\
-1-j \\
-1-j \\
-1+j \\
\end{array} } \right]\right\rangle
=\left\langle \left[ {\begin{array}{cc}
1-j \\
-1+j \\
-1+j \\
1+j \\
\end{array} } \right]\right\rangle 
=\left\langle \left[ {\begin{array}{cc}
-1+j \\
1-j \\
1-j \\
-1-j \\
\end{array} } \right]\right\rangle 
=\left\langle \left[ {\begin{array}{cc}
-2j \\
2j \\
2j \\
2 \\
\end{array} } \right]\right\rangle
=\left\langle \left[ {\begin{array}{cc}
2j \\
-2j \\
-2j \\
-2 \\
\end{array} } \right]\right\rangle
=\left\langle \left[ {\begin{array}{cc}
-2 \\
2 \\
2 \\
-2j \\
\end{array} } \right]\right\rangle 
=\left\langle \left[ {\begin{array}{cc}
2 \\
-2 \\
-2 \\
2j \\
\end{array} } \right]\right\rangle $\\ 
\vspace{0.15cm} %%%%%%%%%%%%%%%%%%%%%%%%%%%%%%%%%%
$ 48. \left\langle \left[ {\begin{array}{cc}
-1-j \\
1+j \\
1+j \\
-1+j \\
\end{array} } \right]\right\rangle
=\left\langle \left[ {\begin{array}{cc}
1+j \\
-1-j \\
-1-j \\
1-j \\
\end{array} } \right]\right\rangle
=\left\langle \left[ {\begin{array}{cc}
1-j \\
-1+j \\
-1+j \\
-1-j \\
\end{array} } \right]\right\rangle 
=\left\langle \left[ {\begin{array}{cc}
-1+j \\
1-j \\
1-j \\
1+j \\
\end{array} } \right]\right\rangle 
=\left\langle \left[ {\begin{array}{cc}
-2j \\
2j \\
2j \\
-2 \\
\end{array} } \right]\right\rangle
=\left\langle \left[ {\begin{array}{cc}
2j \\
-2j \\
-2j \\
2 \\
\end{array} } \right]\right\rangle
=\left\langle \left[ {\begin{array}{cc}
-2 \\
2 \\
2 \\
2j \\
\end{array} } \right]\right\rangle 
=\left\langle \left[ {\begin{array}{cc}
2 \\
-2 \\
-2 \\
-2j \\
\end{array} } \right]\right\rangle $\\ 
$ 49. \left\langle \left[ {\begin{array}{cc}
1-j \\
1+j \\
1+j \\
-1-j \\
\end{array} } \right]\right\rangle
=\left\langle \left[ {\begin{array}{cc}
-1+j \\
-1-j \\
-1-j \\
1+j \\
\end{array} } \right]\right\rangle
=\left\langle \left[ {\begin{array}{cc}
1+j \\
-1+j \\
-1+j \\
1-j \\
\end{array} } \right]\right\rangle 
=\left\langle \left[ {\begin{array}{cc}
-1-j \\
1-j \\
1-j \\
-1+j \\
\end{array} } \right]\right\rangle 
=\left\langle \left[ {\begin{array}{cc}
2 \\
2j \\
2j \\
-2j \\
\end{array} } \right]\right\rangle
=\left\langle \left[ {\begin{array}{cc}
-2 \\
-2j \\
-2j \\
2j \\
\end{array} } \right]\right\rangle
=\left\langle \left[ {\begin{array}{cc}
-2j \\
2 \\
2 \\
-2 \\
\end{array} } \right]\right\rangle 
=\left\langle \left[ {\begin{array}{cc}
2j \\
-2 \\
-2 \\
2 \\
\end{array} } \right]\right\rangle $\\ 
\vspace{0.15cm} %%%%%%%%%%%%%%%%%%%%%%%%%%%%%%%%%%
$ 50. \left\langle \left[ {\begin{array}{cc}
1-j \\
1+j \\
1+j \\
-1+j \\
\end{array} } \right]\right\rangle
=\left\langle \left[ {\begin{array}{cc}
-1+j \\
-1-j \\
-1-j \\
1-j \\
\end{array} } \right]\right\rangle
=\left\langle \left[ {\begin{array}{cc}
1+j \\
-1+j \\
-1+j \\
-1-j \\
\end{array} } \right]\right\rangle 
=\left\langle \left[ {\begin{array}{cc}
-1-j \\
1-j \\
1-j \\
1+j \\
\end{array} } \right]\right\rangle 
=\left\langle \left[ {\begin{array}{cc}
2 \\
2j \\
2j \\
-2 \\
\end{array} } \right]\right\rangle
=\left\langle \left[ {\begin{array}{cc}
-2 \\
-2j \\
-2j \\
2 \\
\end{array} } \right]\right\rangle
=\left\langle \left[ {\begin{array}{cc}
-2j \\
2 \\
2 \\
2j \\
\end{array} } \right]\right\rangle 
=\left\langle \left[ {\begin{array}{cc}
2j \\
-2 \\
-2 \\
-2j \\
\end{array} } \right]\right\rangle $\\ 
\vspace{0.15cm} %%%%%%%%%%%%%%%%%%%%%%%%%%%%%%%%%%
$ 51. \left\langle \left[ {\begin{array}{cc}
-1+j \\
1+j \\
1+j \\
-1-j \\
\end{array} } \right]\right\rangle
=\left\langle \left[ {\begin{array}{cc}
1-j \\
-1-j \\
-1-j \\
1+j \\
\end{array} } \right]\right\rangle
=\left\langle \left[ {\begin{array}{cc}
-1-j \\
-1+j \\
-1+j \\
1-j \\
\end{array} } \right]\right\rangle 
=\left\langle \left[ {\begin{array}{cc}
1+j \\
1-j \\
1-j \\
-1+j \\
\end{array} } \right]\right\rangle 
=\left\langle \left[ {\begin{array}{cc}
-2 \\
2j \\
2j \\
-2j \\
\end{array} } \right]\right\rangle
=\left\langle \left[ {\begin{array}{cc}
2 \\
-2j \\
-2j \\
2j \\
\end{array} } \right]\right\rangle
=\left\langle \left[ {\begin{array}{cc}
2j \\
2 \\
2 \\
-2 \\
\end{array} } \right]\right\rangle 
=\left\langle \left[ {\begin{array}{cc}
-2j \\
-2 \\
-2 \\
2 \\
\end{array} } \right]\right\rangle $\\ 
\vspace{0.15cm} %%%%%%%%%%%%%%%%%%%%%%%%%%%%%%%%%%
$ 52. \left\langle \left[ {\begin{array}{cc}
-1+j \\
1+j \\
1+j \\
1-j \\
\end{array} } \right]\right\rangle
=\left\langle \left[ {\begin{array}{cc}
1-j \\
-1-j \\
-1-j \\
-1+j \\
\end{array} } \right]\right\rangle
=\left\langle \left[ {\begin{array}{cc}
-1-j \\
-1+j \\
-1+j \\
1+j \\
\end{array} } \right]\right\rangle 
=\left\langle \left[ {\begin{array}{cc}
1+j \\
1-j \\
1-j \\
-1-j \\
\end{array} } \right]\right\rangle 
=\left\langle \left[ {\begin{array}{cc}
-2 \\
2j \\
2j \\
2 \\
\end{array} } \right]\right\rangle
=\left\langle \left[ {\begin{array}{cc}
2 \\
-2j \\
-2j \\
-2 \\
\end{array} } \right]\right\rangle
=\left\langle \left[ {\begin{array}{cc}
2j \\
2 \\
2 \\
-2j \\
\end{array} } \right]\right\rangle 
=\left\langle \left[ {\begin{array}{cc}
-2j \\
-2 \\
-2 \\
2j \\
\end{array} } \right]\right\rangle $\\ 
\vspace{0.15cm} %%%%%%%%%%%%%%%%%%%%%%%%%%%%%%%%%%
$ 53. \left\langle \left[ {\begin{array}{cc}
-1-j \\
1-j \\
1+j \\
1+j \\
\end{array} } \right]\right\rangle
=\left\langle \left[ {\begin{array}{cc}
1+j \\
-1+j \\
-1-j \\
-1-j \\
\end{array} } \right]\right\rangle
=\left\langle \left[ {\begin{array}{cc}
1-j \\
1+j \\
-1+j \\
-1+j \\
\end{array} } \right]\right\rangle 
=\left\langle \left[ {\begin{array}{cc}
-1+j \\
-1-j \\
1-j \\
1-j \\
\end{array} } \right]\right\rangle 
=\left\langle \left[ {\begin{array}{cc}
-2j \\
2 \\
2j \\
2j \\
\end{array} } \right]\right\rangle
=\left\langle \left[ {\begin{array}{cc}
2j \\
-2 \\
-2j \\
-2j \\
\end{array} } \right]\right\rangle
=\left\langle \left[ {\begin{array}{cc}
-2 \\
-2j \\
2 \\
2 \\
\end{array} } \right]\right\rangle 
=\left\langle \left[ {\begin{array}{cc}
2 \\
2j \\
-2 \\
-2 \\
\end{array} } \right]\right\rangle $\\ 
\vspace{0.15cm} %%%%%%%%%%%%%%%%%%%%%%%%%%%%%%%%%%
$ 54. \left\langle \left[ {\begin{array}{cc}
-1-j \\
-1+j \\
1+j \\
1+j \\
\end{array} } \right]\right\rangle
=\left\langle \left[ {\begin{array}{cc}
1+j \\
1-j \\
-1-j \\
-1-j \\
\end{array} } \right]\right\rangle
=\left\langle \left[ {\begin{array}{cc}
1-j \\
-1-j \\
-1+j \\
-1+j \\
\end{array} } \right]\right\rangle 
=\left\langle \left[ {\begin{array}{cc}
-1+j \\
1+j \\
1-j \\
1-j \\
\end{array} } \right]\right\rangle 
=\left\langle \left[ {\begin{array}{cc}
-2j \\
-2 \\
2j \\
2j \\
\end{array} } \right]\right\rangle
=\left\langle \left[ {\begin{array}{cc}
2j \\
2 \\
-2j \\
-2j \\
\end{array} } \right]\right\rangle
=\left\langle \left[ {\begin{array}{cc}
-2 \\
2j \\
2 \\
2 \\
\end{array} } \right]\right\rangle 
=\left\langle \left[ {\begin{array}{cc}
2 \\
-2j \\
-2 \\
-2 \\
\end{array} } \right]\right\rangle $\\ 
\vspace{0.15cm} %%%%%%%%%%%%%%%%%%%%%%%%%%%%%%%%%%
$ 55. \left\langle \left[ {\begin{array}{cc}
1-j \\
-1-j \\
1+j \\
1+j \\
\end{array} } \right]\right\rangle
=\left\langle \left[ {\begin{array}{cc}
-1+j \\
1+j \\
-1-j \\
-1-j \\
\end{array} } \right]\right\rangle
=\left\langle \left[ {\begin{array}{cc}
1+j \\
1-j \\
-1+j \\
-1+j \\
\end{array} } \right]\right\rangle 
=\left\langle \left[ {\begin{array}{cc}
-1-j \\
-1+j \\
1-j \\
1-j \\
\end{array} } \right]\right\rangle 
=\left\langle \left[ {\begin{array}{cc}
2 \\
-2j \\
2j \\
2j \\
\end{array} } \right]\right\rangle
=\left\langle \left[ {\begin{array}{cc}
-2 \\
2j \\
-2j \\
-2j \\
\end{array} } \right]\right\rangle
=\left\langle \left[ {\begin{array}{cc}
-2j \\
-2 \\
2 \\
2 \\
\end{array} } \right]\right\rangle 
=\left\langle \left[ {\begin{array}{cc}
2j \\
2 \\
-2 \\
-2 \\
\end{array} } \right]\right\rangle $\\ 
\vspace{0.15cm} %%%%%%%%%%%%%%%%%%%%%%%%%%%%%%%%%%
$ 56. \left\langle \left[ {\begin{array}{cc}
1-j \\
-1+j \\
1+j \\
1+j \\
\end{array} } \right]\right\rangle
=\left\langle \left[ {\begin{array}{cc}
-1+j \\
1-j \\
-1-j \\
-1-j \\
\end{array} } \right]\right\rangle
=\left\langle \left[ {\begin{array}{cc}
1+j \\
-1-j \\
-1+j \\
-1+j \\
\end{array} } \right]\right\rangle 
=\left\langle \left[ {\begin{array}{cc}
-1-j \\
1+j \\
1-j \\
1-j \\
\end{array} } \right]\right\rangle 
=\left\langle \left[ {\begin{array}{cc}
2 \\
-2 \\
2j \\
2j \\
\end{array} } \right]\right\rangle
=\left\langle \left[ {\begin{array}{cc}
-2 \\
2 \\
-2j \\
-2j \\
\end{array} } \right]\right\rangle
=\left\langle \left[ {\begin{array}{cc}
-2j \\
2j \\
2 \\
2 \\
\end{array} } \right]\right\rangle 
=\left\langle \left[ {\begin{array}{cc}
2j \\
-2j \\
-2 \\
-2 \\
\end{array} } \right]\right\rangle $\\ 
\vspace{0.15cm} %%%%%%%%%%%%%%%%%%%%%%%%%%%%%%%%%%
$ 57. \left\langle \left[ {\begin{array}{cc}
-1+j \\
-1-j \\
1+j \\
1+j \\
\end{array} } \right]\right\rangle
=\left\langle \left[ {\begin{array}{cc}
1-j \\
1+j \\
-1-j \\
-1-j \\
\end{array} } \right]\right\rangle
=\left\langle \left[ {\begin{array}{cc}
-1-j \\
1-j \\
-1+j \\
-1+j \\
\end{array} } \right]\right\rangle 
=\left\langle \left[ {\begin{array}{cc}
1+j \\
-1+j \\
1-j \\
1-j \\
\end{array} } \right]\right\rangle 
=\left\langle \left[ {\begin{array}{cc}
-2 \\
-2j \\
2j \\
2j \\
\end{array} } \right]\right\rangle
=\left\langle \left[ {\begin{array}{cc}
2 \\
2j \\
-2j \\
-2j \\
\end{array}%
 } \right]\right\rangle $\\
\label{fig:sfscase4}
\end{figure*}
%%%%%%%%%%%%%%%%%%%%%%%%%%%%%%%%%%%%%%%%%


\begin{thebibliography}{160}
\bibitem{PoY1}
P. Popovski and H. Yomo, ``The Anti-Packets can Increase the Achievable Throughput of a Wireless Multi-Hop Network'', Proc. IEEE ICC 2006, Istanbul, Turkey, June 2006.

\bibitem{KPT}
T. Koike-Akino, P. Popovski and V. Tarokh, ``Optimized constellation for two-way wireless relaying with physical network coding'', IEEE Journal on selected Areas in Comm., Vol.27, pp. 773--787, June 2009.

\bibitem{SVR} 
Srishti Shukla, Vijayvaradharaj T. Muralidharan and B. Sundar Rajan, "Wireless Network-Coded Three-Way Relaying using Latin Cubes", 23rd IEEE Personal Indoor Mobile Radio Communications (PIMRC)-2012, Sydney, Australia.

\bibitem{ZLL}
S. Zhang, S. C. Liew and P. P. Lam, ``Hot topic: Physical-layer Network Coding'', ACM MobiCom '06, pp. 358--365, Sept. 2006.

\bibitem{KMT}
S. J. Kim, P. Mitran and V. Tarokh, ``Performance Bounds for Bidirectional Coded Cooperation Protocols'', IEEE Trans. Inf. Theory, Vol. 54, pp. 5235--5241, Nov. 2008.

\bibitem{PoY}
P. Popovski and H. Yomo, ``Physical Network Coding in Two-Way Wireless Relay Channels'', IEEE ICC, Glasgow, Scotland, pp. 707--712, June 2007.

\bibitem{LiA}
C.H. Liu and A. Arapostathis, ``Joint network coding and superposition coding for multi-user information exchange in wireless relaying networks'', in Proc. of IEEE-Globecom, pp. 1--6, Dec. 2008. 

\bibitem{PiR}
Mylene Pischella and Didier Le Ruyet, ``Lattice based coding scheme for MIMO bi-directional relaying with three nodes'', 22nd IEEE Personal Indoor Mobile Radio Communications, Toronto, Canada, pp. 1459--1463, Sept. 2011. 

\bibitem{PaO}
Moonseo Park and Seong Keun Oh, ``An Iterative Network Code Optimization for Three-Way Relay Channels'', Vehicular Technology Conference Fall (VTC-Fall), 2009 IEEE 70th, pp. 1--5, Sept. 2009. 

\bibitem{JKPL}
Youngil Jeon, Young-Tae Kim, Moonseo Park, Inkyu Lee, ``Opportunistic Scheduling for Three-way Relay Systems with Physical Layer Network Coding'', Vehicular Technology Conference (VTC Spring) 2011, Budapest, Hungary, IEEE 73rd, pp. 1--5, May 2011. 

\bibitem{Kis}
K. Kishen, ``On Latin and Hyper-Graeco-Latin Cubes and Hyper Cubes'', Current Science, Vol. 11, pp. 98--99, 1942. 

\bibitem{NMR}
Vishnu Namboodiri, Vijayvaradharaj T Muralidharan and B. Sundar Rajan, ``Wireless Bidirectional Relaying and Latin Squares'', available online at arXiv:1110.0084v2 [cs.IT], 1 Oct. 2011.

\bibitem{ShR}
Srishti Shukla and B. Sundar Rajan, ``Wireless Network-Coded Four-Way Relaying Using Latin Hyper-Cubes'', available online at arXiv [cs.IT], Oct. 2012.
\end{thebibliography}
\end{document}